\documentclass[11pt]{article}
\usepackage[utf8]{inputenc}
\usepackage[english]{babel}

\usepackage{amsmath}
\usepackage{graphicx,psfrag,epsf}
\usepackage{enumerate}
\usepackage{natbib}
\usepackage{times}
\usepackage{bm}

\usepackage{authblk}
\author[1]{Javier Zapata\thanks{Corresponding Author}}
\author[2]{Sang-Yun Oh}
\author[3]{Alexander Petersen}
\affil[1,2]{Department of Statistics and Applied Probability \protect\\ University of California Santa Barbara}
\affil[2]{Lawrence Berkeley Lab, Berkeley, CA}
\affil[3]{Department of Statistics, Brigham Young University, Provo, Utah, U.S.A.}

\usepackage{algorithm,algorithmic}

\makeatletter
\providecommand{\@LN}[2]{}
\makeatother

\usepackage{amsmath,amsthm}

\newtheorem{theorem}{Theorem}
\newtheorem{corollary}{Corollary}
\newtheorem{lemma}{Lemma}
\newtheorem{proposition}{Proposition}

\theoremstyle{definition}
\newtheorem{definition}{Definition}
\newtheorem{assumption}{Assumption}
\newtheorem{remark}{Remark}

\newcommand{\blind}{1}

\usepackage{microtype}
\usepackage[compact]{titlesec}
\titlespacing{\section}{0pt}{2ex}{1ex}
\titlespacing{\subsection}{0pt}{1ex}{0ex}
\titlespacing{\subsubsection}{0pt}{0.5ex}{0ex}

 \usepackage{mathtools}
 \usepackage{amssymb}
 \usepackage{amsfonts}
 \usepackage{amsbsy} 
 \usepackage{color}
\usepackage{url}
\usepackage{enumerate}
\usepackage{tikz}
\usepackage{enumitem}
\usepackage{float}
\usepackage{caption}
\usepackage{subcaption} 
\usepackage{graphicx} 

\def\T{{^\mathrm{\scriptscriptstyle T} }} 
\definecolor{DarkBlue}{rgb}{0,.08,.45}
\definecolor{DarkRed}{rgb}{.7,0,.4}

\definecolor{HCP_Blue}{rgb}{0.12156, 0.47058, 0.70588}
\definecolor{HCP_LightGreen}{rgb}{0.69803, 0.87450, 0.54117}
\definecolor{HCP_Green}{rgb}{0.2, 0.62745, 0.17254}
\definecolor{HCP_Red}{rgb}{0.89019, 0.10196, 0.10980}
\definecolor{HCP_Purple}{rgb}{0.41568, 0.23921, 0.60392}
\def\HCPblue{\textcolor{HCP_Blue}}
\def\HCPlightgreen{\textcolor{HCP_LightGreen}}
\def\HCPgreen{\textcolor{HCP_Green}}
\def\HCPred{\textcolor{HCP_Red}}
\def\HCPpurple{\textcolor{HCP_Purple}}
\DeclareRobustCommand\sampleline[1]{%
  \tikz\draw[#1] (0,0) (0,\the\dimexpr\fontdimen22\textfont2\relax)
  -- (2em,\the\dimexpr\fontdimen22\textfont2\relax);%
}

\def\bco{\iffalse}
\def\cov{{\rm cov}}
\def\var{{\rm var}}

\def\diag{{\rm diag}}
\def\trace{{\rm tr}} 

\newcommand\independent{\protect\mathpalette{\protect\independenT}{\perp}}
\def\independenT#1#2{\mathrel{\rlap{$#1#2$}\mkern2mu{#1#2}}}

\def\ci{\cite}
\def\cp{\citep}

\def\Up{\Upsilon}

\def\ra{\rightarrow}
\def\mc{\mathcal}
\def\inv{^{-1}}

\def\s1n{\sum_{i=1}^n}
\def\p1n{\prod_{i=1}^n}
\def\i01{\int_0^1}
\def\zo{[0,1]}
\def\dt{\mathrm{d}t}

\def\suml{\sum_{l = 1}^\infty}
\def\R{\mathcal{R}}

\def\pr{\mathrm{pr}}

\def\Ltp{(L^2\zo)^p}
\def\Lt{L^2\zo}
\DeclarePairedDelimiterX{\ip}[2]{\langle}{\rangle}{#1, #2} 
\DeclarePairedDelimiterX{\ipp}[2]{\langle}{\rangle_p}{#1, #2} 
\DeclarePairedDelimiterX{\norm}[1]{\lVert}{\rVert}{#1} 
\DeclarePairedDelimiterX{\normp}[1]{\lVert}{\rVert_p}{#1} 
\DeclarePairedDelimiterX{\normHS}[1]{\lVert}{\rVert_{\mathrm{HS}}}{#1}

\DeclareFontEncoding{FMS}{}{}
\DeclareFontSubstitution{FMS}{futm}{m}{n}
\DeclareFontEncoding{FMX}{}{}
\DeclareFontSubstitution{FMX}{futm}{m}{n}
\DeclareSymbolFont{fouriersymbols}{FMS}{futm}{m}{n}
\DeclareSymbolFont{fourierlargesymbols}{FMX}{futm}{m}{n}
\DeclareMathDelimiter{\VERT}{\mathord}{fouriersymbols}{152}{fourierlargesymbols}{147}

\usepackage{xr}

\makeatletter
\newcommand*{\addFileDependency}[1]{
  \typeout{(#1)}
  \@addtofilelist{#1}
  \IfFileExists{#1}{}{\typeout{No file #1.}}
}
\makeatother

\begin{document}

\bibliographystyle{agsm}

\def\spacingset#1{\renewcommand{\baselinestretch}%
{#1}\small\normalsize} \spacingset{1}


\if1\blind
{
  \title{\bf Partial Separability and Functional Graphical Models for Multivariate Gaussian Processes}
  \maketitle
} \fi

\if0\blind
{
  \bigskip
  \bigskip
  \bigskip
  \begin{center}
    {\LARGE\bf Title}
\end{center}
  \medskip
} \fi

\bigskip

\begin{abstract}

The covariance structure of multivariate functional data can be highly complex, especially if the multivariate dimension is large, making extensions of statistical methods 
for standard multivariate data to the functional data setting challenging.  For example, Gaussian graphical models have recently been extended to the setting of multivariate functional data by applying multivariate methods to the coefficients of truncated basis expansions.  However, a key difficulty compared to multivariate data is that the covariance operator is compact, and thus not invertible. The methodology in this paper addresses the general problem of covariance modeling for multivariate functional data, and functional Gaussian graphical models in particular. As a first step, a new notion of separability for the covariance operator of multivariate functional data is proposed, termed partial separability, leading to a novel Karhunen-Lo\`eve-type expansion for such data.  Next, the partial separability structure is shown to be particularly useful in order to provide a well-defined functional Gaussian graphical model that can be identified with a sequence of finite-dimensional graphical models, each of identical fixed dimension. This motivates a simple and efficient estimation procedure through application of the joint graphical lasso. Empirical performance of the method for graphical model estimation is assessed through simulation and analysis of functional brain connectivity during a motor task.
\end{abstract}

\noindent%
\vfill

\newpage

\section{Introduction}
\label{sec: intro}

The analysis of functional data continues to be an important field for statistical development given the abundance of data collected over time via sensors or other tracking equipment.  Frequently, such time-dependent signals are vector-valued, resulting in multivariate functional data.  Prominent examples include longitudinal behavioral tracking \cp{mull:14:1}, blood protein levels \cp{mull:05:2}, traffic measurements \cp{chio:14,chio:16}, and neuroimaging data \cp{mull:16:2,happ:18}, for which dimensionality reduction and regression have been the primary methods investigated.  As for standard multivariate data, the nature of dependencies between component functions of multivariate functional data constitute an important question requiring careful consideration.   

Dependencies between functional magnetic resonance imaging (fMRI) signals for a large number of regions across the brain during a motor task experiment are the motivating example for this paper.  Since fMRI signals are collected simultaneously, it is natural to model these as a multivariate process $\left\{X(t) \in \mathbb{R}^p: t \in \mathcal{T}\right\}$, where $\mc{T} \subset \mathbb{R}$ is a time interval over which the scans are taken \cp{qiao:19}.  The dual multivariate and functional aspects of the data make the covariance structure of $X$ 
complex, particularly if the multivariate dimension $p$ is large.  This leads to difficulties in extending highly useful multivariate analysis techniques, such as graphical models, to multivariate functional data without further structural assumptions.  For example, in the analogous setting of spatio-temporal data, it is common to impose further structure to the covariance, usually assuming that the spatial and temporal effects can be separated in some way.  However, similar notions for multivariate functional data have not yet been considered.

As for ordinary multivariate data, the conditional independence properties of $X$ are perhaps of greater interest than the marginal covariance, leading to the consideration of inverse covariance operators and graphical models for functional data. If $X$ is Gaussian, each component function $X_j$ corresponds to a node in the functional Gaussian graphical model, which is a single network of $p$ nodes.  This is inherently different from the estimation of time-dependent graphical models \cp{zhou:10,kola:11,qiu:16,qiao2020doubly}, in which the graph is dynamic and has nodes corresponding to scalar random variables.  In this paper, the graph is considered to be static while each node represents an infinite-dimensional functional object.  This is an important distinction, as covariance operators for functional data are compact and thus not invertible in the usual sense, so that presence or absence of edges cannot in general be identified immediately with zeros in any precision operator.  In the past few years, there has been some investigation into this problem. \ci{duns:16} developed a Bayesian framework for graphical models on product function spaces, including the extension of Markov laws and appropriate prior distributions. \ci{qiao:19} implemented a truncation approach, whereby each function is represented by the coefficients of a truncated basis expansion using functional principal components analysis, and a finite-dimensional graphical model is estimated by a modified graphical lasso criterion. \cite{li:18} developed a non-Gaussian variant, where conditional independence was replaced by a notion of so-called additive conditional independence.

The methodology proposed in this paper is within the setting of multivariate Gaussian processes as in \cite{qiao:19}, and exploits a notion of separability for multivariate functional data to develop efficient estimation of suitable inverse covariance objects. There are at least three novel contributions of this methodology to the fields of functional data analysis and Gaussian graphical models.  First, a structure termed partial separability is defined for the covariance operator of multivariate functional data, yielding a novel Karhunen-Lo\`eve type representation. 
The second contribution is to show that, when the process is indeed partially separable, the functional graphical model is well-defined and can be identified with a sequence of finite-dimensional graphical models.  In particular, the assumption of partial separability overcomes the problem of noninvertibility of the covariance operator when $X$ is infinite-dimensional, in contrast with \cite{duns:16,qiao:19} which assumed that the functional data were concentrated on finite-dimensional subspaces.  Third, an intuitive estimation procedure is developed 
based on simultaneous estimation of multiple graphical models.  Furthermore, theoretical properties are derived under the regime of fully observed functional data.  Empirical performance of the proposed method is then compared to that of \cite{qiao:19} through simulations involving dense and noisily observed functional data, including a setting where partial separability is violated.  Finally, the method is applied to the study of brain connectivity (also known as functional connectivity in the neuroscience literature) using data from the Human Connectome Project corresponding to a motor task experiment. Through these practical examples, our proposed method is shown to provide improved efficiency in estimation and computation.  
An R package 
\texttt{fgm}
implementing the proposed methods is freely available via the CRAN repository.

\section{Preliminaries} 
\label{sec: prelim}

\subsection{Gaussian Graphical Models}
Consider a $p$-variate random variable $\theta = (\theta_1, \ldots, \theta_p)\T,$ $p > 2.$  For any distinct indices $j,k=1,\ldots,p,$ let $\theta_{-(j,k)} \in \R^{p-2}$ denote the subvector of $\theta$ obtained by removing its $j$th and $k$th entries.  A graphical model \citep{laur:96}  for $\theta$ is an undirected graph $G = (V, E),$ where $V = \{1,\ldots,p\}$ is the node set and $E \subset V \times V \setminus \{(j,j): j \in V\}$ is called the edge set.  The edges in $E$ encode the presence or absence of conditional independencies amongst the distinct components of $\theta$ by excluding $(j,k)$ from $E$ if and only if $\theta_j \independent \theta_k \mid \theta_{-(j,k)}.$  In the case that $\theta \sim \mc{N}_p(0, \Sigma),$ the corresponding Gaussian graphical model is intimately connected to the positive definite covariance matrix $\Sigma$ through its inverse $\Omega = \Sigma^{-1},$ known as the precision matrix of $\theta$. Specifically, the edge set $E$ can be readily obtained from $\Omega$ by the relation $(j,k) \in E$ if and only if $\Omega_{jk} \neq 0$ \citep{laur:96}. This identification of edges in $E$ with the non-zero off-diagonal entries of $\Omega$ is due to the simple fact that the latter are proportional to the conditional covariance between components.  Thus, the zero/non-zero structure of $\Omega$ serves as an adjacency matrix of the graph $G,$ making disposable a vast number of statistical tools for sparse inverse covariance estimation in order to recover a sparse graph structure from data. 

\subsection{Functional Gaussian Graphical Models}
\label{ss: mfd}

We first introduce some notation.  Let $L^2\zo$ denote the space of square-integrable measurable functions on $\zo$ endowed with the standard inner product $\ip{g_1}{g_2} = \i01 g_1(t)g_2(t)\,\dt$ and associated norm $\norm{\cdot}.$  $\Ltp$ is its $p$-fold Cartesian product or direct sum, endowed with inner product $\ipp{f_1}{f_2} = \sum_{j = 1}^p \ip{f_{1j}}{f_{2j}}$ and its associated norm $\normp{\cdot}.$  For a generic compact covariance operator $\mc{A}$ defined on an arbitrary Hilbert space, let $\lambda_j^{\mc{A}}$ denote its $j$-th largest eigenvalue.  Suppose $f \in \Ltp,$ $g \in \Lt,$ $a \in \R^p,$ $\Delta$ is a $p\times p$ matrix, and $\mc{B}: \Lt \ra \Lt$ is a linear operator.  Then $ag \in \Ltp$ takes values $\{g(x)\}a \in \R^p,$ $\Delta f \in \Ltp$ takes values $\Delta \{ f(x)\} \in \R^p,$ $\mc{B}(f) = (\mc{B}(f_1),\ldots,\mc{B}(f_p)) \in \Ltp,$ and $(\Delta \mc{B})(f) = \mc{B}(\Delta f).$ The tensor products $g\otimes g$ and $f \otimes_p f$  signify the operators $(g\otimes g)(\cdot) = \ip{g}{\cdot}g$ and $(f \otimes_p f)(\cdot) = \ipp{f}{\cdot}f$ on $\Lt$ and $\Ltp,$ respectively.

In this paper, multivariate functional data constitute a random sample from a multivariate process $\left\{X(t) \in \mathbb{R}^p: t \in\zo\right\}$, which, for the moment, is assumed to be zero-mean such that $X \in \Ltp$ almost surely and $E\left(\normp{X}^2\right) < \infty$.  
If $X$ is also assumed to be Gaussian, then its distribution is uniquely characterized by its covariance operator $\mc{G},$ the infinite-dimensional counterpart of the covariance matrix for standard multivariate data.  In fact, one can think of it as a matrix of operators \mbox{$\mc{G} = \{\mc{G}_{jk}: j,k\in V \},$} where each entry $\mc{G}_{jk}$ is a linear, trace class integral operator on $L^2\zo$ \citep{hsin:15} with kernel $G_{jk}(s,t) = \cov\{X_j(s), X_k(t)\}.$ That is, for any $g \in L^2\zo,$ $\mc{G}_{jk}(g)(\cdot) = \i01 G_{jk}(\cdot,t)g(t) \dt.$  Then $\mc{G}$ is an integral operator on $\Ltp$ with $\{\mc{G}(f)\}_j = \sum_{k = 1}^p \mc{G}_{jk}(f_k)$ ($f \in \Ltp,$ $j \in V$).

A functional Gaussian graphical model for $X$ is a graph $G = (V, E)$ that encodes the conditional independency structure amongst its components.  As in the finite-dimensional case, the edge set can be recovered from the conditional covariance functions
\begin{equation}
\label{eq: cond_cov}
C_{jk}(s, t) = \cov\{X_j(s), X_k(t)\mid X_{-(j, k)}\} \quad (j,k \in V, j \neq k ),
\end{equation}
through the relation $(j,k) \in E$ if and only if $C_{jk}(s,t) = 0$ for all $s,t \in [0,1].$  However, unlike the finite-dimensional case, the covariance operator $\mc{G}$ is compact and thus not invertible, with the consequence that the connection between conditional independence and an inverse covariance operator is lost, as the latter does not exist.  This is an established issue for infinite-dimensional functional data, for instance in linear regression models with functional predictors; see \cite{mull:16:3} and references therein.  Thus, a common approach is to regularize the problem by first performing dimensionality reduction, most commonly through a truncated basis expansion of the functional data.  Specifically, one chooses an orthonormal functional basis $\{\phi_{jl}\}_{l = 1}^\infty$ of $L^2\zo$ for each $j,$ and expresses each component of $X$ as
\begin{equation}
    \label{eq: basisExp}
    X_j(t) = \sum_{l = 1}^\infty \xi_{jl}\phi_{jl}(t), \quad \xi_{jl} = \i01 X_j(t)\phi_{jl}(t)\dt.
\end{equation}
These expansions are then truncated at a finite number of basis functions to perform estimation, and the basis size is allowed to diverge with the sample size to obtain asymptotic properties.

Previous work related to functional Gaussian graphical models include \ci{duns:16} and \ci{qiao:19}.  The former rigorously considered the notion of conditional independence for functional data, and proposed a family of priors for the covariance operator $\mc{G}.$  The latter truncated \eqref{eq: basisExp} at $L$ terms using the functional principal component basis \citep{hsin:15}, and set $\xi_j = (\xi_{j1},\ldots,\xi_{jL})\T$ ($j \in V).$  \cite{qiao:19} define a $pL \times pL$ covariance matrix $\Gamma$ blockwise for the concatenated vector $(\xi_1\T,\ldots,\xi_p\T)\T$, as $\Gamma = (\Gamma_{jk})_{j,k=1}^p, (\Gamma_{jk})_{lm} = \cov(\xi_{jl}, \xi_{km}), \ (l,m=1,\ldots,L).$
Then, a functional graphical lasso algorithm was developed to estimate $\Gamma\inv$ with sparse off-diagonal blocks in order to estimate the edge set. 

The method of \ci{qiao:19} constitutes an intuitive approach to functional graphical model estimation, but encounters some difficulties that we seek to address in this paper.  From a theoretical point of view, even when $p$ is fixed, consistent estimation of the graphical model requires that one permit $L$ to diverge, so that the number of covariance parameters needing to be estimated is $(pL)^2.$  Additionally, the identification of zero off-diagonal blocks in $\Gamma\inv$ was only shown to be linked to the true functional graphical model under the strict assumption that each $X_j$ take values in a finite-dimensional space almost surely.  In many practical applications, the dimension $p$ can be high, the number of basis functions $L$ may need to be large in order to retain a suitable representation of the observed data, or both of these may occur simultaneously.  It is thus desirable to introduce structure on $\mc{G}$ in order to provide a parsimonious basis expansion for multivariate functional data that is amenable to graphical model estimation.

\section{Partial Separability}
\label{sec: notions}

\subsection{A Parsimonious Basis for Multivariate Functional Data}

Functional principal component analysis is a commonly used tool for functional data, with one of its most useful features being the parsimonious reduction of each univariate component $X_j$ to a countable sequence of uncorrelated random variables through the Karhunen-Lo\`eve expansion \citep{hsin:15}, taking the form of \eqref{eq: basisExp} when the basis is chosen as the eigenbasis of $\mc{G}_{jj}$.  \cite{chio:14} extended this to multivariate functional data via the spectral decomposition $\mc{G} = \sum_{m = 1}^\infty \lambda_m^\mc{G} \rho_m \otimes_p \rho_m,$ leading to the multivariate Karhunen-Lo\`eve expansion
$
X(t) = \sum_{l = 1}^\infty \ipp{X}{\rho_l}\rho_l(t)$ 
where $\{\rho_l\}_{l = 1}^\infty$ is an orthonormal basis of $\Ltp.$ While this decomposition is indeed parsimonious, the multivariate aspect of the data is lost since the random coefficients $\ipp{X}{\rho_l}$ are scalar.  As a consequence, one cannot readily apply tools from finite-dimensional Gaussian graphical model estimation to these coefficients.  We begin by proposing a novel structural assumption on the eigenfunctions of $\mc{G},$ termed \emph{partial separability}, and will then demonstrate its advantages for defining and estimating functional Gaussian graphical models.

\begin{definition}
\label{def: partial}
A covariance operator $\mc{G}$ on $\Ltp$ is \emph{partially separable} if there exist orthonormal bases $\{e_{lj}\}_{j = 1}^p$ $(l \in \mathbb{N})$ of $\R^p$ and $\{\varphi_l\}_{l = 1}^\infty$ of $\Lt$ such that the eigenfunctions of $\mc{G}$ take the form $e_{lj}\varphi_l$ ($l \in \mathbb{N},$ $j \in V$).
\end{definition}
We first draw a connection to separability of covariance operators as they appear in spatiotemporal analyses, after which the implications of partial separability will be further explored.  Dependent functional data arise naturally in the context of a spatiotemporal random field that is sampled at $p$ discrete spatial locations.  In many instances, it is assumed that the covariance of $X$ is \emph{separable}, meaning that there exists a $p\times p$ covariance matrix $\Delta$ and covariance operator $\mc{B}$ on $L^2\zo$ such that $\mc{G} = \Delta \mc{B}$ \cp{gnei:06,gent:07,asto:17}.  Letting $\{e_j\}_{j = 1}^p$ and $\{\varphi_l\}_{l = 1}^\infty$ be the orthonormal eigenbases of $\Delta$ and $\mc{B}$ respectively, it is clear that $e_j\varphi_l$ are the eigenfunctions of $\mc{G}.$  Hence, a separable covariance operator $\mc{G}$ satisfies the conditions of Definition~\ref{def: partial}.  It should also be noted that the property of $\mc{G}$ having eigenfunctions of the form $e_j\varphi_l$ has also been referred to as \emph{weak separability} \cp{chen:18}, and is a consequence and not a characterization of separability.  
The connections between these three separability notions are summarized in the following result, whose proof is simple, and thus omitted.

\begin{proposition}
\label{prop: sepCon}
Suppose $\mc{G}$ is partially separable according to Definition~\ref{def: partial}.  Then $\mc{G}$ is also weakly separable if and only if the bases $\{e_{lj}\}_{j = 1}^p$ do not depend on $l.$  If $\mc{G}$ is weakly separable, then it is also separable if and only if the eigenvalues take the form $\ipp{\mc{G}(e_j\varphi_l)}{e_j\varphi_l} = c_jd_l$ for positive sequences $\{c_j\}_{j = 1}^p,$ $\{d_l\}_{l = 1}^\infty.$
\end{proposition}

The next result gives several characterizations of partial separablity.  The proof of this and all remaining theoretical results can be found in the Appendix.
\begin{theorem}
\label{thm: PSequiv}
Let $\{\varphi_l\}_{l = 1}^\infty$ by an orthonormal basis of $\Lt.$  The following are equivalent:
\begin{enumerate}
 \item $\mc{G}$ is partially separable with $\Lt$ basis $\{\varphi_l\}_{l = 1}^\infty.$
    \item There exists a sequence of $p\times p$ covariance matrices $\{\Sigma_l\}_{l = 1}^\infty$ such that
    $$
    \mc{G} = \sum_{l = 1}^\infty \Sigma_l \varphi_l \otimes \varphi_l.
    $$
    \item The covariance operator of each $X_j$ can be written as $\mc{G}_{jj} = \sum_{l = 1}^\infty \sigma_{ljj} \varphi_l \otimes \varphi_l,$ with $\sigma_{ljj} > 0$ and $\suml \sigma_{ljj} < \infty,$ and $\cov(\ip{X_j}{\varphi_l},\ip{X_k}{\varphi_{l'}}) = 0$ ($j,k\in V,$\, $l \neq l'$). 
    \item The expansion 
    \begin{equation}
    \label{eq: PS_KL}
        X = \suml \theta_{l}\varphi_l, \quad \theta_l = (\ip{X_1}{\varphi_l},\ldots,\ip{X_p}{\varphi_l})\T,
    \end{equation}
    holds almost surely in $\Ltp$, where the $\theta_l$ are mutually uncorrelated random vectors.
\end{enumerate}
\end{theorem}

As will be seen in Section~\ref{ss: ps_FGGM}, the matrices $\Sigma_l$ in point 2 of Theorem~\ref{thm: PSequiv} contain all of the necessary information to form the functional graphical model when $X$ is Gaussian and $\mc{G}$ is partially separable.  For clarity, when $\mc{G}$ is partially separable, the expansion in point 2 is assumed to be ordered according to decreasing values of $\trace\left(\Sigma_l\right).$  Property 3 reveals that the $\mc{G}_{jj}$ share common eigenfunctions and are thus simultaneously diagonalizable, with projections of any features onto different eigenfunctions being uncorrelated.  This is related to the concept known as \emph{coregionalization} \citep{bane:06} but is more general in that it allows for different eigenvalues for each $\mc{G}_{jj}.$  Consequently, one obtains the vector Karhunen-Lo\`eve type expansion in \eqref{eq: PS_KL}. If one truncates \eqref{eq: PS_KL} at $L$ components, the covariance matrix of the concatenated vector $(\theta_1\T,\ldots,\theta_L\T)\T$ is block diagonal, with the $p\times p$ matrices $\Sigma_l = \var(\theta_l)$ constituting the diagonal blocks.   Figure~\ref{fig:Cov_Structure} visualizes this covariance structure in comparison with that of \cite{qiao:19}, along with comparisons of the inverse covariance structure.  The latter comparison is the more striking and relevant one, since the model of \cite{qiao:19} possesses a  potentially full inverse covariance structure, whereas that under partial separability remains block diagonal.  As a consequence, the model of \ci{qiao:19} has $\mathcal{O}(L^2p^2)$ free parameters, while the corresponding model under partial separability has only $\mathcal{O}(Lp^2)$ free parameters. 

\begin{figure}[t]
\centering
\begin{subfigure}[t]{0.23 \linewidth}
		\includegraphics[width=0.9\linewidth]{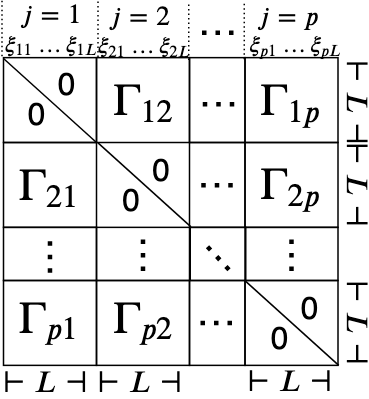}
		\caption{}
		\label{fig:Cov_Structure:a}		
\end{subfigure}
\hspace{-.4cm}
\begin{subfigure}[t]{0.23 \linewidth}
		\centering\captionsetup{width=1\linewidth}%
		\includegraphics[width=0.9\linewidth]{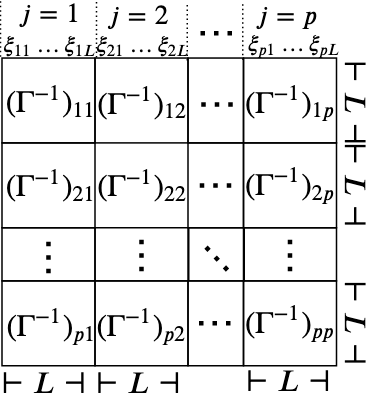}
		\caption{}
		\label{fig:Cov_Structure:b}		
\end{subfigure}
\hspace{0.2cm}
\begin{subfigure}[t]{0.23 \linewidth}
		\centering\captionsetup{width=1\linewidth}%
		\includegraphics[width=0.9\linewidth]{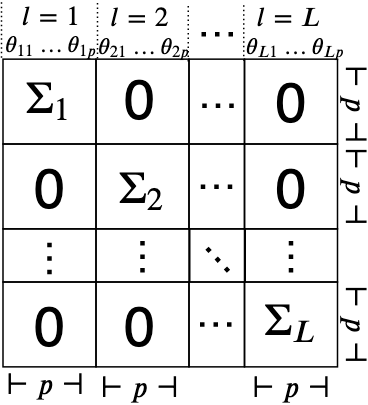}
		\caption{}
		\label{fig:Cov_Structure:c}		
\end{subfigure}
\hspace{-.4cm}
\begin{subfigure}[t]{0.23 \linewidth}
		\centering\captionsetup{width=1\linewidth}%
		\includegraphics[width=0.9\linewidth]{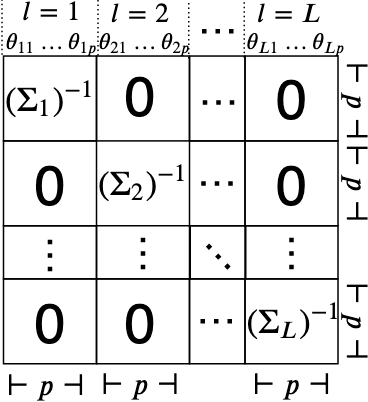}
		\caption{}
		\label{fig:Cov_Structure:d}		
\end{subfigure}

\caption{ Covariance structures of $\mathbb{R}^{Lp}$-valued random coefficients from different $L$-truncated  Karhunen-Lo\`eve type expansions. (a) and (b): covariance and precision matrices, respectively, of functional principal component coefficients $(\xi_1^T, \dots, \xi_p^T)^T$ in \eqref{eq: basisExp} as in \cite{qiao:19}.  (c) and (d): block diagonal covariance and precision matrices, respectively, of coefficients $(\theta_1\T,\ldots,\theta_L\T)\T$ under partial separability in \eqref{eq: PS_KL}.}
\label{fig:Cov_Structure}
\end{figure}

Lastly, we establish optimality and uniqueness properties for the basis $\{\varphi_l\}_{l = 1}^\infty$ of a partially separable $\mc{G}$.  A key object is the trace class covariance operator 
\begin{equation}
    \label{eq: H_cov}
    \mc{H} = \frac{1}{p}\sum_{j = 1}^p \mc{G}_{jj}.
\end{equation}
Let $\lambda_l = \lambda_l^\mc{H}$ ($l \in \mathbb{N})$ denote the eigenvalues of $\mc{H},$ in nonincreasing order.
\begin{theorem}
\label{thm: psOptimal} Suppose the eigenvalues of $\mc{H}$ in \eqref{eq: H_cov} have multiplicity one. 
\begin{enumerate}
    \item For any $L \in \mathbb{N},$ and for any orthonormal set $\{\tilde{\varphi}_l\}_{l = 1}^L$ in $L^2[0,1],$
$
\sum_{l = 1}^L \sum_{j = 1}^p \var(\ip{X_j}{\tilde{\varphi}_l}) \leq  \sum_{l = 1}^L \lambda_l,
$
with equality if and only if $\{\tilde{\varphi}\}_{l=1}^L$ span the first $L$ eigenspaces of $\mc{H}.$  
\item If $\mc{G}$ is partially separable with $\Lt$ basis $\{\varphi_l\}_{l = 1}^\infty,$ then 
\begin{equation}
    \label{eq: H_covPS}
    \mc{H} = \suml \lambda_l \varphi_l \otimes\varphi_l, \quad \lambda_l = \frac{1}{p}\trace(\Sigma_l).
\end{equation}
\end{enumerate}
\end{theorem}
Part 1 states that, independent of partial separability, the eigenbasis of $\mc{H}$ is optimal in terms of retaining the greatest amount of total variability in vectors of the form $(\ip{X_1}{\tilde{\varphi}_l},\ldots,\ip{X_p}{\tilde{\varphi}_l})\T,$ subject to orthogonality constraints.  Part 2 indicates that, if $\mc{G}$ is partially separable, the unique basis of $\Lt$ that makes Definition~\ref{def: partial} hold corresponds to this optimal basis. 

\subsection{Consequences for Functional Gaussian Graphical Models}
\label{ss: ps_FGGM}

Assume that $\mc{G}$ is partially separable according to Definition~\ref{def: partial}, so that the partially separable Karhunen-Lo\`eve expansion in \eqref{eq: PS_KL} holds. If we further assume that $X$ is Gaussian, then $\theta_l \sim \mc{N}(0, \Sigma_l),$ $l \in \mathbb{N}$, are independent, where $\Sigma_l$ is positive definite for each $l.$ These facts follow from Theorem~\ref{thm: PSequiv}.  Recall that, in order to define a coherent functional Gaussian graphical model, one needs that the conditional covariance functions $C_{jk}$ in \eqref{eq: cond_cov} between component functions $X_j$ and $X_k$ be well-defined.  The expansion in \eqref{eq: PS_KL} facilitates a simple connection between the $C_{jk}$ and the inverse covariance matrices $\Omega_l = \Sigma_l\inv,$ as follows.  Let $\Sigma_{l} = (\sigma_{ljk})_{j,k=1}^p.$ For any fixed $j, k \in V,$ define the partial covariance between $\theta_{lj}$ and $\theta_{lk}$ as 
\begin{equation}
\label{eq: cond_cov_theta}
\tilde{\sigma}_{ljk} = \sigma_{ljk} - \cov\{\theta_{lj}, \theta_{l, -(j,k)}\}\var\{\theta_{l,-(j,k)}\}\inv \cov\{\theta_{l, -(j,k)}, \theta_{lk}\}.
\end{equation}
It is well-known that these partial covariances are directly related to the precision matrix $\Omega_l = (\omega_{ljk})_{j,k=1}^p$ by 
$
\tilde{\sigma}_{ljk} = -\omega_{ljk}/ (\omega_{ljj}\omega_{lkk} - \omega_{ljk}^2),
$
so that $\tilde{\sigma}_{ljk} = 0$ if and only if $\omega_{ljk} = 0.$  The next result establishes that the conditional covariance functions $C_{jk}$ can be expanded in the partial separability basis $\{\varphi_l\}_{l=1}^\infty$ with coefficients $\tilde{\sigma}_{ljk}$.
\begin{theorem}
\label{thm: ps-FGGM}
If $\mc{G}$ is partially separable, then the cross-covariance kernel between $X_j$ and $X_k$, conditional on the multivariate subprocess $\{X_{-(j,k)}(u):\, u \in [0,1]\},$ is
\begin{equation}
\label{eq: cond_cov_ps}
C_{jk}(s,t) = \sum_{l = 1}^\infty \tilde{\sigma}_{ljk}\varphi_l(s)\varphi_l(t) \quad (j,k \in V,\, j \neq k,\, s,t \in [0,1]).
\end{equation}
\end{theorem}

Now, the conditional independence graph for the multivariate Gaussian process can be defined by $(j, k) \notin E$ if and only if $C_{jk}(s, t) \equiv 0$.  Due to the above result, the edge set $E$ is connected to the sequence of edge sets $\{E_l\}_{l=1}^\infty,$ for which $(j, k) \notin E_l$ if and only if $\tilde{\sigma}_{ljk} = \omega_{ljk} = 0,$ corresponding to the sequence of Gaussian graphical models $(V,E_l)$ for each $\theta_l.$  

\begin{corollary}
\label{cor: PSedges}
Under the setting of Theorem~\ref{thm: ps-FGGM}, the functional graph edge set $E$ is related to the sequence of edge sets $E_l$ by $E = \bigcup_{l = 1}^\infty E_l.$
\end{corollary}

We have thus established that, under partial separability, the problem of functional graphical model estimation can be simplified to estimation of a sequence of decoupled graphical models.  When partial separability fails, the edge sets remain meaningful. Recall from Theorem~\ref{thm: psOptimal} that the eigenbasis of $\mc{H}$ is optimal in a sense independent of partial separability, so that the vectors $\theta_l = (\ip{X_1}{\varphi_l},\ldots,\ip{X_p}{\varphi_l})\T$ are still the coefficients of $X$ in an optimal expansion. Although one loses a direct connection between the $E_l$ and the edge set of the functional graph, each $E_l$ remains the edge set of the Gaussian graphical model for the coefficient vector $\theta_l$ in this optimal expansion.  Moreover, the equivalence $E = \bigcup_{l = 1}^\infty E_l$ may hold independent of partial separability.  For instance, Proposition~\ref{prop: markov} in Section~\ref{s-sec: notions} of the Appendix gives sufficient conditions, based on a Markov-type property and a edge coherence assumption, under which the equivalence holds.

\section{Graph Estimation and Theory}
\label{sec: graphEst}

\subsection{Joint Graphical Lasso Estimator}

Consider a $p$-variate process $X,$ with means $\mu_j(t) = E\{X_j(t)\}$ and covariance operator $\mc{G}$. Let $\{\varphi_l\}_{l = 1}^\infty$ be an orthonormal eigenbasis of $\mc{H}$ in \eqref{eq: H_cov}, and set $\theta_{lj} = \ip{X_j}{\varphi_l},$ $\Sigma_l = \var(\theta_l)$. The targets are the edge sets $E_l,$ where $(j,k) \in E_l$ if and only if $(\Sigma_l\inv)_{jk} \neq 0$, as motivated by the developments of Section~\ref{ss: ps_FGGM}.  Specifically, when $X$ is Gaussian and $\mc{G}$ is partially separable, the conditional independence graph of $X$ has edge set $E = \bigcup_{l = 1}^\infty E_l.$  When partial separability fails, these targets still provide useful information about the conditional independencies of $X$ when projected onto the eigenbasis of $\mc{H},$ which is optimal in the sense of Theorem~\ref{thm: psOptimal}.  Furthermore, when $X$ is not Gaussian, rather than representing conditional independence, the $E_l$ represent the sparsity structure of the partial correlations of $\theta_l,$ which may still be of interest.
By Theorem~\ref{thm: PSequiv}, $\trace(\Sigma_l) = \lambda_l \downarrow 0$ as $l \rightarrow \infty$. As a practical consideration, this makes estimators of $\Sigma_l$ progressively more unstable to work with as $l$ increases. To avoid this, we work with $\Xi_l = R_l\inv,$ where $R_l$ is the correlation matrix corresponding to $\Sigma_l.$  $\Xi_l$ and $\Omega_l$ share the same edge information as entries in these two matrices are either zero or nonzero simultaneously.

We first define the estimation procedure with targets $\Xi_l,$  from a random sample $X_1,\ldots,X_n$, each distributed as $X.$ $X$ is not required to be Gaussian, nor $\mc{G}$ to be partially separable, in developing the theoretical properties of the estimators, which also allow the dimension $p$ to diverge with $n.$  In order to make these methods applicable to any functional data set, it is assumed that preliminary mean and covariance estimates $\hat{\mu}_j$ and $\hat{\mc{G}}_{jk},$ $j,k = 1,\ldots,p,$ have been computed for each component.  As an example, if the $X_i$ are fully observed, cross-sectional estimates 
\begin{equation}
\label{eq: mean_cov_fo}
\hat{\mu}_j = \frac{1}{n}\sum_{i = 1}^n X_{ij}, \quad \hat{\mc{G}}_{jk} = \frac{1}{n}\sum_{i = 1}^n (X_{ij} - \hat{\mu}_j) \otimes (X_{ik} - \hat{\mu}_k), 
\end{equation}
can be used. For practical observational designs, smoothing can be applied to the pooled data to estimate these quantities \cp{mull:05:4, mull:11:2}.  Given such preliminary estimates, the estimate of $\mc{H}$ is
$
\hat{\mc{H}} = p\inv\sum_{j = 1}^p \hat{\mc{G}}_{jj},
$
leading to empirical eigenfunctions $\hat{\varphi}_l,$ which are uniquely defined only up to a sign and for $1 \leq l \leq np = \mathrm{rank}(\mc{H}).$  These quantities produce estimates of $\sigma_{ljk} = \ip{\mc{G}_{jk}(\varphi_l)}{\varphi_l}$ by plugin as
\begin{equation}
\label{eq: Sigma_est}
s_{ljk} = \left(S_l\right)_{jk}  
= \ip{\hat{\mc{G}}_{jk}(\hat{\varphi}_l)}{\hat{\varphi}_l}.
\end{equation}

A group graphical lasso approach \cp{dana:14} will be used to estimate the $\Xi_l.$  Let $(\hat{R}_l)_{jk} = \hat{r}_{ljk} = s_{ljk}/[s_{ljj}s_{lkk}]^{1/2}$ be the estimated correlations. The estimates target the first $L \leq np$ inverse correlation matrices $\Xi_l$ by
\begin{equation}
\label{eq: Xi_est}
(\hat{\Xi}_1,\ldots,\hat{\Xi}_L) = \arg \min_{\Up_l \succ 0, \Up_l = \Up_l\T} \sum_{l = 1}^L\left\{\trace (\hat{R}_l \Up_l) - \log(|\Up_l|)\right\} + P(\Up_1,\ldots,\Up_L),
\end{equation}
In the Gaussian case, these are penalized likelihood estimators with penalty
\begin{equation}
\label{eq: penalty}
P(\Up_1,\ldots,\Up_L) = \gamma\left\{\alpha \sum_{l = 1}^L\sum_{j \neq k} |\upsilon_{ljk}| + (1-\alpha)\sum_{j \neq k}\left(\sum_{l = 1}^L \upsilon_{ljk}^2\right)^{1/2}\right\}, \quad (\Up_l)_{jk} = \upsilon_{ljk}.
\end{equation}
The parameter $\gamma > 0$ controls the overall penalty level, while $\alpha \in [0,1]$ distributes the penalty between the two penalty terms.  Then the estimated edge set is $(j, k) \in \hat{E}_l$ if and only if $\hat{\Xi}_{ljk} \neq 0.$ The joint graphical lasso was chosen to borrow structural information across multiple bases instead of multiple classes as was done in \ci{dana:14}.  If $\alpha = 1$, the first penalty will encourage sparsity in each $\hat{\Xi}_l$ and the corresponding edge set $\hat{E}_l,$ but the overall estimate $\hat{E} = \bigcup_{l = 1}^L \hat{E}_l$ may not be sparse.  While consistent graph recovery is still possible with $\alpha = 1$ as demonstrated below in Theorem~\ref{thm: selection}, the influence of the second penalty term when $\alpha < 1$ ensures that the overall graph estimate is sparse, enhancing interpretation.

In practice, tuning parameters $\gamma$ and $\alpha$ can be chosen with cross-validation to minimize \eqref{eq: Xi_est} for out-of-sample data. Specifically, the procedure would select $\gamma$ and $\alpha$ that minimize the average of \eqref{eq: Xi_est} evaluated over each fold, where $\Up_1,\dots,\Up_L$ are computed with the training set and $\hat R_l$ are from the validation set. Another practically useful, and less computationally intensive, approach is to choose these parameters to yield a desired sparsity level of the estimated graph \citep{qiao:19}.  This latter approach is implemented in the data example of Section~\ref{sec: app}.

\subsection{Asymptotic Properties}
\label{ss: asym}

The goal of the current section is to provide lower bounds on the sample size $n$ so that, with high probability, $\hat{E}_l = E_l$ ($l = 1,\ldots,L)$.  The approach follows that of \cite{ravikumar2011high}, adapting the results to the case of functional graphical model estimation in which multiple graphs are estimated simultaneously.  For simplicity, and to facilitate comparisons with the asymptotic properties of \cite{qiao:19}, the results are derived under the setting of fully observed functional data, so that $\hat{\mu}$ and $\hat{\mc{G}}_{jk}$ are as in \eqref{eq: mean_cov_fo}.  As a preliminary result, we first derive a concentration inequality for the estimated covariances $s_{ljk}$ in \eqref{eq: Sigma_est}, requiring the following mild assumptions.
\begin{assumption}
\label{asm: lambda}
The eigenvalues $\lambda_l$ of $\mc{H}$ have multiplicity one, and are thus strictly decreasing.  
\end{assumption}
\begin{assumption}
\label{asm: subGauss}
There exists $\varsigma^2 > 0$ such that $E\left(e^{s\theta_{lj}}\right) \leq e^{s^2\varsigma^2\sigma_{ljj}/2}$ for all $l \in \mathbb{N}$, $j\in V,$ and $s \in \R$; that is, the standardized scores $\theta_{lj}/\sigma_{ljj}^{1/2}$ are sub-Gaussian random variables with parameter $\varsigma^2.$ Furthermore, there is $M$ independent of $p$ such that $\sup_{j \in V} \sum_{l = 1}^\infty \sigma_{ljj} < M < \infty.$
\end{assumption}
Assumption 2 can be relaxed to accommodate eigenvalues with multiplicity greater than 1, at the cost of an increased notational burden.  The eigenvalue spacings play a key role through the quantities $\tau_1 = 2\sqrt{2}(\lambda_1 - \lambda_2)\inv$ and $\tau_l = 2\sqrt{2}\max\left\{(\lambda_{l-1} - \lambda_l)\inv, (\lambda_{l} - \lambda_{l+1})\inv\right\},$ for $l \geq 2.$  Assumption~\ref{asm: subGauss} clearly holds in the Gaussian case, and can be relaxed to accommodate different parameters $\varsigma^2_l$ for each $l,$ though for simplicity these are assumed uniform.
\begin{theorem}
\label{thm: concIneq}
Suppose that Assumptions~\ref{asm: lambda} and \ref{asm: subGauss} hold.  Then there exist constants $C_1, C_2, C_3 > 0$ such that, for any $0 < \delta \leq C_3$ and for all $l \leq np$ and $j,k \in V,$
\begin{equation}
\label{eq: concIneq}
\pr\left(|s_{ljk} - \sigma_{ljk}| \geq \delta\right) \leq C_2\exp\left(-C_1 \tau_l^{-2} n \delta^2\right).
\end{equation}
\end{theorem}
Concentrations inequalities such as \eqref{eq: concIneq} are generally required in penalized estimation problems where the dimension diverges to infinity.  For the current problem, even if the dimension $p$ of the process remains fixed, the dimension still diverges since one requires $L$ to diverge with $n.$  Furthermore, in contrast to standard multivariate scenarios, the bound in Theorem~\ref{thm: concIneq} contains the additional factor $\tau_l^{-2}$. Since $\lambda_l \downarrow 0,$ $\tau_l$ diverges to infinity with $l,$ so that \eqref{eq: concIneq} reflects the increased difficulty of estimating covariances corresponding to eigenfunctions with smaller eigenvalue gaps. 

\begin{remark}
\label{rmk: qiaoComp}
A similar result to Theorem~\ref{thm: concIneq} was obtained by \ci{qiao:19} under a specific eigenvalue decay rate and truncation parameter scheme. Imposing similar assumptions on the eigenvalues of $\mc{H},$ we have $\tau_l = O(l^{1 + \beta_1})$ for some $\beta_1 > 1$, so that for any $0 < \beta_2 < 1/(4\beta_1)$ and $L = n^{\beta_2},$ \eqref{eq: concIneq} implies
$$
\max_{l = 1,\ldots,L} \max_{j,k \in V} \pr\left(|s_{ljk} - \sigma_{ljk}| \geq \delta\right) \leq C_2\exp\{-C_1 n^{1 - 2\beta_2(1 + \beta_1)} \delta^2\},
$$
matching the rate of \ci{qiao:19}.  In addition to establishing the concentration inequality for a general eigenvalue decay rate, our proof is greatly simplified by using the inequality
\begin{equation}
\label{eq: bosqIneq}
|s_{ljk} - \sigma_{ljk}| \leq 2\tau_l\lVert \mc{G}_{jk}\rVert_{\mathrm{HS}} \lVert \hat{\mc{H}} - \mc{H}\rVert_{\mathrm{HS}} + \lVert \hat{\mc{G}}_{jk} - \mc{G}_{jk} \rVert_{\mathrm{HS}},
\end{equation}
where $\lVert \cdot \rVert_{\mathrm{HS}}$ is the Hilbert-Schmidt operator norm. 
\end{remark}
\begin{remark}
\label{rmk: jirak}
The bound in \eqref{eq: bosqIneq} utilizes a basic eigenfunction inequality found, for example, in Lemma 4.3 of \cite{bosq:00}; see also \cite{bhatia1983perturbation}.  However, using expansions instead of geometric inequalities, \cite{jirak2016optimal} and other authors cited therein established stonger results for differences between true and estimated eigenfunctions in the form of limiting distributions and moment bounds.  Thus, it is likely that the bound in \eqref{eq: concIneq} is suboptimal, although improvements along the lines of \cite{jirak2016optimal} would require further challenging work in order to establish the required exponential tail bounds. 
\end{remark}
As the objective \eqref{eq: Xi_est} utilizes the correlations $\hat{r}_{ljk},$ the following corollary is needed.
\begin{corollary}
\label{cor: concIneqCor}
Under the assumptions of Theorem~\ref{thm: concIneq}, there exists constants $D_1, D_2, D_3 > 0$ such that, for any $0 < \delta \leq D_3$ and for all $l\leq np$ and $j,k\in V,$
\begin{equation}
    \label{eq: concIneqCor}
    \pr\left(|\hat{r}_{ljk} - r_{ljk}| \geq \delta\right) \leq D_2 \exp\left(-D_1 n m_l^2\delta^2\right), \quad m_l = \tau_l\inv\pi_l, \, \pi_l = \min_{j \in V} \sigma_{ljj}.
\end{equation}
\end{corollary}

To establish consistency of $\hat{E}_l$, some additional notation will be introduced.  Let $\Psi_l = R_l \tilde{\otimes} R_l,$ where $\tilde{\otimes}$ is the Kronecker product, and $\overline{E}_l = E_l \cup\{(1,1),\ldots,(p,p)\}$.  For $B \subset V \times V,$ let $\Psi_{l, BB}$ denote the submatrix of $\Psi_l$ indexed by sets of pairs $(j,k) \in B,$ where $\Psi_{l,(j,k),(j',k')} = R_{ljj'}R_{lkk'}.$  For a $p\times p$ matrix $\Delta,$ let $\left\VERT \Delta \right\VERT_\infty = \max_{j =1,\ldots,p} \sum_{k = 1}^p \lvert\Delta_{jk}\rvert$. The following assumption corresponds to the irrepresentability or neighborhood stability condition often seen in sparse matrix and regression estimation \citep{ravikumar2011high,meinshausen2006high}.   
\begin{assumption}
\label{asm: irrepresentability}
For $l = 1,\ldots,L,$ there exists $\eta_l \in (0,1]$ such that
$$
 \left\VERT \Psi_{l,\overline{E}_l^c \overline{E}_l}\left(\Psi_{l, \overline{E}_l \overline{E}_l}\right)\inv\right\VERT_\infty \leq 1 - \eta_l.
$$
\end{assumption}
\noindent For fixed $l$, Assumption~\ref{asm: irrepresentability} was employed by \cite{ravikumar2011high} as sufficient for model selection consistency.  As Theorem~\ref{thm: selection} below implies simultaneous consistency of the first $L$ edge sets, we require the assumption for each $l$.  

Set $\kappa_{R_l} = \left\VERT R_l \right\VERT_\infty,$ $\kappa_{\Psi_l} = \VERT (\Psi_{l,\overline{E}_l\overline{E}_l})\inv\VERT_\infty,$ let $y_l = \max_{j \in V}\lvert \{k \in V:\, \Xi_{ljk} \neq 0\}\rvert$ be the maximum degree of the graph $(V, E_l),$ and $\xi_{\mathrm{min},l} = \min\{|\Xi_{ljk}|: \Xi_{ljk} \neq 0 \}.$ Finally, when Assumptions~\ref{asm: lambda}--\ref{asm: irrepresentability} hold, for any $\alpha > 1 - \min_{l = 1,\ldots,L} \eta_l$, define $\eta_l' = \alpha + \eta_l - 1 > 0$ and $\epsilon_L = \min_{1 \leq l \leq L} \eta_l' m_l$.  Then, with $D_3$ as in Corollary~\ref{cor: concIneqCor}, set
\[
\begin{split}
    \mathfrak{a}_L &= D_3\min_{l = 1,\ldots,L} m_l, \\
    \mathfrak{b}_L &= \min_{l = 1,\ldots,L} \left\{6 y_l m_l \max(\kappa_{\Psi_l}^2\kappa_{R_l}^3,\kappa_{\Psi_l}\kappa_{R_l})(m_l\inv + 8 \epsilon_L\inv)^2\right\}\inv, \\
    \mathfrak{c}_L &= \min_{l = 1,\ldots,L} \xi_{\mathrm{min},l}\left\{4 \kappa_{\Psi_l}(m_l\inv + 8 \epsilon_L\inv)\right\}^{-1}.
\end{split}
\]
Tracking $\mathfrak{a}_L$ and $\mathfrak{b}_L$, including the maximal degrees $y_l$, allows one to obtain uniform consistency of the estimates $\hat{\Xi}_l$ in \eqref{eq: Xi_est}, and to conclude that $\hat{E}_l \subset E_l$ with high probability; see Lemma~\ref{lma: estimation} in Section~\ref{ss: selProofs} of the Appendix.  The quantity $\mathfrak{c}_L$ involves the weakest nonzero signal $\xi_{\mathrm{min},l}$ of each graph, with weaker signals requiring larger $n$ for recovery. The quantities $\mathfrak{b}_L$ and $\mathfrak{c}_L$ decrease with $\alpha$, so that smaller values of $\alpha$ also require larger $n$. 

\begin{theorem}
\label{thm: selection}
Suppose Assumptions~\ref{asm: lambda}--\ref{asm: irrepresentability} hold, $(1 - \min_{l = 1,\ldots,L} \eta_l) < \alpha \leq 1$, $L \leq np,$ and that, for some $\varrho > 2,$ $\gamma = 8 \epsilon_L\inv\left\{(D_1 n)\inv \log\left(D_2L^{\varrho-1}p^\varrho\right)\right\}^{1/2}$, where $D_1$ and $D_2$ are constants from Corollary~\ref{cor: concIneqCor}.  If the sample size $n$ satisfies
\begin{equation}
    \label{eq: nlb2}
    n\min(\mathfrak{a}_L,\mathfrak{b}_L,\mathfrak{c}_L)^2 > D_1\inv \left\{\log(D_2) + (\varrho - 1)\log(n) + (2\varrho - 1)\log(p)\right\},
\end{equation}
then, with probability at least $1 - (Lp)^{2 - \varrho},$ $\hat{E}_l = E_l$ for all $l = 1,\ldots,L.$
\end{theorem}

\begin{remark}
\label{rmk: HD}
If $L$ grows sufficiently slowly with respect to $n$, \eqref{eq: nlb2} becomes
$
n\min(\mathfrak{a}_L, \mathfrak{b}_L, \mathfrak{c}_L)^2 \gtrsim \varrho \log(p)
$
for large $n$, $\gtrsim$ denoting inequality up to a multiplicative constant.  Hence, the conclusion of Theorem~\ref{thm: selection} holds for large $n$ when $\log(p) = o(n)$ by suitably limiting the growth of $L$.
\end{remark}

\begin{remark}
\label{rmk: graphProp}
To understand how the graph properties affect the lower bound, assume $\kappa_{\Psi_l},$ $\kappa_{R_l},$ and $\eta_l$ do not depend on $l, n,$ or $p,$ and $\min_{1 \leq l \leq n} m_l \gtrsim n^{-d}$, $0 < d < 1/4.$ Then \eqref{eq: nlb2} becomes
$$
n\gtrsim \left[\left\{\left(\max_{1 \leq l \leq L} \xi_{\mathrm{min},l}^{-2}\right) + \left(\max_{1 \leq l \leq L} y_l^2\right)\right\}\varrho \log(p)\right]^{1 - 4d}
$$
asymptotically.  In particular, if $L$ remains bounded so that $d = 0,$ the above bound is consistent with that of \cite{ravikumar2011high}, where the maxima over $l$ reflect the need to satisfy the bound for the edge set $E_l$ that is most difficult to estimate.
\end{remark}

If the conditional independence graph of $X$ is $E = \bigcup_{l = 1}^\infty E_l,$ as is the case when $X$ is Gaussian and $\mc{G}$ partially separable, Theorem~\ref{thm: selection} can lead to edge selection consistency in the functional graphical model.  For a given $n$, there exists a finite $L_p^*$ such that $E = \bigcup_{l = 1}^{L^*_p} E_l$, where $L_p^*$ can only diverge with $n$ if $p$ does so.  The following corollary is then immediate.
\begin{corollary}
\label{cor: fgm_consistency}
Under the assumptions of Theorem~\ref{thm: selection}, suppose $\log(p) = o(n),$ $L \rightarrow \infty$ such that $\min(\mathfrak{a}_{L},\mathfrak{b}_L, \mathfrak{c}_L)[n/\log\{\max(n,p)\}]^{1/2} \rightarrow \infty$ and, for large $n,$ $L \geq L_p^*$. Then, for large $n$, $E = \bigcup_{l = 1}^L \hat{E}_l$ with probability at least $1 - (Lp)^{2 - \varrho}.$
\end{corollary}

\section{Numerical Experiments}
\label{sec: sim}

\subsection{Simulation Settings}
The simulations in this section compare the proposed method for partially separable functional Gaussian graphical models, with that of \citet{qiao:19}. Throughout this section we denote these methods as psFGGM and FGGM, respectively. Throughout this section we denote these methods as psFGGM and FGGM, respectively. Other potentially competing non-functional based approaches are not included since they are clearly outperformed by the latter \cp{qiao:19}. An initial conditional independence graph $G=(V,E)$ is generated from a power law distribution with parameter $\pi = \pr\{ (j,k) \in E\}$. Then, for a fixed $M$, a sequence of edge sets $E_1,\dots,E_M$ is generated so that $E = \bigcup_{l = 1}^M E_l$. A set of common edges to all edge sets is computed for a given proportion of common edges $\tau \in [0,1]$. Next, $p \times p$ precision matrices $\Omega_l$ are generated for each $E_l$ based on the algorithm of \citet{peng:09}. A fully detailed description of this step is included in the Section~\ref{s-sec: sim} of the Appendix.

Random vectors $\theta_i\in \mathbb{R}^{Mp}$ are then generated from a mean zero multivariate normal distribution with covariance matrix $\Sigma,$ yielding discrete and noisy functional data
$$
Y_{ijk}= X_{ij}(t_k)+ \varepsilon_{ijk}, \quad \varepsilon_{ijk}\sim N(0,\sigma_\varepsilon^2)\quad (i=1,\ldots,n ;\ j=1,\ldots, p;\ k=1,\ldots, M).$$ 
Here, $\sigma_\epsilon^2= 0.05\sum_{l=1}^M \trace(\Sigma_l)/p$ and $X_{ij}(t_k) = \sum_{l=1}^M \theta_{ilj}\varphi_l(t_k)$ according to the partially separable Karhunen-Lo\`eve expansion in \eqref{eq: PS_KL}. Fourier basis functions $\varphi_1,\dots,\varphi_M$ evaluated on an equally spaced time grid of $t_1,\dots,t_T$, with $t_1=0$ and $t_T=1$, were used to generate the data. In all settings, $100$ simulations were conducted. To resemble real data example from Section~\ref{sec: app} below, we set $T=30$, $M=20$ and $\pi=5\%$ for a sparse graph.

We consider two models for $\Sigma$, corresponding to partially separable and non-partially separable $X$, respectively.  In the first, the covariance $\Sigma_{\text{ps}}$ is formed as a block diagonal matrix with $p\times p$ diagonal blocks $\Sigma_l = a_l \Omega_l^{-1}$.  The decaying factors $a_l =  3l^{-1.8}$ guarantee that $\trace(\Sigma_l)$ decreases monotonically in $l.$   In the second, $\Sigma_{\text{ps}}$ is modified to violate partial separability.  Specifically, a block-banded precision matrix $ \Omega$ is computed with $p\times p$ blocks $\Omega_{l,l} =\Omega_{l}$ and $\Omega_{l+1,l}=\Omega_{l,l+1}=0.5(\Omega_{l}^*+\Omega_{l+1}^*)$ with $\Omega_l^*=\Omega_l-\diag(\Omega_l)$.  Then, the non-partially separable covariance is computed as $\Sigma_\text{non-ps}=\diag(\Sigma_\text{ps})^{1/2} \Omega\inv  \diag(\Sigma_\text{ps})^{1/2}$.

\subsection{Comparison of Results}

Comparisons between the proposed method and that of \ci{qiao:19}, implemented using code provided by the authors, are presented here.  Additional comparisons obtained by thresholding correlations are provided in Section~\ref{ss: screen} of the Appendix.  Although this alternative method does not estimate a sparse inverse covariance structure, its graph recovery is competitive with that of the proposed method in some settings. As performance metrics, the true and false positive rates of correctly identifying edges in graph $G$ are computed over a range of $\gamma$ values and a coarse grid of five evenly spaced points $\alpha \in [0,1]$. The value of $\alpha$ maximizing the area under the receiver operating characteristic curve is considered for the comparison. In all cases, we set $\pi=0.05$ and $\tau = 0.$  The two methods are compared using $L$ principal components explaining at least $90\%$ of the variance. For all simulations and both methods, this threshold results in the choice of $L=5$ or $L = 6$ components. For higher variance explained thresholds, however, we see a sharp contrast. While the proposed method consistently converges to a solution, that of \ci{qiao:19} does not, due to increasing numerical instability. The reason for the instability is the need to estimate $L = 5$ or $6$ times more parameters compared to the proposed method.  The proposed method can thus accommodate larger $L$, thereby incorporating more information from the data.  In the figures and tables, additional results are available for the proposed method when $L$ is increased to explain at least $95\%$ of the variance.

\begin{figure}[t] 
\centering
\subcaptionbox{$n = p / 2$\label{fig:ROC1}}{\includegraphics[height=0.26\textheight]{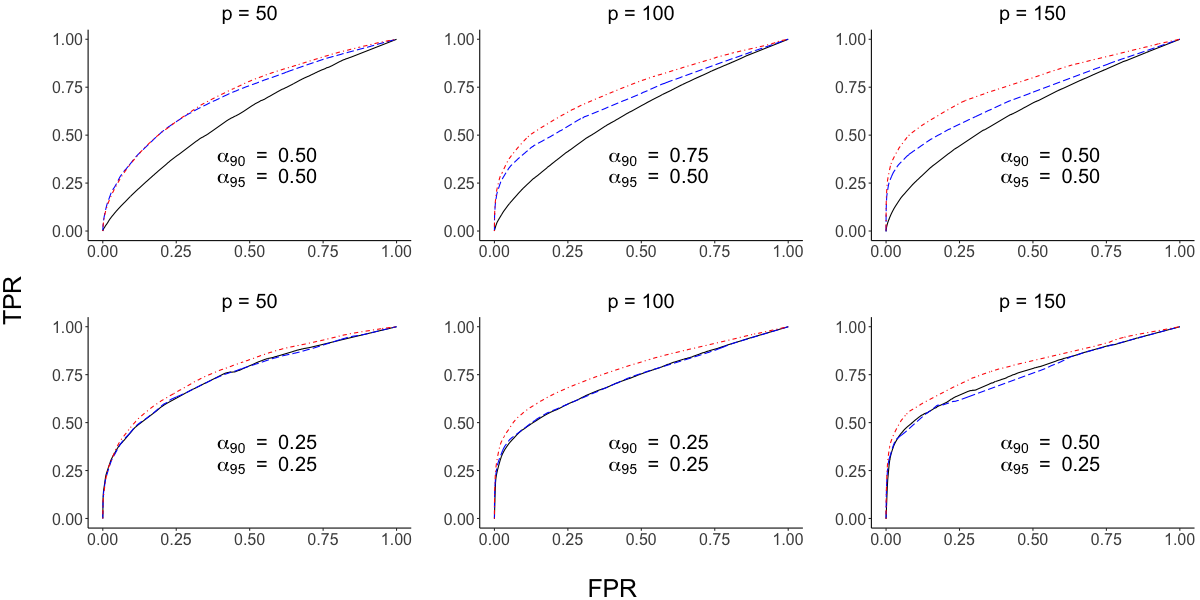}} %
\subcaptionbox{$n = 1.5p$\label{fig:ROC2}}{\includegraphics[height=0.26\textheight]{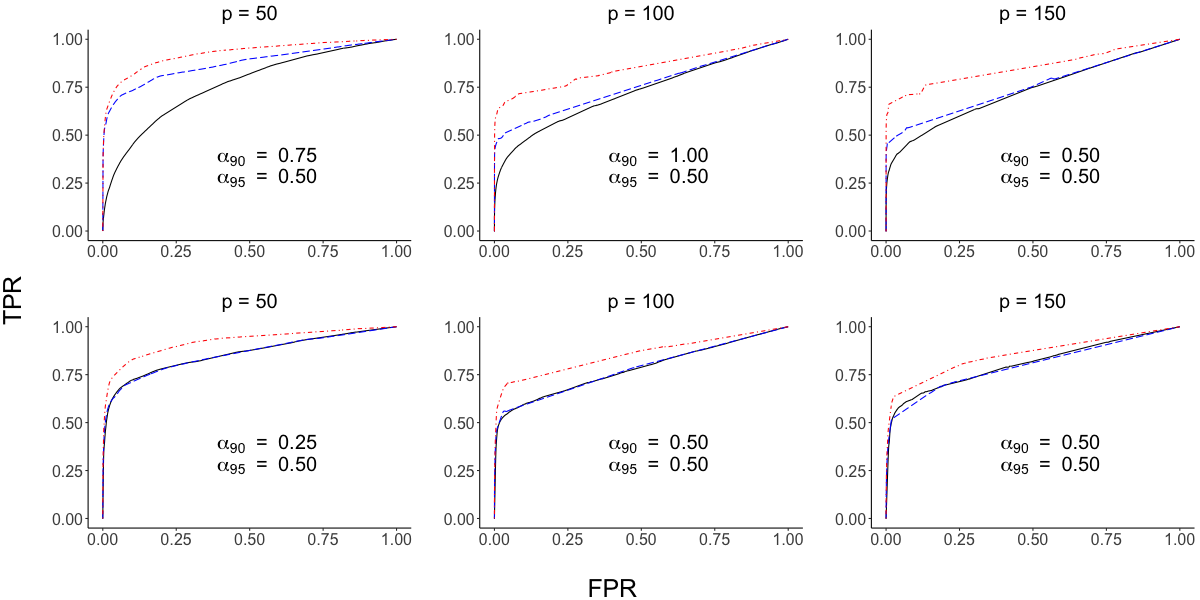}} %
\caption{ Mean receiver operating characteristic curves for the proposed method (psFGGM) and that of \ci{qiao:19} (FGGM). In subfigures (a) and (b), $\Sigma_\text{ps}$ (top) and $\Sigma_\text{non-ps}$ (bottom) were used for $p=50,100,150$. Curves are coded as psFGGM (\textcolor{blue}{\sampleline{dashed}}) and FGGM (\sampleline{}) at $90\%$ of variance and psFGGM (\textcolor{red}{\sampleline{dash pattern=on .7em off .2em on .05em off .2em}}) at $95\%$ of variance explained. For psFGGM, the values of $\alpha$ used to compute the curve values are printed in each panel. }
\label{fig:fig:ROC1_and_2}
\end{figure}

Figure~\ref{fig:ROC1} shows average true/false positive rate curves for the high-dimensional case $n=p/2$. The smoothed curves are computed using the \texttt{supsmu} R package that implements SuperSmoother \cp{friedman:84}, a variable bandwidth smoother that uses cross-validation to find the best bandwidth. Table $\ref{tab:AUC12}$ shows the mean and standard deviation of area under the curve estimates for various settings.  When partial separability holds, $\Sigma = \Sigma_\text{ps}$, the proposed method exhibits uniformly higher true positive rates across the full range of false positive rates. Even when partial separability is violated, $ \Sigma = \Sigma_\text{non-ps}$, the two methods perform comparably. More importantly, and in all cases, the proposed method is able to leverage $95\%$ level of variance explained, owing to the numerical stability mentioned above. Figure $\ref{fig:ROC2}$ and Table $\ref{tab:AUC12}$ summarize results for the large sample case $n=1.5p$ with similar conclusions.  Comparisons under additional simulation settings can be found in Section~\ref{s-sec: sim} of the Appendix.

\begin{table}[t]\centering \small
\caption{{Mean area under the curve (and standard error) values for Figures $\ref{fig:ROC1}$ and $\ref{fig:ROC2}$}}
\label{tab:AUC12}
{\small
\begin{tabular}{ccc|ccc|ccc} 
\hline
& & & \multicolumn{3}{|c|}{$\Sigma = \Sigma_\text{ps}$} & \multicolumn{3}{c}{$\Sigma = \Sigma_\text{non-ps}$}\\
$n$ &  &  & $p = 50$ & $p = 100$ & $p = 150$ & $p = 50$ & $p = 100$ & $p = 150$ \\
\hline
\hline

$p/2$& AUC & $\text{FGGM}_\text{90\%}$ & 0.60(0.03) &  0.62(0.02) & 0.63(0.01)& 0.75(0.03) & 0.72(0.02) & 0.75(0.02)\\
&  & $\text{psFGGM}_\text{90\%}$& 0.71(0.04) & 0.69(0.02)  & 0.70(0.01) & 0.75(0.03) & 0.73(0.02)  &  0.74(0.03)\\
&  & $\text{psFGGM}_\text{95\%}$ & 0.72(0.04)  &  0.74(0.02) &  0.77(0.02)  &0.77(0.03) & 0.78(0.02) & 0.79(0.02) \\
\cline{2-9} 
& $\text{AUC}15^\dagger$ & $\text{FGGM}_\text{90\%}$ &  0.15(0.04)  & 0.18(0.02) & 0.20(0.01) & 0.39(0.04) & 0.40(0.02)& 0.45(0.03) \\
& & $\text{psFGGM}_\text{90\%}$ & 0.30(0.05) & 0.35(0.02)  & 0.37(0.02)  & 0.39(0.04) & 0.42(0.03) & 0.44(0.04) \\
& &$\text{psFGGM}_\text{95\%}$ & 0.29(0.05) & 0.40(0.03) & 0.46(0.03) &  0.41(0.05) & 0.48(0.03) & 0.51(0.03)\\
\hline
$1.5p$&AUC & $\text{FGGM}_\text{90\%}$ & 0.76(0.02) & 0.72(0.02) & 0.73(0.01) & 0.86(0.02)& 0.78(0.02) & 0.80(0.03)\\
& & $\text{psFGGM}_\text{90\%}$ & 0.87(0.03) & 0.75(0.02)  &  0.75(0.01) & 0.85(0.02) &  0.78(0.02) & 0.79(0.03) \\
& & $\text{psFGGM}_\text{95\%}$ & 0.92(0.02) & 0.84(0.02) & 0.85(0.02) & 0.92(0.03) & 0.85(0.02) & 0.85(0.02)\\
\cline{2-9} 
&$\text{AUC}15^\dagger$ & $\text{FGGM}_\text{90\%}$ &  0.37(0.04)& 0.41(0.02) & 0.44(0.02) & 0.66(0.03) &  0.55(0.03) & 0.57(0.04) \\
&& $\text{psFGGM}_\text{90\%}$ & 0.69(0.04) & 0.52(0.02) & 0.52(0.02)& 0.65(0.04)& 0.56(0.04) & 0.55(0.05)\\
&&$\text{psFGGM}_\text{95\%}$ & 0.75(0.04) & 0.68(0.03) & 0.69(0.03) & 0.76(0.06) & 0.68(0.03) & 0.64(0.03) \\
\hline

\end{tabular}}
{\footnotesize$\dagger$AUC15 is AUC computed for FPR in the interval [0, 0.15], normalized to have maximum area 1.}
\end{table}

\section{Application to Functional Brain Connectivity}
\label{sec: app}

In this section, the proposed method is used to reconstruct the brain connectivity structure using functional magnetic resonance imaging (fMRI) data from the Human Connectome Project.  We analyze the ICA-FIX preprocessed data variant that controls for spatial distortions and alignments across both subjects and modalities \cp{glasser:13}.  In particular, we use the 1200 Subjects 3T MR imaging data available at \url{https://db.humanconnectome.org} that consists of fMRI scans of individuals performing basic body movements. During each scan, a three-second visual cue signals the subject to move a specific body part, which is then recorded for 12 seconds at a temporal resolution of 0.72 seconds. For this work, we considered only the data from left- and right-hand finger movements.


The left- and right-hand tasks data for $n=1054$ subjects with complete meta-data were preprocessed by averaging the blood oxygen level dependent signals over $p=360$ regions of interest (ROIs) \cp{glasser:16}. After removing cool down and ramp up observations, $T=16$ time points of pure movement tasks remained. Motivated by part 3 of Theorem \ref{thm: PSequiv}, diagnostics were performed to assess the plausibility of the partial separability assumption, with no indications to the contrary; see Section~\ref{s-sec: app} of the Appendix for more details. Penalty parameters $\gamma=0.91$ and $\alpha=0.95$ were used to estimate very sparse graphs in both tasks.

\begin{figure}[H]
\subcaptionbox{Left-hand task\label{fig:left_hand}}{\includegraphics[width=0.5\textwidth]{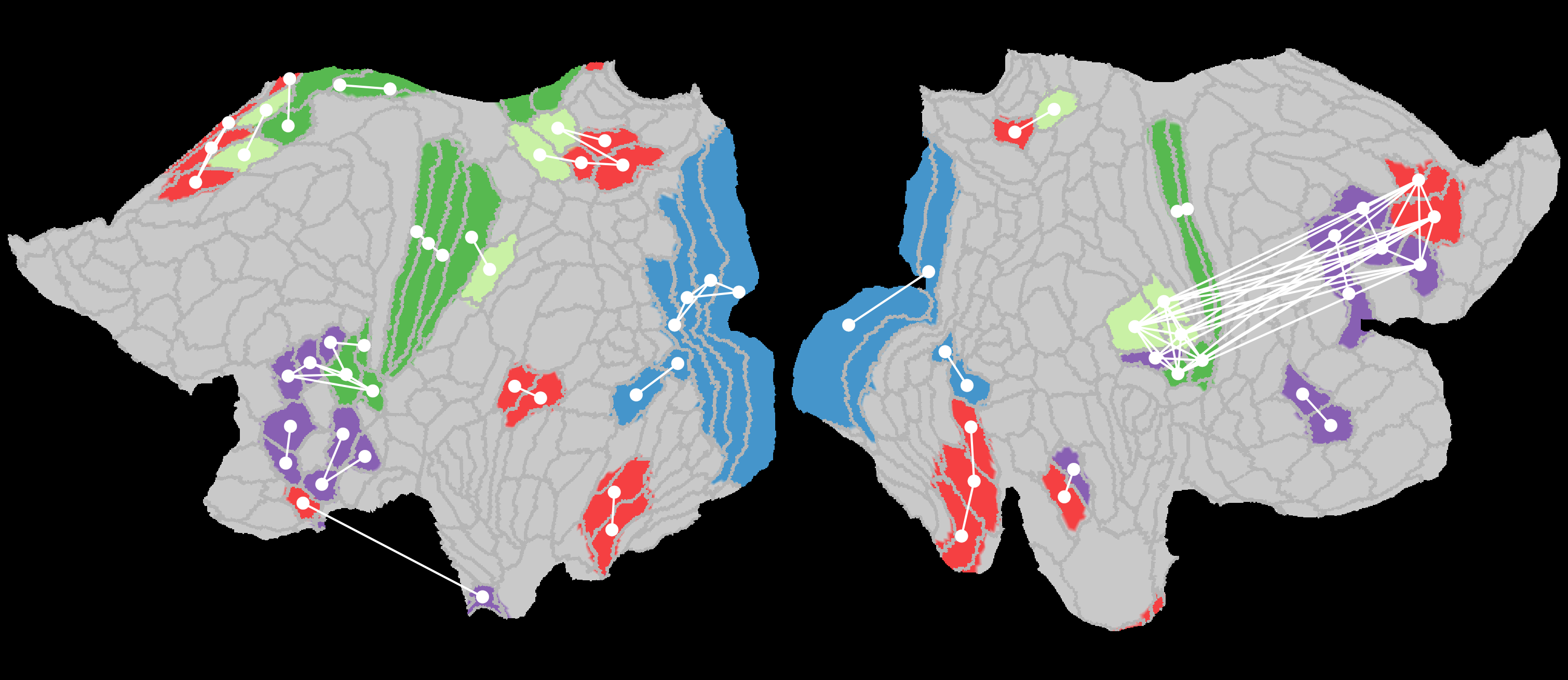}}%
\subcaptionbox{Right-hand task\label{fig:right_hand}}{\includegraphics[width=0.5\textwidth]{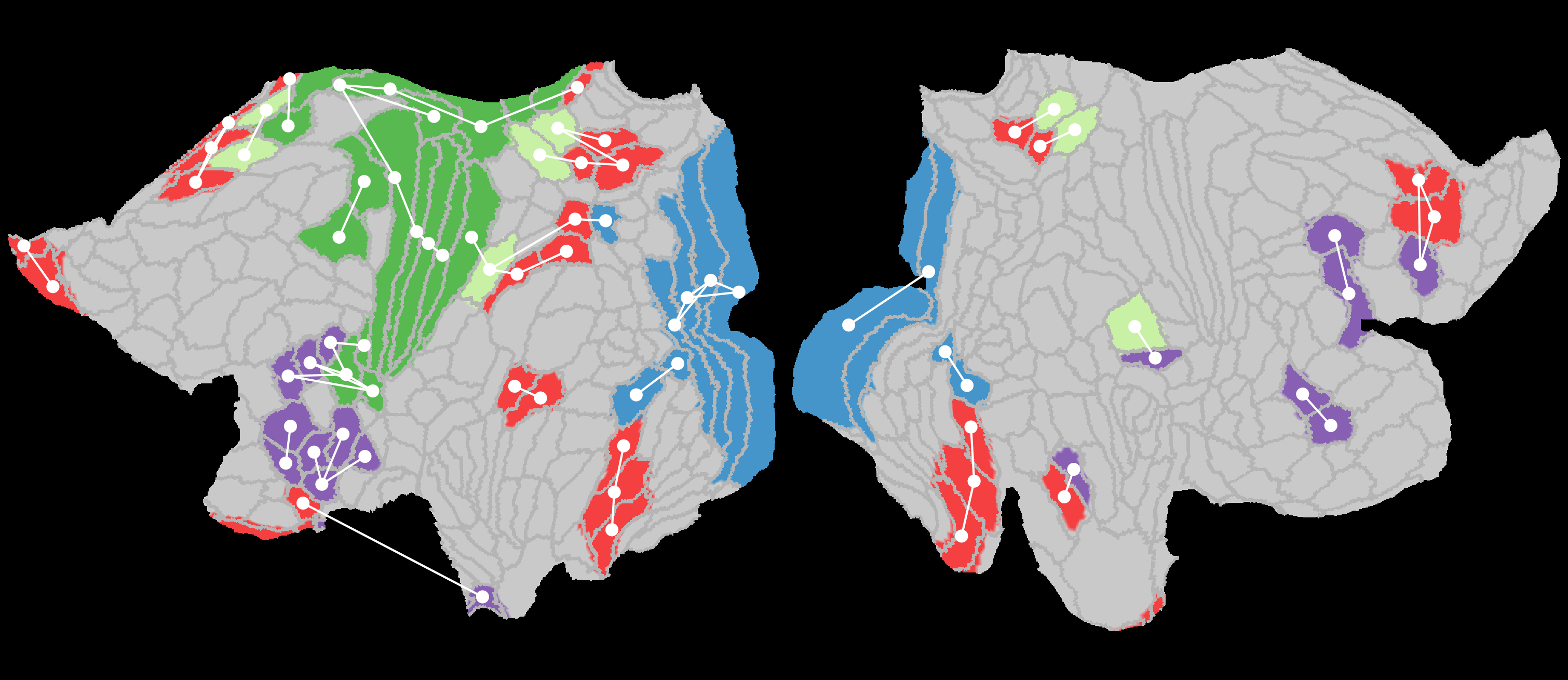}}%
\hfill
\subcaptionbox{Activated ROIs unique to left-hand task\label{fig:lh_unique}}{\includegraphics[width=0.5\textwidth]{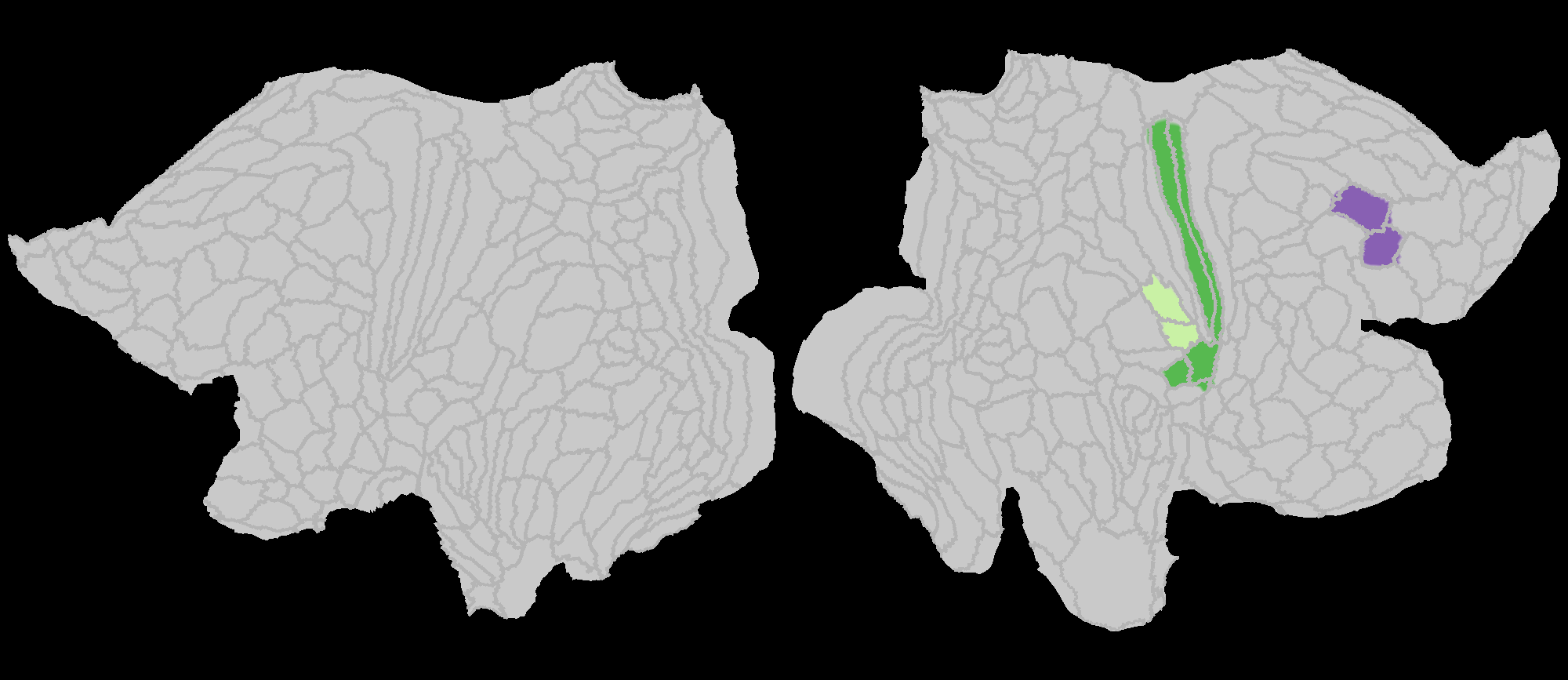}}%
\hfill
\subcaptionbox{Activated ROIs unique to right-hand task\label{fig:rh_unique}}{\includegraphics[width=0.5\textwidth]{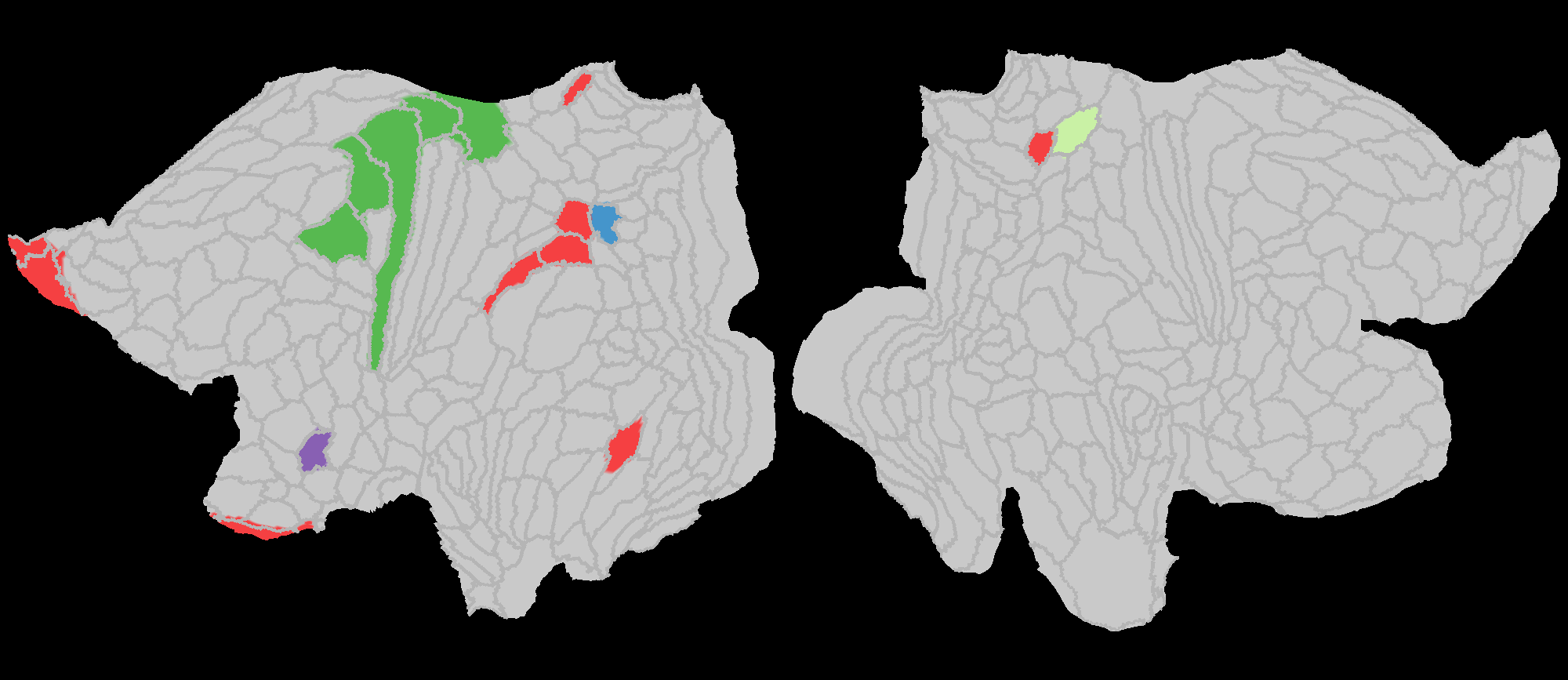}}%
\hfill
\subcaptionbox{Activated ROIs common to both tasks\label{fig:rhlhcommonROIs}}{\includegraphics[width=0.5\textwidth]{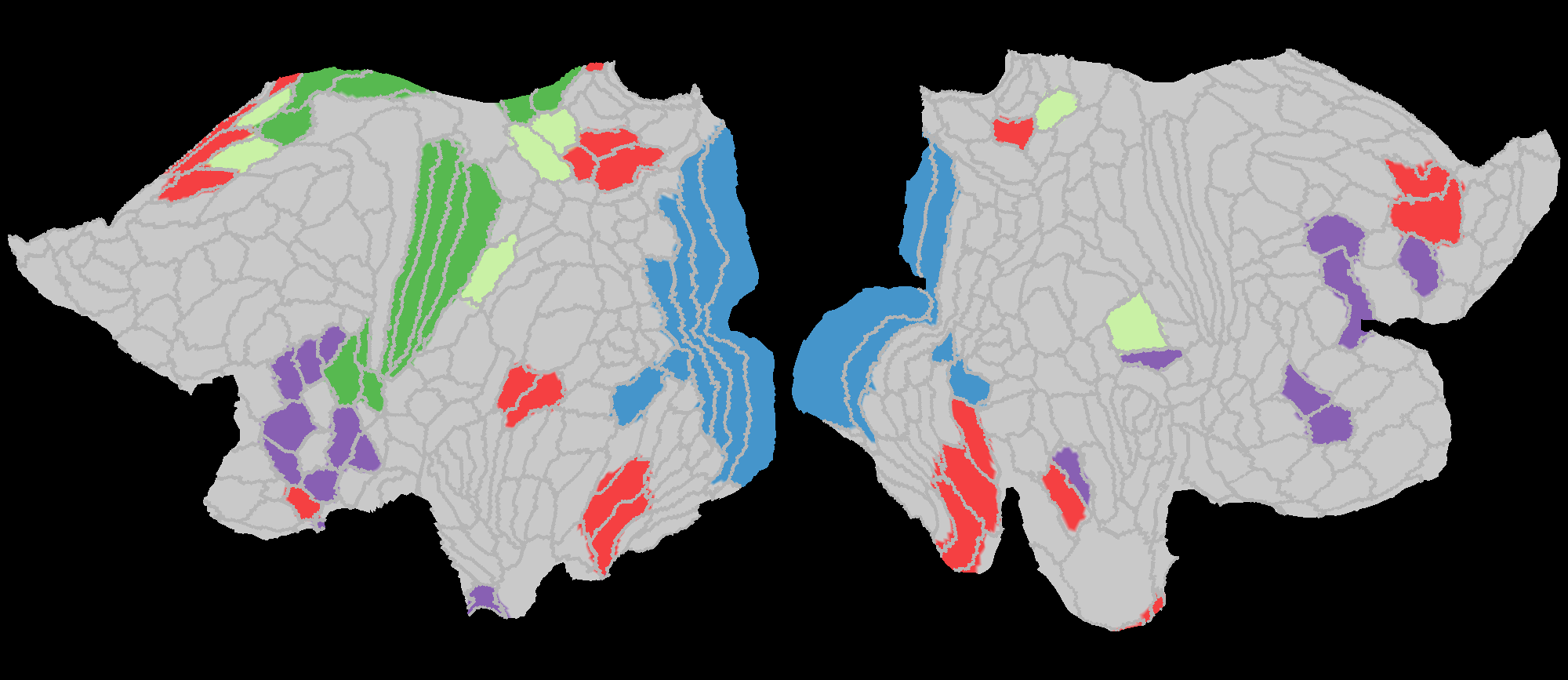}}%
\hfill
\subcaptionbox{ROI task activation map \ci{glasser:16}\label{fig:allROIs}}{\includegraphics[width=0.5\textwidth]{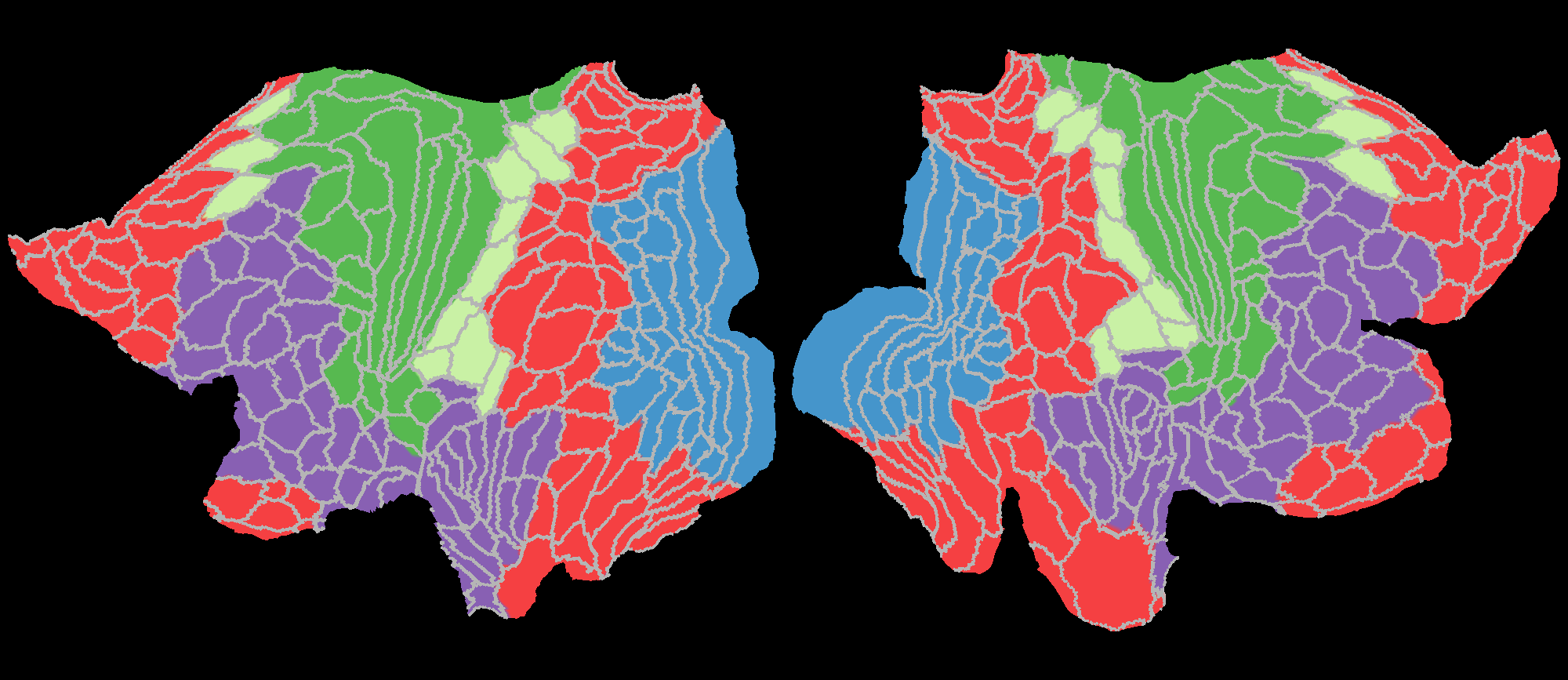}}%
\caption{psFGGM estimated functionally connected cortical ROIs for the left- and right-hand motor tasks. Each sub-figure shows a flat brain map of the left and right hemispheres (in that order). ROIs having a positive degree of connectivity in each estimated graph are colored based on their functionality \ci{glasser:16}: visual (\HCPblue{\textbf{blue}}),  motor (\HCPgreen{\textbf{green}}), mixed motor (\HCPlightgreen{\textbf{light green}}), mixed other (\HCPred{\textbf{red}}) and other (\HCPpurple{\textbf{purple}}). 
}
\label{fig:motorTask1}
\end{figure}

Figure \ref{fig:motorTask1} shows comparison of activation patterns from left and right-hand task datasets. Figures $\ref{fig:left_hand}$ and $\ref{fig:right_hand}$ show the recovered ROI graph on a flat brain map, and only those ROIs with positive degree of connectivity are colored. Figures $\ref{fig:lh_unique}$ and $\ref{fig:rh_unique}$ show connected ROIs that are unique to each task, whereas Figure $\ref{fig:rhlhcommonROIs}$ show only those that are common to both tasks. In this map, one can see that almost all of the visual cortex ROIs in the occipital lobe are shared by both maps. This is expected as both tasks require individuals to watch visual cues. 
Furthermore, the primary sensory cortex, corresponding to touch and motor sensory inputs, and intraparietal sulcus, corresponding to perceptual motor coordination, are activated during both left and right-hand tasks.
On the other hand, the main difference between these motor tasks lies at the motor cortex near the central sulcus. In Figure $\ref{fig:lh_unique}$ and $\ref{fig:rh_unique}$ the functional maps for the left- and right-hand tasks present particular motor-related cortical areas in the right and left hemisphere, respectively. These results are in line with the motor task activation maps obtained by \ci{barch:13}.

\section{Discussion}

Partial separability for multivariate functional data is a novel structural assumption with further potential applications beyond graphical models.  For example, it is well-known that the functional linear model is simplified by univariate functional principal component analysis, which parses out the problem into a sequence of simple linear regressions; see \ci{mull:16:3} and references therein.  The partially separable Karhunen-Lo\`eve expansion in \eqref{eq: PS_KL} demonstrates a similar potential, namely to break down a problem involving multivariate functional data into a sequence of standard multivariate problems.  This potential was shown in this paper by decomposing a functional graphical model into a sequence of standard multivariate graphical models.  

Motivated by the brain connectivity example, we have presented partial separability for multivariate processes with components $X_j$ defined on the same domain.  However, this restriction is not necessary in order to define partial separability.  If $X_j$ are elements of $L^2(\mathcal{T}_j),$ $j = 1,\ldots,p,$ a more general definition of partial separability would be the existence of orthonormal bases $\{\varphi_{jl}\}_{l = 1}^\infty$ of $L^2(\mathcal{T}_j),$ $j = 1,\ldots,p,$ such that the vectors
$
\theta_l = \left(\theta_{l1},\ldots,\theta_{lp}\right)\T,$ where $\theta_{lj} = \int_{\mathcal{T}_j} X_j(t)\varphi_{jl}(t)\dt
$
are mutually uncorrelated across $l$.  Such a generalization is highly desirable, as many multivariate functional data sets consist of functions on different domains. In fact, the above notion is even applicable when the domains $\mathcal{T}_j$ are of different dimension \cp{happ:18} or even a complex manifold, such as the surface of the brain. 
 
The proposed method for functional graphical model estimation is equally applicable to dense or sparse functional data, observed with or without noise.  However, rates of convergence will inevitably suffer as observations become more sparse or are contaminated with higher levels of noise.  The results in Theorem~\ref{thm: concIneq} of this paper or Theorem~1 of \ci{qiao:19} have been derived under the setting of fully observed functional data, so that future work will include similar derivations under more general observation schemes.


\section*{Acknowledgement}
The authors thank Scott Grafton for an interesting discussion of our neuroscience data analysis results, as well as two anonymous referees and an associate editor for their helpful suggestions.
\newpage 

\renewcommand\thesection{A\arabic{section}}
\renewcommand\theequation{A\arabic{equation}}
\renewcommand\thetable{A\arabic{table}}
\renewcommand\thefigure{A\arabic{figure}}

\setcounter{section}{0}
\setcounter{equation}{0}
\setcounter{table}{0}
\setcounter{figure}{0}

\section*{Appendix}

Sections \ref{s-sec: notions} and \ref{s-sec: graphEst} of the Appendix contain proofs of theoretical results from the main paper, while Section~\ref{sec: addResult} contains an alternative edge selection consistency result. Section \ref{s-sec: sim} contains details for the algorithm used in Section \ref{sec: sim} and includes additional figures for very sparse graphs. Finally, Section \ref{s-sec: app} assesses the asssumption of partial separability for the neuroimaging data example. 

\section{Proofs for Section ~\ref{sec: notions}}
\label{s-sec: notions}


\noindent{\bf Proof of Theorem ~\ref{thm: PSequiv}.}\medskip

($1 \Leftrightarrow 2$) Suppose 1 holds and let $\lambda_{jl}^{\mc{G}} = \ipp{\mc{G}(e_{lj}\varphi_l)}{e_{lj}\varphi_l}$ be the eigenvalues of $\mc{G},$ and set $\Sigma_l = \sum_{j = 1}^p \lambda_{jl}^\mc{G}e_{lj} e_{lj}\T_.$  Since $(e_{lj}\varphi_l)\otimes_p (e_{lj}\varphi_l) = (e_{lj}e_{lj}\T)\varphi_l \otimes \varphi_l,$
$$
\mc{G} = \suml \sum_{j = 1}^p \lambda_{jl}^{\mc{G}} (e_{lj}\varphi_l)\otimes_p (e_{lj}\varphi_l) = \suml \left(\sum_{j = 1}^p \lambda_{jl}^{\mc{G}} e_{lj}e_{lj}\T\right) \varphi_l \otimes \varphi_l = \suml \Sigma_l \varphi_l \otimes \varphi_l,
$$
and 2 holds.  If 2 holds, let $\{e_{lj}\}_{j = 1}^p$ be an orthonormal basis for $\Sigma_l$. Then
\[
\begin{split}
\mc{G}(e_{lj}\varphi_l) &= \sum_{l' = 1}^\infty \varphi_{l' }\otimes \varphi_{l'}\left\{ \Sigma_{l'}\left(e_{lj}\varphi_l\right)\right\} = \sum_{l' = 1}^\infty \left(\Sigma_{l'}e_{lj}\right)\varphi_{l'}\otimes \varphi_{l'}(\varphi_l) \\
&= (\Sigma_{l}e_{lj})\varphi_l = \lambda_j^{\Sigma_l}e_{lj}\varphi_l,
\end{split}
\]
so $e_{lj}\varphi_l$ are the eigenfunctons of $\mc{G}.$

($2 \Leftrightarrow 3$) If 2 holds, set $\sigma_{ljj} = \left(\Sigma_l\right)_{jj},$ so the expression for $\mc{G}_{jj}$ clearly holds.  Next, for $l \neq l',$
$$
\cov\left(\ip{X_j}{\varphi_l},\ip{X_k}{\varphi_{l'}}\right) = \ip{\varphi_l}{\mc{G}_{jk}(\varphi_{l'})} = \ip{\varphi_l}{\sigma_{l'jk}\varphi_{l'}} = 0,
$$
so 3 holds. If 3 holds, then define $\sigma_{ljk} = \cov(\ip{X_j}{\varphi_l},\ip{X_k}{\varphi_l})$ for $j \neq k,$ and set $(\Sigma_l)_{jk} = \sigma_{ljk}$ ($j,k \in V$).  Then 2 clearly holds.

($1 \Leftrightarrow 4$) Suppose 1 holds.  By Theorem 7.2.7 of \cite{hsin:15}, and since $\sum_{j = 1}^p e_{lj}e_{lj}\T$ is the identity matrix,
$$
X = \suml \sum_{j = 1}^p \ipp{X}{e_{lj}\varphi_l}e_{lj}\varphi_l = \suml \sum_{j = 1}^p \left(\sum_{k = 1}^p e_{ljk}\ip{X_k}{\varphi_l}\right)e_{lj}\varphi_l = \suml \left(\sum_{j=1}^pe_{lj}e_{lj}\T\right)\theta_l \varphi_l,
$$
so that 4 holds.  If 4 holds, let $\Sigma_l$ be the covariance matrix of $\theta_l,$ and $\{e_{lj}\}_{j = 1}^p$ an orthonormal eigenbasis for $\Sigma_l.$  Then
\[
\begin{split}
\mc{G} &= \suml \sum_{l' = 1}^\infty \cov(\theta_l,\theta_{l'})\varphi_l \otimes \varphi_{l'} = \suml \var(\theta_l)\varphi_l \otimes \varphi_l = \suml \left(\sum_{j = 1}^p \lambda_j^{\Sigma_l}e_{lj}e_{lj}\T\right) \varphi_l \otimes \varphi_l \\
&= \suml \sum_{j = 1}^p \lambda_{j}^{\Sigma_l} (e_{lj}\varphi_l)\otimes_p (e_{lj}\varphi_l),
\end{split}
\]
so that $e_{lj}\varphi_l$ are the eigenfunctions of $\mc{G}.$
\medskip

\noindent{\bf Proof of Theorem ~\ref{thm: psOptimal}.}\medskip

To prove 1, for any orthonormal basis $\{\tilde{\varphi}_l\}_{l = 1}^\infty$ of $L^2[0,1],$
$$
\sum_{l = 1}^L \sum_{j = 1}^p  \var\left(\ip{X_j}{\tilde{\varphi}_l}\right) = \sum_{l = 1}^L \sum_{j = 1}^p \ip{\mc{G}_{jj}(\tilde{\varphi}_l)}{\tilde{\varphi}_l} = p\sum_{l = 1}^L \ip{\mc{H}(\tilde{\varphi}_l)}{\tilde{\varphi}_l}.
$$
Because the eigenvalues of $\mc{H}$ have multiplicity one, its eigenspaces are one-dimensional, and equality is obtained if and only if the $\tilde{\varphi}_l$ span the first $L$ eigenspaces of $\mc{H},$ as claimed.

For the second claim, using part 2 of Theorem~\ref{thm: PSequiv}, we have $\mc{H} = p\inv \suml \trace(\Sigma_l)\varphi_l \otimes \varphi_l,$ and $\{\varphi_l\}_{l = 1}^\infty$ is an orthonormal eigenbasis of $\mc{H}.$ Since it was assumed that the eigenvalues of $\mc{H}$ are unique, and the $\Sigma_l$ are assumed to be ordered so that $\trace(\Sigma_l)$ is nonincreasing, this yields the spectral decomposition of $\mc{H}.$

\medskip

\noindent{\bf Proof of Theorem ~\ref{thm: ps-FGGM}.}\medskip

We have
\begin{align*}
&\cov\left\{X_j(s),X_k(t)\mid X_{-(j,k)}\right\} \\
&\hspace{1.5cm} = \cov\left\{\sum_{l = 1}^\infty \theta_{lj}\varphi_l(s), \sum_{l' = 1}^\infty \theta_{l'j} \varphi_{l'}(s)\mid X_{-(j,k)}\right\} \\
&\hspace{1.5cm} = \sum_{l,l' = 1}^\infty \cov\left\{\theta_{lj}, \theta_{l'k} \mid X_{-(j,k)}\right\} \varphi_l(s)\varphi_{l'}(t)\\
&\hspace{1.5cm} = \sum_{l = 1}^\infty \cov\left\{\theta_{lj},\theta_{lk} \mid \theta_{l,-(j,k)}\right\}\varphi_l(s)\varphi_l(t) \\
&\hspace{1.5cm} = \sum_{l = 1}^\infty \tilde{\sigma}_{ljk}\varphi_l(s)\varphi_l(t)
\end{align*}
Convergence of the sum in the last line follows since $\suml \tilde{\sigma}_{ljk}^2 \leq \suml \sigma_{ljj}\sigma_{lkk} < \infty.$
\medskip

\noindent{\bf Proof of Corollary ~\ref{cor: PSedges}.}\medskip

The result follows immediately, since $(j,k) \notin E$ if and only if $\tilde{\sigma}_{ljk} = 0$ for all $l \in \mathbb{N},$ which holds if and only if $(j,k) \notin E_l$ for all $l \in \mathbb{N}.$
\medskip

\begin{proposition}
\label{prop: markov}
Let $\theta_{lj} = \ip{X_j}{\varphi_l},$ and $E_l$ be the edge set of the Gaussian graphical model for $\theta_l.$ Suppose that the following properties hold for each $j,k \in V$ and $l \in \mathbb{N}.$
\begin{itemize}
    \item $E(\theta_{lj}|X_{-(j,k)}) = E(\theta_{lj}| \theta_{l, -(j,k)})$ and $E(\theta_{lj}\theta_{lk}|X_{-(j,k)}) = E(\theta_{lj}\theta_{lk}| \theta_{l, -(j,k)})$
    \item $(j,k) \notin E_l$ and $(j,k)\notin E_l'$ implies $\cov(\theta_{lj},\theta_{l'k}|X_{-(j,k)}) = 0.$
\end{itemize}
Then $E = \bigcup_{l=1}^\infty E_l.$
\end{proposition}
\noindent{\bf Proof of Proposition ~\ref{prop: markov}.}\medskip

In general, we may write $X_j = \suml \theta_{lj}\varphi_l,$ though the coefficients $\theta_{lj} = \ip{X_j}{\varphi_l}$ need not be uncorrelated across $l$ when $\mc{G}$ is not partially separable.  Then, under the first assumption of the proposition,
\[
\begin{split}
    \cov\left\{X_j(s),X_k(t) \mid X_{-(j,k)}\right\} &= \suml \cov\{\theta_{lj},\theta_{lk} \mid X_{-(j,k)}\}\varphi_l(s)\varphi_l(t) \\ & \hspace{1cm} + \sum_{l \neq l'} \cov\{\theta_{lj},\theta_{l'k} \mid X_{-(j,k)}\}\varphi_l(s)\varphi_{l'}(t) \\
    &= \suml \cov\{\theta_{lj}, \theta_{lk} \mid \theta_{l, -(j,k)}\}\varphi_l(s)\varphi_l(t) \\
    &\hspace{1cm} + \sum_{l \neq l'} \cov\{\theta_{lj},\theta_{l'k} \mid X_{-(j,k)}\}\varphi_l(s)\varphi_{l'}(t).
\end{split}
\]
Now, since $\{\varphi_l \otimes \varphi_{l'}\}_{l,l'=1}^\infty$ is an orthonormal basis of $L^2([0,1]^2),$ $(j,k) \notin E$ if and only if all of the coefficients in the above expansion are zero.  Hence, if $(j,k) \notin E,$ we have $\cov\{\theta_{lj},\theta_{lk} \mid \theta_{l, -(j,k)}\} = 0$ for all $l \in \mathbb{N},$ hence $(j,k) \notin \bigcup_{l = 1}^\infty E_l.$ On the other hand, if $(j,k) \notin E_l$ for all $l,$ the second assumption of the proposition implies that $\cov\{\theta_{lj}, \theta_{l'k}\mid X_{-(j,k)}\} = 0$ for all $l \neq l',$ whence all of the coefficients in the above display are zero, and $(j,k) \notin E.$
\medskip

\section{Proof of results from Section~\ref{sec: graphEst}}
\label{s-sec: graphEst}

Recall that $\norm{\cdot}$ is the ordinary norm on $\Lt$.  For a linear operator $\mc{A}$ on $\Lt$ and an orthonormal basis $\{\phi_l\}_{l = 1}^\infty$ for this space, the Hilbert-Schmidt norm of $\mc{A}$ is $\normHS{\mc{A}} = \left(\suml \norm{\mc{A}(\phi_l)}^2\right)^{1/2},$ where this definition is independent of the chosen orthonormal basis.  In particular, for any $f \in \Lt,$ $\normHS{f \otimes f} = \norm{f}^2.$  
\begin{lemma}
\label{lma: cov_concIneq}
Suppose Assumption \ref{asm: subGauss} holds, and let $\hat{\mu}_j$ and $\hat{\mc{G}}_{jk}$ ($j,k \in V$) be the mean and covariance estimates in \eqref{eq: mean_cov_fo} for a sample of fully observed functional data $X_i \sim X.$  Then there exist constants $\tilde{C}_1,\, \tilde{C}_2, \tilde{C}_3 > 0$ such that, for any $ 0 < \delta \leq \tilde{C}_3$ and for all $j,k \in V$,
$$
\pr\left(\normHS{\hat{\mc{G}}_{jk} - \mc{G}_{jk}} \geq \delta\right) \leq \tilde{C}_2 \exp\left(-\tilde{C}_1n\delta^2\right).
$$
\end{lemma}
\noindent{\bf Proof of Lemma ~\ref{lma: cov_concIneq}.}\medskip

Without loss of generality, assume $\mu_j(t) = E\left\{X_{1j}(t)\right\} \equiv 0$ and set $Y_{ijk} = X_{ij}\otimes X_{ik},$ $\overline{Y}_{jk} = n\inv \sum_{i = 1}^n Y_{ijk}.$  Then the triangle inequality implies
\begin{equation}
\label{eq: cov_split}
\pr\left(\normHS{\hat{\mc{G}}_{jk} - \mc{G}_{jk}} \geq \delta\right) \leq \pr\left(\normHS{\overline{Y}_{jk} - \mc{G}_{jk}} \geq \frac{\delta}{2}\right)
+ 2\max_{j \in V} \pr\left(\lVert \hat{\mu}_j \rVert^2 \geq \frac{\delta}{2}\right).
\end{equation}

We begin with the first term on the right-hand side of \eqref{eq: cov_split}, and will apply Theorem 2.5 of \ci{bosq:00}.  Specifically, we need to find $L_1, L_2 > 0$ such that
$$
E\left(\normHS{Y_{ijk} - \mc{G}_{jk}}^b\right) \leq \frac{b!}{2}L_1L_2^{b-2},\quad (b = 2,3,\ldots),
$$
which will then imply that
\begin{equation}
\label{eq: bound1}
\pr\left(\normHS{\overline{Y}_{jk} - \mc{G}_{jk}} \geq \frac{\delta}{2}\right) \leq 2\exp\left(-\frac{n\delta^2}{8L_1 + 4L_2\delta}\right).
\end{equation}
Let $M_j = \sum_{l = 1}^\infty \sigma_{ljj} < M$ and write $X_{ij} = \sum_{l = 1}^\infty \sigma_{ljj}^{1/2}\xi_{ilj}\varphi_l,$ where $\xi_{ilj} = \theta_{ilj}/\sigma_{ljj}^{1/2}$ are standardized random variables with mean zero and variance one, independent across $i$.  Then, for any $b = 2,3,\ldots,$ by Jensen's inequality,
\begin{align*}
\normHS{Y_{ijk} - \mc{G}_{jk}}^b &= \left\{\sum_{l,l' = 1}^\infty \sigma_{ljj}\sigma_{l'kk}\left(\xi_{ilj}\xi_{il'k} - \delta_{ll'}r_{ljk}\right)^2\right\}^{b/2} \\
& = (M_jM_k)^{b/2}\left\{\sum_{l,l' = 1}^\infty \frac{\sigma_{ljj}\sigma_{l'kk}}{M_jM_k}\left(\xi_{ilj}\xi_{il'k} - \delta_{ll'}r_{ljk}\right)^2\right\}^{b/2} \\
&\leq (M_jM_k)^{b/2 - 1}\sum_{l,l'= 1}^\infty \sigma_{ljj}\sigma_{l'kk}\left|\xi_{ilj}\xi_{il'k} - \delta_{ll'}r_{ljk}\right|^b,
\end{align*}
where $\delta_{ll'}$ is the Kronecker delta.  By Assumption~\ref{asm: subGauss}, one has $E(|\xi_{ilj}|^{2b}) \leq 2(2\varsigma^2)^bb! $ where, without loss of generality, we may assume $\varsigma^2 \geq 1$.  The fact that $|r_{ljk}| < 1$ combined with the $C_r$ inequality implies that 
$$
\sup_{l,l'} E\left(\left|\xi_{ilj}\xi_{il'k} - \delta_{ll'}r_{ljk}\right|\right) \leq 2^{b-1}\sup_{l} \bigg\{ E(|\xi_{ilj}|^{2b}) + 1\bigg\} \leq 2^{b+1}\{(2\varsigma^2)^b b!\}.
$$
Thus,
$$
E\left(\normHS{Y_{ijk} - \mc{G}_{jk}}^b\right) \leq \frac{b!}{2}(4M\varsigma^2)^{b-2}(8M\varsigma^2)^2, 
$$
and we can take $L_2 = 4M\varsigma^2$ and $L_1 = 2L_2^2$ in \eqref{eq: bound1}.

By similar reasoning, we can find constants $\tilde{L}_1,\, \tilde{L}_2 > 0$ such that 
$$
E(\lVert X_{1j} \rVert^b) \leq \frac{b!}{2}\tilde{L}_1\tilde{L}_2^{b-2},\quad (b= 2,3,\ldots),
$$
whence
\begin{equation}
\label{eq: bound2}
\pr\left(\lVert \hat{\mu}_j \rVert \geq \frac{\delta}{2} \right) \leq 2\exp\left(-\frac{n\delta^2}{8\tilde{L}_1 + 4\tilde{L}_2\delta}\right)
\end{equation}
Now, setting $\tilde{C}_3\leq 2$ and $0 < \delta < \tilde{C}_3,$ we find that $\pr(\lVert \hat{\mu}_j\rVert^2 \geq \delta/2)$ is also bounded by the right hand side of \eqref{eq: bound2}, since $\delta /2 < 1$.  Hence, choosing $\tilde{C}_1\inv = \max\{8L_1 + 4L_2\tilde{C}_3,8\tilde{L}_1 + 4\tilde{L}_2\tilde{C}_3\}$ and $\tilde{C}_2 = 6,$ \eqref{eq: cov_split}--\eqref{eq: bound2} together imply the result.
\medskip

\noindent{\bf Proof of Theorem ~\ref{thm: concIneq}.}\medskip

Recall that $\sigma_{ljk} = \ip{\mc{G}_{jk}(\varphi_l)}{\varphi_l}$ and $s_{ljk} = \ip{\mc{\hat{G}}_{jk}(\hat{\varphi}_l)}{\hat{\varphi}_l}.$  Thus, by Lemma 4.3 of \ci{bosq:00} and Assumption~\ref{asm: subGauss},
\begin{align}
|s_{ljk} - \sigma_{ljk}| &\leq |\ip{\mc{G}_{jk}(\varphi_l)}{\varphi_l - \hat{\varphi}_l}| + |\ip{\mc{G}_{jk}(\varphi_l - \hat{\varphi}_l)}{\hat{\varphi}_l}| + |\ip{[\mc{G}_{jk} - \hat{\mc{G}}_{jk}](\hat{\varphi}_l)}{\hat{\varphi}_l}| \nonumber \\
&\leq \lVert \mc{G}_{jk}(\varphi_l) \rVert\lVert \varphi_l - \hat{\varphi}_l\rVert + \lVert \mc{G}_{jk}(\varphi_l - \hat{\varphi}_l)\rVert\lVert \hat{\varphi}_l\rVert + \lVert [\mc{G}_{jk} - \hat{\mc{G}}_{jk}](\hat{\varphi}_l)\rVert \lVert \hat{\varphi}_l \rVert \nonumber \\
&\leq 2M\tau_l\normHS{\hat{\mc{H}} - \mc{H}} + \normHS{\hat{\mc{G}}_{jk} - \mc{G}_{jk}} \label{eq: bound3}.
\end{align}
Now, by applying similar reasoning as in the proof of Lemma~\ref{lma: cov_concIneq}, there exist $C_1^*,\, C_2^*, \, C_3^* > 0$ such that, for all $0 < \delta \leq C_3^*,$ 
$$
\pr\left(\normHS{\hat{\mc{H}} - \mc{H}} \geq \delta\right) \leq C_2^* \exp\left\{-C_1^* n \delta^2\right\}.
$$

Next, let $\tau_{\mathrm{min}} = \min_{l \in \mathbb{N}}\tau_l > 0,$ and $\tilde{C}_j$ as in Lemma~\ref{lma: cov_concIneq}. Set $C_3 = \min\{4M\tau_{\mathrm{min}}C_3^*, 2\tilde{C}_3\}$ and observe that $0 < \delta < C_3$ implies that $\delta(4M\tau_l)\inv < C_3^*$ ($l \in \mathbb{N}$) and $\delta/2 < \tilde{C}_3.$ Hence, by applying \eqref{eq: bound3}, when $0 < \delta < C_3,$ for any $l \in \mathbb{N}$ and any $j,k\in V,$
\[
\begin{split}
\pr\left(|s_{ljk} - \sigma_{ljk}| > \delta\right) &\leq \pr\left(\normHS{\hat{\mc{H}} - \mc{H}} > \frac{\delta}{4M\tau_l}\right) + \pr\left(\normHS{\hat{\mc{G}_{jk}} - \mc{G}_{jk}} > \frac{\delta}{2}\right) \\
&\leq C_2^*\exp\{-C_1^*n\delta^2(4M\tau_l)^{-2}\} + \tilde{C}_2\exp\{-\tilde{C}_1n(\delta/2)^2\}.
\end{split}
\]
Hence, setting $C_2 = C_2^* + \tilde{C}_2$ and $C_1 = \min\{C_1^*(4M)^{-2},\tilde{C}_1\min(\tau_{\mathrm{min}}^2, 1)/4\},$ the result holds.

\medskip

\noindent{\bf Proof of Corollary ~\ref{cor: concIneqCor}.}\medskip

Define $\hat{c}_{lj} = \sqrt{s_{ljj}/\sigma_{ljj}}$ and, for $\epsilon \in (0,1),$ the events
$$
A_{lj}(\epsilon) = \left\{|1 - \hat{c}_{lj}| \leq \epsilon\right\} \quad (l \in \mathbb{N}, j \in V).
$$
Note that
$$
|\hat{r}_{ljk} - r_{ljk}| \leq \frac{|s_{ljk} - \sigma_{ljk}|}{\hat{c}_{lj}\hat{c}_{lk}\pi_l} + |1 - (\hat{c}_{lj}\hat{c}_{lk})\inv|.
$$
Suppose $0 < \delta \leq \epsilon.$  Then
\begin{align}
\pr\left(|\hat{r}_{ljk} - r_{ljk}| \geq 2\delta\right) &\leq 2 \max_{j \in V} \pr\{A_{lj}(\epsilon)^c\} + \pr\left[A_{lj}(\epsilon) \cap A_{lk}(\epsilon) \cap \left\{|\hat{r}_{ljk} - r_{ljk}| \geq 2\delta\right\}\right]  \nonumber \\
&\leq 2\max_{j \in V} \pr\{A_{lj}(\epsilon)^c\} + \pr\left\{|s_{ljk} - \sigma_{ljk}| \geq \delta(1-\epsilon)^2\pi_l\right\} \nonumber  \\
&\hspace{1cm} + \pr\left\{|1 - \hat{c}_{lj}\hat{c}_{lk}| \geq \delta(1-\epsilon)^2\right\} \label{eq: bound4}.
\end{align}

We next obtain bounds for the first and last terms of the last line above. Let $C_j$ be as in Theorem~\ref{thm: concIneq}, and $\overline{\pi} = \max_{l \in \mathbb{N}}\pi_l$. If $D_3 = \min(\epsilon,C_3\overline{\pi}\inv)$ and $0 < \delta \leq D_3,$ then
\[
\begin{split}
\pr \{A_{lj}(\epsilon)^c\}&\leq \pr(|1 - \hat{c}_{lj}^2| > \epsilon) = \pr(|s_{ljj} - \sigma_{ljj}| > \epsilon \sigma_{ljj}) \leq \pr(|s_{ljj} - \sigma_{ljj}| > \delta \pi_l) \\
&\leq C_2\exp\{-C_1 n \tau_l^{-2}\pi_l^2\delta^2\}.
\end{split}
\]
Next, for any $a,b,c > 0$ such that $|1-ab| \geq 3c,$ we must have either $|1-a| \geq c$ or $|1-b| \geq c.$ By Theorem~\ref{thm: concIneq},
\[
\begin{split}
     \pr\left\{|1 - \hat{c}_{lj}| > \frac{\delta(1 - \epsilon)^2}{3}\right\}  &\leq \pr\left\{|s_{ljj} - \sigma_{ljk}| > \frac{\delta(1 - \epsilon)^2\pi_l}{3}\right\} \\
    &\leq C_2\exp\left\{-C_1 n \pi_l^2 \delta^2\tau_l^{-2}(1 - \epsilon)^4/9\right\}.
\end{split}
\]
Putting these facts together, \eqref{eq: bound4} becomes
\[
\begin{split}
\pr\left(|\hat{r}_{ljk} - r_{ljk}| > 2\delta\right) &\leq 2C_2\exp\{-C_1n\tau_l^{-2}\pi_l^2\delta^2\} + C_2\exp\{-C_1(1 - \epsilon)^4n\tau_l^{-2}\pi_l^2\delta^2\} \\
&\hspace{1cm} + 2C_2\exp\left\{-C_1 n \pi_l^2 \delta^2\tau_l^{-2}(1 - \epsilon)^4/9\right\}.
\end{split}
\]
Taking $D_3$ as already stated, $D_2 = 5C_2$ and $D_1 = C_1(1-\epsilon)^4/9,$ the result holds.
\medskip

\subsection{Lemma~\ref{lma: estimation} and Proof of Theorem~\ref{thm: selection}}
\label{ss: selProofs}

We introduce some additional notation.  First, let $\norm{\cdot}_E$ denote the usual Euclidean norm on $\R^m$ of any dimension, where the dimension will be clear from context.  Recall that, for a $p\times p$ matrix $\Delta,$ $\left\VERT \Delta\right\VERT_\infty = \max_{j = 1,\ldots,p} \sum_{k = 1}^p |\Delta_{jk}|,$ $\left\VERT \Delta\right\VERT_1 = \left\VERT \Delta\T\right\VERT_\infty,$ and define the vectorized norm $\norm{\Delta}_\infty = \max_{j,k=1,\ldots,p} |\Delta|_{jk}.$  Additionally, for $S \subset V\times V,$ $\Delta_S$ is the vector formed by the elements $\Delta_{jk},$ $(j,k) \in S.$  Finally, with $D_j$ and $m_l$ as in Corollary~\ref{cor: concIneqCor}, define functions
\begin{equation}
    \label{eq: tailFuns}
    \overline{n}(\delta; c) = \frac{\log(D_2 c)}{D_1\delta^2}, \quad \overline{\delta}(n; c) = \left\{\frac{\log(D_2 c)}{D_1 n}\right\}^{1/2} \quad (c,\delta > 0, n \in \mathbb{N}).
\end{equation}

Before proceeding to the results, we describe our primal-dual witness approach as a modification of that of \ci{ravikumar2011high} to account for the presence of the group Lasso penalty in \eqref{eq: penalty}.  Of importance are the sub-differentials of each of the penalty terms in \eqref{eq: penalty}, omitting the tuning parameter factor, evaluated at a generic set of inputs $(\Upsilon_1,\ldots,\Upsilon_L)$.  Let $\upsilon_{ljk} = (\Upsilon_l)_{jk}.$  The sub-differential contains a restricted set of stacked matrices $Z = (Z_l)_{l = 1}^L,$ $Z_l \in \R^{p\times p}$.  For the Lasso penalty, these satisfy
\begin{equation}
    \label{eq: LassoSD}
    (Z_l)_{jk} = \left\{\begin{array}{ll} 0 & \textrm{if } j = k \\ \mathrm{sgn}(\upsilon_{ljk}) & \textrm{if } j \neq k,\, \upsilon_{ljk} \neq 0 \\ \in [-1,1] & \textrm{if } j \neq k, \, \upsilon_{ljk} = 0. \end{array}\right. 
\end{equation}
In the case of the group penalty, define $\upsilon_{\cdot jk} = (\upsilon_{1jk},\ldots,\upsilon_{Ljk})\T$ and $z_{\cdot jk} = \{(Z_1)_{jk},\ldots,(Z_L)_{jk}\}\T$.  Then, for the group penalty, $Z$ must satisfy
\begin{equation}
    \label{eq: groupSD}
    z_{\cdot jk} = \left\{\begin{array}{ll} 0 & \textrm{if } j = k \\ \frac{\upsilon_{\cdot jk}}{\norm{\upsilon_{\cdot jk}}_E} & \textrm{if } j \neq k, \, \norm{\upsilon_{\cdot jk}}_E \neq 0\\ \in \{y \in \R^L: \norm{y}_E \leq 1\} & \textrm{if } j \neq k, \norm{\upsilon_{\cdot jk}}_E = 0.\end{array}\right.
\end{equation}

We construct the so-called primal-dual witness solutions $\{(\tilde{\Xi}_l, \tilde{Z}_l): l = 1,\ldots,L\}$ as follows.

\begin{enumerate}[label=(\alph*)]
    \item With $\overline{E}_l = E_l \cup (1,1) \cup \cdots \cup (p,p)$, define
\begin{align}
\label{eq: Xi_primal}
(\tilde{\Xi}_1, \ldots,\tilde{\Xi}_L) &= \arg \min_{\Upsilon_l \succ0, \Upsilon_l = \Upsilon_l^\T, \Up_{l, \overline{E}_l^c} = 0}  \sum_{l = 1}^L \left\{\trace(\hat{R}_l\Upsilon_L) - \log(|\Upsilon_l|)\right\} \nonumber \\
&\hspace{1cm} + \gamma\left\{ \alpha\sum_{l = 1}^L\sum_{j \neq k} |\upsilon_{ljk}| + (1-\alpha)\sum_{j \neq k}\left(\sum_{l = 1}^L \upsilon_{ljk}^2\right)^{1/2}\right\}
\end{align}

\item Select elements $\tilde{Z}_1$ and $\tilde{Z}_2$ of the Lasso and group penalty sub-differentials evalauated at $(\tilde{\Xi}_1,\ldots,\tilde{\Xi}_L)$, respectively, that satisfy the optimality condition
\begin{equation}
\label{eq: optimal}
\left[\hat{R}_l - \tilde{\Xi}_l\inv + \gamma\left\{\alpha \tilde{Z}_{1l} + (1 - \alpha)\tilde{Z}_{2l}\right\}\right]_{\overline{E}_l} = 0 \quad (l = 1,\ldots,L).
\end{equation}

\item Update
\begin{equation}
    \label{eq: update}
    \left(\tilde{Z}_{1,l}\right)_{jk} = \frac{1}{\gamma\alpha}\left\{\left(\tilde{\Xi}_l\inv\right)_{jk} - \hat{r}_{ljk}\right\}, \quad \left(\tilde{Z}_{2,l}\right)_{jk} = 0, \quad \quad \{(j, k) \in \overline{E}_l^c, \, l = 1,\ldots,L\}.
\end{equation}

\item Verify strict dual feasibility condition
\begin{equation}
    \label{eq: duality}
    \left|\left(\tilde{Z}_{1,l}\right)_{jk}\right| < 1, \quad \{(j, k) \in \overline{E}_l^c, \, l = 1,\ldots,L\}.
\end{equation}

\end{enumerate}

\begin{lemma}
\label{lma: estimation}
Suppose Assumptions~\ref{asm: lambda}--\ref{asm: irrepresentability} hold and that $\gamma = 8 \epsilon_L\inv\overline{\delta}(n;L^{\varrho-1}p^\varrho)$ for some $\varrho > 2.$ If the sample size satisfies the lower bound
\begin{equation}
    \label{eq: nlb1}
    n > \overline{n}\left(\min\left\{\mathfrak{a}_L, \mathfrak{b}_L\right\}; L^{\varrho-1}p^\varrho \right),
\end{equation}
then, with probability at least $1 - (Lp)^{2 - \varrho},$ the bounds
\begin{equation}
\label{eq: lma4conc}
\lVert \hat{\Xi}_l - \Xi_l\rVert_\infty \leq 2 \kappa_{\Psi_l}\left(m_l\inv + 8 \epsilon_L\inv\right)\overline{\delta}(n;L^{\varrho-1}p^\varrho)
\end{equation}
hold simultaneously for $l = 1,\ldots,L.$
\end{lemma}  

\noindent{\bf Proof of Lemma ~\ref{lma: estimation}.}\medskip

Define $D_j$ and $m_l$ as in Corollary~\ref{cor: concIneqCor}.  Observe that $n > \overline{n}(\delta; c)$ implies $\overline{\delta}(n; c) < \delta$. Since \eqref{eq: nlb1} implies that $n >  \overline{n}(D_3 m_l; L^{\varrho - 1}p^{\varrho})$ for each $l$, $1 \leq l \leq L$, we have $m_l\inv \overline{\delta}(n; L^{\varrho - 1} p^\varrho) \leq D_3,$ for $1 \leq l \leq L$. Define
\begin{equation}
\label{eq: Al}
\mc{A}_l = \left\{\norm{R_l - \hat{R}_l}_\infty \leq m_l\inv \overline{\delta}(n; L^{\varrho - 1} p^\varrho)\right\}.
\end{equation}
Applying Corollary~\ref{cor: concIneqCor} with $\delta = m_l\inv\overline{\delta}(n; L^{\varrho - 1}p^\varrho)$ together with the union bound, we obtain $\pr(\mc{A}_l^c) \leq L^{1 - \varrho}p^{2 - \varrho},$ $1 \leq l \leq L,$ so that $\pr\left(\bigcap_{l = 1}^L \mc{A}_l\right) \geq 1 - (Lp)^{2 - \varrho}$.  Let $\{(\tilde{\Xi}_l, \tilde{Z}_l),\, l = 1,\ldots,L\}$ be the primal-dual witness solutions constructed in steps (a)--(d) preceding the lemma statement.  The result in \eqref{eq: lma4conc} will follow once we have established that, on $\bigcap_{l = 1}^L \mc{A}_l,$ we have $\tilde{\Xi}_l = \hat{\Xi}_l$ ($l = 1,\ldots,L$), and that \eqref{eq: lma4conc} holds with $\hat{\Xi}_l$ replaced by $\tilde{\Xi}_l.$

When $\mc{A}_l$ holds, we apply the condition $\gamma = 8\epsilon_L\inv \overline{\delta}(n;\, L^{\varrho - 1}p^{\varrho})$ to conclude that
\begin{equation}
\label{eq: Rlbound}
\norm{\hat{R}_l - R_l}_\infty \leq m_l\inv \overline{\delta}(n; L^{\varrho - 1}p^\varrho) = \left(\frac{\epsilon_L}{m_l\eta_l'}\right)\left\{8\epsilon_L\inv\overline{\delta}(n; L^{\varrho-1}p^\varrho)\right\}\frac{\eta_l'}{8} \leq \frac{\gamma \eta_l'}{8},
\end{equation}
as $\epsilon_L = \min_{l = 1,\ldots,L}m_l\eta_l'.$  Additionally, since $n > \overline{n}(\mathfrak{b}_L; L^{\varrho -1}p^\varrho),$ it must be that
\begin{equation}
\label{eq: rBound}
\begin{split}
2\kappa_{\Psi_l}\left(\norm{\hat{R}_l - R_l}_\infty + \gamma\right) &\leq 2\kappa_{\Psi_l}(m_l\inv + 8\epsilon_L\inv)\overline{\delta}(n; L^{\varrho - 1}p^\varrho) \\
&\leq \frac{2\kappa_{\Psi_l}(m_l\inv + 8\epsilon_L\inv)}{6y_l m_l\max(\kappa_{\Psi_l}^2\kappa_{R_l}^3, \kappa_{\Psi_l}\kappa_{R_l})(m_l\inv + 8\epsilon_L\inv)^2}  \\
&\leq \min\left(\frac{1}{3y_l\kappa_{\Psi_l}\kappa_{R_l}^3}, \frac{1}{3y_l\kappa_{R_l}}\right) 
\end{split}
\end{equation}
whenever $\mc{A}_l$ holds, where the last line follows since $m_l(m_l\inv + 8 \epsilon_L\inv) > 1$.  Hence, the assumptions of Lemma 6 in \cite{ravikumar2011high} are satisfied whenever $\mc{A}_l$ holds, so that
\begin{equation}
\label{eq: unifDev}
    \norm{\tilde{\Xi}_l - \Xi_l}_\infty \leq 2 \kappa_{\Psi_l}(\norm{\hat{R}_l - R_l}_\infty + \gamma) \leq 2\kappa_{\Psi_l}(m_l\inv + 8\epsilon_L\inv)\overline{\delta}(n; L^{\varrho - 1}p^\varrho).
\end{equation}

Define $\mc{W}_l = \tilde{\Xi}_l\inv - R_l + R_l(\tilde{\Xi}_l - \Xi_l)R_l.$  Having established \eqref{eq: rBound} and \eqref{eq: unifDev}, we apply Lemma~5 of \cite{ravikumar2011high} to conclude that
\begin{equation}
\label{eq: Wlbound}    
\begin{split}
\norm{\mc{W}_l}_\infty &\leq \frac{3}{2}y_l\norm{\tilde{\Xi}_l - \Xi_l}_\infty^2 \kappa_{R_l}^3 \\
&\leq 6\kappa_{R_l}^3\kappa_{\Psi_l}^2y_l(m_l\inv + 8 \epsilon_L\inv)^2\left\{\overline{\delta}(n; L^{\varrho-1}p^\varrho)\right\}^2 \\
&= \left\{6 \kappa_{R_l}^3\kappa_{\Psi_l}^2 y_l (m_l\inv + 8 \epsilon_L\inv)^2\overline{\delta}(n; L^{\varrho-1}p^\varrho)\right\}\left(\frac{\epsilon_L}{\eta_l'}\right)\frac{\gamma \eta_l'}{8} \\
&\leq \left\{6 \kappa_{R_l}^3\kappa_{\Psi_l}^2 y_l m_l(m_l\inv + 8 \epsilon_L\inv)^2\overline{\delta}(n; L^{\varrho-1}p^\varrho)\right\}\left(\frac{\epsilon_L}{m_l\eta_l'}\right)\frac{\gamma \eta_l'}{8} \\
&\leq \frac{\gamma \eta_l'}{8}.
\end{split}
\end{equation}
The last line follows by \eqref{eq: nlb1} and because $\epsilon_L \leq m_l\eta_l'.$

Together, \eqref{eq: Rlbound} and \eqref{eq: Wlbound} imply that, when $\bigcap_{l = 1}^L \mc{A}_l$ holds,
$$
\max\{\norm{\hat{R}_l - R_l}_\infty, \norm{\mc{W}_l}_\infty\} \leq \frac{\gamma \eta_l'}{8} \quad (l = 1,\ldots,L).
$$
Following similar derivations to those of Lemma~4 of \ci{ravikumar2011high}, for any $l = 1,\ldots,L$ and $(j, k) \in \overline{E}_l^c$,
\[
\begin{split}
    \left|\left(\tilde{Z}_{1,l}\right)\right|_{jk} &\leq \frac{\eta_l'}{4\alpha} + \frac{1}{\gamma \alpha}\left\VERT\Psi_{l, \overline{E}_l^c \overline{E}_l}\left(\Psi_{l, \overline{E}_l\overline{E}_l}\inv\right)\right\VERT_1 \frac{2\gamma\eta_l'}{8} \\ &\hspace{1cm} + \frac{1}{\alpha}\left\VERT \Psi_{l, \overline{E}_l^c \overline{E}_l}\left(\Psi_{l, \overline{E}_l\overline{E}_l}\inv\right)\right\VERT_1 \left\lVert \alpha\left(\tilde{Z}_{1, l}\right)_{\overline{E}_l} + (1 - \alpha)\left(\tilde{Z}_{2, l}\right)_{\overline{E}_l}\right\lVert_\infty.
\end{split}
\]
Using Assumption~\ref{asm: irrepresentability}, the definitions of the sub-differentials in \eqref{eq: LassoSD} and \eqref{eq: groupSD}, and the fact that $\eta_l' = \alpha - (1 - \eta_l) > 0,$ the bound then becomes
$$
\left|\left(\tilde{Z}_{1,l}\right)\right|_{jk} \leq \frac{\eta_l'(2 - \eta_l)}{4 \alpha} + \frac{1 - \eta_l}{\alpha} < 1,
$$
and strict dual feasibility holds.  Therefore, $\tilde{\Xi}_l = \hat{\Xi}_l$ for each $l$ when $\bigcap_{l = 1}^L \mc{A}_l$ holds.  Together with \eqref{eq: unifDev}, this completes the proof.
\medskip

\noindent{\bf Proof of Theorem ~\ref{thm: selection}.}\medskip

By construction of the primal witness in \eqref{eq: Xi_primal}, it is clear that $(j,k) \notin E_l$ implies $\tilde{\Xi}_{ljk} = 0.$  Under the given constraint on the sample size, we have $\tilde{\Xi}_l = \hat{\Xi}_l$ with probaility at least $1 - (Lp)^{2 - \varrho}$ by Lemma~\ref{lma: estimation}, so that $\hat{E}_l \subset E_l$ ($1 \leq l \leq L$) with at least the same probability.

Furthermore, \eqref{eq: nlb2} implies $n > \overline{n}(\mathfrak{c}_l; L^{\varrho - 1}p^\varrho),$ so that $\overline{\delta}(n; L^{\varrho - 1}p^\varrho) < \xi_{\mathrm{min},l}\left\{4\kappa_{\Psi_l}(m_l\inv + 8 \epsilon_L\inv)\right\}\inv$ ($1 \leq l \leq L$). Hence, for any $(j,k) \in E_l,$
\[
\begin{split}
    |\hat{\Xi}_{ljk}| &\geq |\Xi_{ljk}| - |\hat{\Xi}_{ljk} - \Xi_{ljk}| \\
    &\geq \xi_{\mathrm{min},l} -  2 \kappa_{\Psi_l}\left(\tau_l^{-2}\pi_l^2 + 8 \epsilon_L\inv\right)\overline{\delta}(n;L^{\varrho-1}p^\varrho) \\
    &\geq \xi_{\mathrm{min},l}/2 > 0.
\end{split}
\]
It follows that, with probability at least $1 - (Lp)^{2 - \varrho},$ $E_l \subset \hat{E}_l,$ ($1 \leq l \leq L$), and the proof is complete.
\medskip

\section{Additional Edge Consistency Result}
\label{sec: addResult}

In this section, we prove a second result on edge selection consistency that is not restrictive on the value of the tuning parameter $\alpha.$  As a trade-off for removing this restriction, we obtain the slightly weaker result that $\bigcup_{l = 1}^L \hat{E}_l = \bigcup_{l=1}^L E_l$ with high probability, rather than accurate recovery of each individual edge set simultaneously.  Unlike the proof of Theorem~\ref{thm: selection}, the result deals more explicitly with the group Lasso penalty, and requires an adapted version of the irrepresentability condition.  However, the constraints on the sample size and divergence of the parameter $L$ are slightly weakened as a result.

Recall the definition of $\Psi_l$ from Section~\ref{ss: asym}.  Define the block matrix $\tilde{\Psi} = \{\tilde{\Psi}_{e, e'}\}_{e,e' \in V\times V}$, where $\Psi_{(j,k),(j'k')}$ is an $L\times L$ diagonal matrix with diagonal equal to $\left\{\Psi_{1,(j,k), (j',k')},\ldots,\Psi_{L,(j,k),(j',k')}\right\}.$  Thus, $\tilde{\Psi}$ groups the elements of each of the $\Psi_l$ within the same edge pairs rather than the same basis.
Letting $S = \bigcup_{l = 1}^L \overline{E}_l$, we can define the submatrix $\tilde{\Psi}_{S^cS},$ with row and column blocks indexed by $S^c$ and $S,$ respectively.  Similarly, define $\Psi_{SS}.$   

We next define an alternative operator norm on $\tilde{\Psi}_{S^cS}\tilde{\Psi}_{SS}\inv$ tailored to the group Lasso sub-differential defined in \eqref{eq: groupSD}.  Let $A$ be an $(|S^c|^2L)\times (|S|^2L)$ matrix consisting of $L\times L$ blocks $A_{(j,k), (j',k')}$ that are themselves diagonal.  Where as the norm in Assumption~\ref{asm: irrepresentability} corresponds to the $\ell_\infty/\ell_\infty$ matrix operator norm, due to the more restricted set of matrices in the group Lasso sub-differential, we define the blockwise $\ell_\infty/ \ell_2$ norm
\begin{equation}
    \label{eq: groupNorm}
    \left\VERT A\right\VERT_{\infty,2} = \max_{e \in S^c} \max_{\norm{z_{e'}}_E \leq 1} \left\lVert\sum_{e' \in S} A_{ee'}z_{e'}\right\rVert_E = \max_{e \in S^c} \left\{\sum_{l = 1}^L \left(\sum_{e' \in S} |A_{ee'll}|\right)^2\right\}^{1/2}.
\end{equation}
We require the following group irrepresentability condition.
\begin{assumption}
\label{asm: groupIrr}
   For some $\eta \in (0,1],$ $\left \VERT \tilde{\Psi}_{S^cS}\tilde{\Psi}_{SS}\inv \right\VERT_{\infty, 2} \leq 1 - \eta.$
\end{assumption}

Next, define $\tilde{\kappa}_{\Psi_l} = \left\VERT \left(\Psi_{l,SS}\right)\inv\right\VERT_\infty$, $y = \max_{j \in V} |\{k \in V: \, \sum_{l = 1}^L \Xi_{ljk}^2 \neq 0\},$ and $\tilde{\xi}_{\mathrm{min}} = \min_{(j,k) \in \bigcup_{l = 1}^LE_l}\left\{\max_{l = 1,\ldots,L} |\Xi_{ljk}|\right\}.$  With $D_j$ as in Corollary~\ref{cor: concIneqCor} and $\tilde{\epsilon}_L = \eta\min_{l = 1,\ldots,L} m_l,$ set
\begin{equation}
    \label{eq: groupABC}
\begin{split}
    \tilde{\mathfrak{a}}_L &= D_3 \min_{l = 1,\ldots,L} m_l, \\
    \tilde{\mathfrak{b}}_L &= (6y)\inv \min_{l = 1,\ldots,L} \left\{m_l(m_l\inv + 8\tilde{\epsilon}_L\inv)^2 \max\left(\tilde{\kappa}_{\Psi_l}^2\kappa_{R_l}^2, \tilde{\kappa}_{\Psi_l}\kappa_{R_l}\right)\right\}\inv \\
    \tilde{\mathfrak{c}}_L &= \tilde{\xi}_{\mathrm{min}}\{4\tilde{\kappa}_{\Psi_l}(m_l\inv + 8\tilde{\epsilon}_L\inv)\}\inv.
\end{split}    
\end{equation}

\begin{theorem}
\label{thm: selection2}
Suppose Assumptions~\ref{asm: lambda}, \ref{asm: subGauss}, and \ref{asm: groupIrr} hold, $L \leq np,$ $\gamma = 8\tilde{\epsilon}_L\inv\left\{(D_1n)\inv\log(D_2L^{\varrho -1}p^\varrho)\right\}^{1/2}$ for some $\varrho > 2$ and that the sample size satisfies the lower bound
\begin{equation}
    \label{eq: nlb3}
    n\min(\tilde{\mathfrak{a}}_L, \tilde{\mathfrak{b}}_L, \tilde{\mathfrak{c}}_L)^2 > D_1\inv\left\{\log(D_2) + (\varrho - 1)\log(n) + (2\varrho - 1)\log(p)\right\}.
\end{equation}
Then, for any $\alpha \in (0,1)$ and with probability at least $1 - (Lp)^{2 - \varrho},$ $\bigcup_{l = 1}^L \hat{E}_l = \bigcup_{l = 1}^L E_l.$
\end{theorem}

Before giving the proof, a few remarks are in order.  For large $n$, the bound in \eqref{eq: nlb3} once again becomes $n\min(\tilde{\mathfrak{a}}_L, \tilde{\mathfrak{b}}_L, \tilde{\mathfrak{c}}_L)^2 \gtrsim \varrho \log(p),$ so that one can achieve selection consistency of the union of the first $L$ edge sets so long as $\log(p) = o(n)$ and $L$ grows slowly enough with $n$.  Second, if we regard $\tilde{\kappa}_{\Psi_l},$ $\kappa_{R_l},$ and $\eta$ to be fixed as $n,p,$ and $L$ diverge, if $\min_{l = 1,\ldots,L} m_l \gtrsim n^{-d}$ for $0 < d < 1/4,$ \eqref{eq: nlb3} becomes
$$
n \gtrsim \left[ \left\{ \tilde{\xi}_{\mathrm{min}}^{-2} + y^2\right\}\varrho\log(p)\right]^{1-4d}.
$$
Compared to the bound given in Remark~\ref{rmk: graphProp} under analagous settings, it is weaker in the first term since $\tilde{\xi}_{\mathrm{min}} > \min_{l = 1,\ldots,L}\xi_{\mathrm{min},l}.$  However, since $y \geq y_l$ for any $l$, the second term here can be more restrictive.  Practically speaking, this is not as much of a concern as it will only be apparent when the individual edge sets $E_l$ all have a much smaller maximal degree than their union.  Finally, one can deduce edge selection consistency if $L$ is capable of growing faster than $\tilde{L}^*_p = \min\left\{L: \bigcup_{l = 1}^L E_l = \bigcup_{l = 1}^\infty E_l\right\}$ while still satisfying $L \leq np$ and $\min(\tilde{\mathfrak{a}}_L, \tilde{\mathfrak{b}}_L, \tilde{\mathfrak{c}}_L)\left\{n/\log(n)\right\}^{1/2} \rightarrow \infty$.

\begin{corollary}
\label{cor: fgm_consistency2}
Under the assumptions of Theorem~\ref{thm: selection2}, suppose $L \rightarrow \infty$ such that $\min(\tilde{\mathfrak{a}}_{L},\tilde{\mathfrak{b}}_L, \tilde{\mathfrak{c}}_L)\{n/\log(n)\}^{1/2} \rightarrow \infty$ and, for large $n,$ $L \geq L_p^*$. Then, for large $n$, $E = \bigcup_{l = 1}^L \hat{E}_l$ with probability at least $1 - (Lp)^{2 - \varrho}.$
\end{corollary}

\noindent{\bf Proof of Theorem ~\ref{thm: selection2}.}\medskip

As the proof follows the same logical flow as that of Theorem~\ref{thm: selection}, we will sketch the proof while outlining major differences.  First of all, steps 1--3 from Section~\ref{ss: selProofs} that detail the construction of the primal/dual witness pairs $(\tilde{\Xi}_l, \tilde{Z}_l)$ are modified as follows.  In the first step, one computes the penalized estimator as in \eqref{eq: Xi_primal} except that one only restricts $\Upsilon_{l,S^c} = 0$.  In step 2, the elements of the sub-differential are chosen to satisfy the optimality condition as in \eqref{eq: optimal}, but over the entire set $S$ rather than $\overline{E}_l.$  In step 3, one updates the sub-differential elements for all $(j,k) \in S^c$ as
$$
\left(\tilde{Z}_{1,l}\right)_{jk} = \left(\tilde{Z}_{2,l}\right)_{jk} = \left(\tilde{Z}_l\right)_{jk} := \frac{1}{\gamma}\left\{\left(\tilde{\Xi}_l\inv\right)_{jk} - \hat{r}_{ljk}\right\} \quad (l = 1,\ldots,L).
$$

With these amendments, one proceeds by showing that, when
$$
b_l := 2\tilde{\kappa}_{\Psi_l}\left(\norm{\hat{R}_l - R_l}_\infty + \gamma\right) \leq \min\left(\frac{1}{3\kappa_{R_l}y}, \frac{1}{3\kappa_{R_l}^3\tilde{\kappa}_{\Psi_l}y}\right),
$$
one has $\norm{\tilde{\Xi}_l - \Xi_l}_\infty \leq b_l.$  This result can be proven using the same logic used in the proofs of Lemmas 5 and 6 of \ci{ravikumar2011high} using the sub-differential properties in \eqref{eq: LassoSD} and \eqref{eq: groupSD}.  Next, one uses the fact that $n > \overline{n}(\tilde{\mathfrak{a}}_L; L^{\varrho -1}p^\varrho)$ to show that, with probability at least $1 - (Lp)^{2 - \varrho},$ the event $\bigcap_{l = 1}^L \mathcal{A}_l$ holds, where $\mathcal{A}_l$ is defined in \eqref{eq: Al}.  The rest of the proof is then conditional on this event.

Using the facts that $\gamma = 8 \tilde{\epsilon}_L\inv\overline{\delta}(n; L^{\varrho - 1}p^\varrho)$ and $n > \overline{n}(\tilde{\mathfrak{b}}_L; L^{\varrho - 1}p^\varrho)$ one can then establish that
$$
\max_{l = 1,\ldots,L} \max\left(\norm{\hat{R}_l - R_l}_\infty, \norm{\mathcal{W}_l}_\infty\right) < \frac{\gamma \eta}{8},
$$
so that 
$$
\norm{\tilde{\Xi}_l - \Xi_l}_\infty \leq 2 \tilde{\kappa}_{\Psi_l}(m_l\inv + 8\tilde{\epsilon}_L\inv)\overline{\delta}(n; L^{\varrho - 1}p^\varrho) \quad (l = 1,\ldots,L).
$$
Then, using arguments similar to Lemma~1 of \ci{ravikumar2011high}, the bound
\[
\begin{split}
    \norm{\{(\tilde{Z}_1)_{jk},\ldots,(\tilde{Z}_L)_{jk}\}^\T}_E &\leq \frac{2}{\gamma}\left(\left\VERT \tilde{\Psi}_{S^cS}\tilde{\Psi}_{SS}\inv\right\VERT_{2,\infty} + 1\right)\left(\frac{\gamma \eta}{8}\right) + \left\VERT \tilde{\Psi}_{S^cS}\tilde{\Psi}_{SS}\inv\right\VERT_{2,\infty} \\
    &\leq \frac{\eta(2 - \eta)}{4} + 1 - \eta < 1
\end{split}
\]
by Assumption~\ref{asm: groupIrr}.  Hence, for each $l = 1,\ldots, L,$ we have $\hat{\Xi}_l = \tilde{\Xi}_l,$ so that $(j, k) \notin \bigcup_{l = 1}^L \hat{E}_l$ for any $(j, k) \in S^c$ and $\bigcup_{l = 1}^L \hat{E}_l \subset \bigcup_{l = 1}^L E_l.$  Finally, using the bound $n > \overline{n}(\tilde{\mathfrak{c}}_L; L^{\varrho - 1}p^\varrho),$ for $(j, k) \in \bigcup_{l = 1}^L E_l,$ one has
\[
\begin{split}
\max_{l = 1,\ldots, L} |\hat{\Xi}_{ljk}| &\geq \tilde{\xi}_{\mathrm{min}} - \max_{l = 1,\ldots,L} \norm{\hat{\Xi}_l - \Xi_l}_\infty  \\
&\geq \tilde{\xi}_{\mathrm{min}} - \max_{l = 1,\ldots,L} 2\tilde{\kappa}_{\Psi_l}(m_l\inv + 8 \tilde{\epsilon}_L\inv)\overline{\delta}(n; L^{\varrho - 1}p^\varrho) \\
&\geq \tilde{\xi}_{\mathrm{min}} - \frac{\tilde{\xi}_{\mathrm{min}}}{2} > 0.
\end{split}
\]
This implies $\bigcup_{l = 1}^L E_l \subset \bigcup_{l = 1}^L \hat{E}_l$, and the proof is complete.

\section{Simulation Details and Additional Figures for Section \ref{sec: sim} }
\label{s-sec: sim}
\subsection{Simulation Details for Section \ref{sec: sim} }

This section describes the generation of edge sets $E_1,\dots,E_M$ and precision matrices $\Omega_1, \dots, \Omega_M$ for the simulation settings in section \ref{sec: sim}.  An initial conditional independence graph $G=(V,E)$ 
is generated from a power law distribution with parameter $\pi = \pr\{ (j,k) \in E\}$. Then, for a fixed $M$, a sequence of edge sets $E_1,\dots,E_M$ is generated so that $E = \bigcup_{l = 1}^M E_l$. This process has two main steps. First, a set of common edges to all edge sets is computed and denoted as $E_c$ for a given proportion of common edges $\tau \in [0,1]$. Next, the set of edges $E\setminus E_c$ is partitioned into $\tilde{E}_1,\dots,\tilde{E}_M$ where $|\tilde{E}_l|\geq|\tilde{E}_{l'}|$ for $l < l'$ and set $E_l = E_c \cup \tilde{E}_l$. The details for this process are described in Algorithm \ref{algo:Algo_edges}.

\begin{algorithm}
\caption{Pseudocode to create the edge sets $E_1,\dots,E_M$}
\label{algo:Algo_edges}
\begin{algorithmic}[1]
  \STATE\enspace Inputs: graph $G=(V,E)$ with nodes $V=\{1,\dots,p\}$ and edge set $E$ \\
  \qquad \qquad \quad \enspace number of basis $M$\\
  \qquad \qquad \quad \enspace proportion of common edges $\tau$\\
  \STATE\enspace Set $E_c \leftarrow \text{random subset of size $\tau |E|$ from $E$}$\\
    \STATE \enspace Set $E_l\leftarrow E_c$ for $l=1,\dots,M$\\ 
    \STATE \enspace Set  $l\leftarrow1$; $B\leftarrow1$\\
    
    \STATE \enspace For $e \in E\setminus E_c$\\
	 \qquad $E_l\leftarrow E_l \bigcup e$\\ 
	 \qquad $l \leftarrow l +1$\\ 
	 \qquad If $(l>B)$ \\
	 \qquad \qquad $l \leftarrow 1$\\
	 \qquad \qquad $B \leftarrow (B+1) \mod M$\\
     
    \STATE \enspace Return $E_1,\dots,E_M$\\
   
\end{algorithmic}
\end{algorithm}

Next, $p \times p$ precision matrices $\Omega_l$ are generated for each $E_l$ based on the algorithm of \citet{peng:09}. Let $\tilde{\Omega}_l$ be a $p \times p$ matrix with entries 
$$ \left(\tilde{\Omega}_l\right)_{jk}=
    \begin{cases}
      1 \quad j=k\\
      0 \quad (j,k) \notin E_l \text{ or } j<k\\
	\sim \mathcal{U}(\mathcal{D}) \quad (j,k) \in E_l
    \end{cases}\quad (j,k=1,\ldots, p)$$
where $\mathcal{D} = [-2/3,-1/3] \cup [1/3,2/3]$. Then, rowwise, we sum the absolute value of the off-diagonal entries and divide each row by 1.5 times this sum componentwise. Finally, $\tilde{\Omega}_l$ is averaged with its transpose and has its diagonal entries set to one. The output is a precision matrix $\Omega_l$ which is guaranteed to be symmetric and diagonally dominant.
\medskip

\subsection{Results for Very Sparse Graphs for Section~\ref{sec: sim}}
\label{secSupp:simverysparse}
A comparison of the two methods under the very sparse case is included. For this case one has $\pi=0.025$ with a proportion of common edges $\tau=0$. 
Finally, we check the robustness of our conclusions under other settings including $\tau \in \{0, 0.1, 0.2\}$ and $\pi \in \{ 2.5\%, 5\%, 10\%\}$, as well as $p$ greater, equal or smaller than $n$.

\begin{table}[H]\centering \small
\caption{{Mean area under the curve (and standard error) values for Figures $\ref{fig:ROC3}$ and $\ref{fig:ROC4}$}}
\label{tab:AUC34}
{\small
\begin{tabular}{ccc|ccc|ccc} 
\hline
& & & \multicolumn{3}{|c|}{$\Sigma = \Sigma_\text{ps}$} & \multicolumn{3}{c}{$\Sigma = \Sigma_\text{non-ps}$}\\
&  &  & $p = 50$ & $p = 100$ & $p = 150$ & $p = 50$ & $p = 100$ & $p = 150$ \\
\hline
\hline

$n = p/2$& AUC & $\text{FGGM}_\text{90\%}$ & 0.61(0.06) & 0.65(0.02) & 0.67(0.02) & 0.78(0.05) & 0.77(0.02) & 0.77(0.02)\\
&  & $\text{psFGGM}_\text{90\%}$ & 0.73(0.05) & 0.75(0.03) & 0.74(0.02) & 0.81(0.05) & 0.79(0.02) & 0.77(0.02)\\
& & $\text{psFGGM}_\text{95\%}$ &  0.70(0.05) &  0.81(0.02)  & 0.82(0.02)  &   0.80(0.05) & 0.85(0.03)&  0.84(0.02)\\
\cline{2-9} 
& $\text{AUC}15^\dagger$ & $\text{FGGM}_\text{90\%}$ & 0.15(0.06) & 0.21(0.03)& 0.26(0.02) &  0.40(0.07) & 0.47(0.03) & 0.51(0.03)\\
&& $\text{psFGGM}_\text{90\%}$ & 0.31(0.08) & 0.44(0.03) & 0.47(0.03) & 0.46(0.08) & 0.52(0.04) & 0.52(0.03) \\
& &$\text{psFGGM}_\text{95\%}$ & 0.27(0.07) & 0.50(0.04) & 0.57(0.03) & 0.42(0.08) & 0.59(0.05) &  0.63(0.04)\\
\hline
$n = 1.5p$& AUC & $\text{FGGM}_\text{90\%}$ & 0.79(0.05) & 0.79(0.02) & 0.77(0.02) & 0.94(0.03) & 0.84(0.02) & 0.79(0.02) \\
& & $\text{psFGGM}_\text{90\%}$ & 0.93(0.03) & 0.82(0.02)  & 0.78(0.02) & 0.96(0.02) & 0.85(0.02) & 0.80(0.02)\\
& & $\text{psFGGM}_\text{95\%}$ & 0.92(0.03) & 0.92(0.02) & 0.87(0.02) & 0.96(0.02) & 0.93(0.02) &0.87(0.02)\\
\cline{2-9} 
&$\text{AUC}15^\dagger$ & $\text{FGGM}_\text{90\%}$ & 0.39(0.08)  & 0.52(0.03) & 0.53(0.02) & 0.82(0.05) & 0.68(0.03) & 0.60(0.03) \\
&& $\text{psFGGM}_\text{90\%}$ & 0.78(0.06) & 0.65(0.03) & 0.58(0.02) & 0.86(0.06) & 0.69(0.04) & 0.60(0.04)\\
&& $\text{psFGGM}_\text{95\%}$ &  0.74(0.07) & 0.83(0.03) & 0.75(0.03) & 0.86(0.05) & 0.83(0.03) &0.74(0.03)\\
\hline

\end{tabular}}
{\footnotesize$\dagger$AUC15 is AUC computed for FPR in the interval [0, 0.15], normalized to have maximum area 1.}
\end{table}

\begin{figure}[H]
\centering
\subcaptionbox{$n = p / 2$\label{fig:ROC3}}{\includegraphics[height=0.26\textheight]{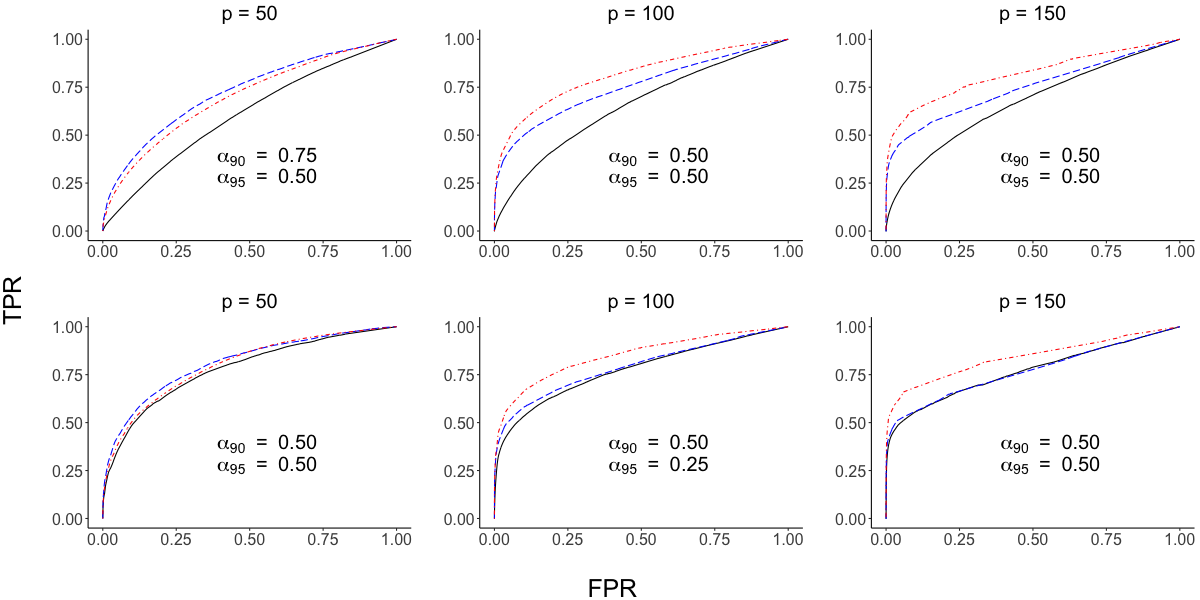}} %
\subcaptionbox{$n = 1.5p$\label{fig:ROC4}}{\includegraphics[height=0.26\textheight]{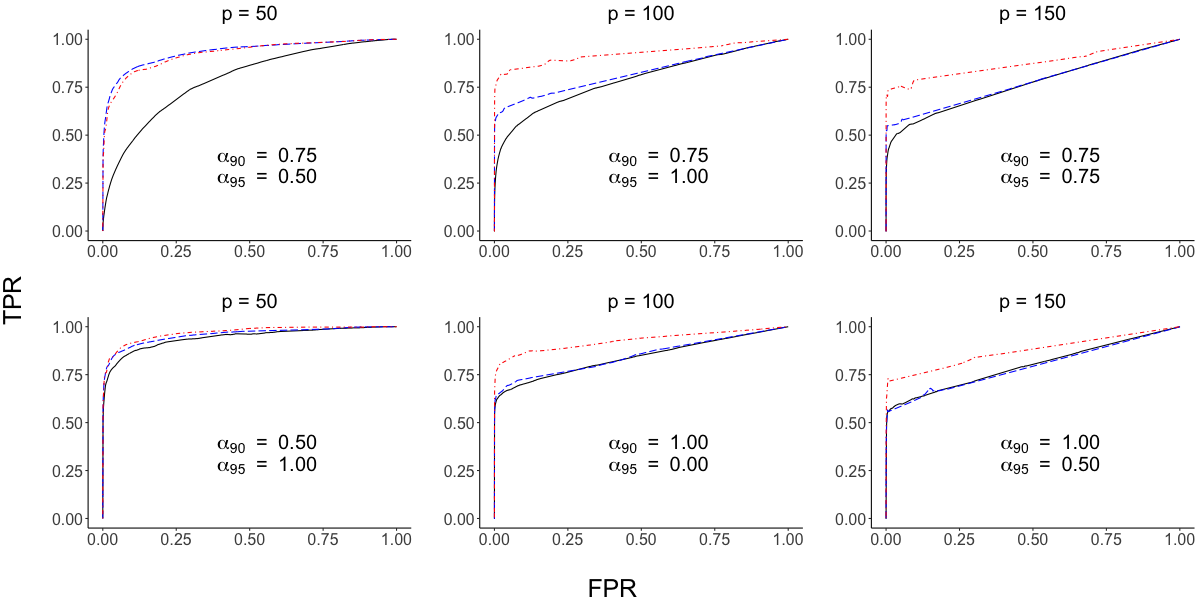}} %
\caption{ Mean receiver operating characteristic curves for the proposed method (psFGGM) with that of \ci{qiao:19} (FGGM) under $\Sigma_\text{ps}$ (top) and $\Sigma_\text{non-ps}$ (bottom) for $p=50,100,150$. We see psFGGM (\textcolor{blue}{\sampleline{dashed}}) and FGGM (\sampleline{}) at $90\%$ of variance and psFGGM (\textcolor{red}{\sampleline{dash pattern=on .7em off .2em on .05em off .2em}}) at $95\%$ of variance explained.
}
\label{fig:fig:ROC3_and_4}
\end{figure}

\subsection{Values of $\alpha$ for psFGGM During Simulations}

\begin{table}[H]\centering \small
\caption{{Values of $\alpha$ used for psFGGM in Figures $\ref{fig:ROC1}$,$\ref{fig:ROC2}$,$\ref{fig:ROC3}$,$\ref{fig:ROC4}$}}
\label{tab:AUC2}
{\small
\begin{tabular}{ccc|ccc|ccc} 
\hline
&& \text{Prop. of Variance}  & \multicolumn{3}{|c|}{$\Sigma = \Sigma_\text{ps}$} & \multicolumn{3}{c}{$\Sigma = \Sigma_\text{non-ps}$}\\
&& \text{Explained}  & $p = 50$ & $p = 100$ & $p = 150$ & $p = 50$ & $p = 100$ & $p = 150$ \\
\hline
\hline

Sparse &$n = 0.5p$&  $\text{90\% }$ & 0.50 & 0.75 & 0.50 & 0.25 & 0.25 & 0.50\\
Case &&  $\text{95\%}$ & 0.50 & 0.50 & 0.50 & 0.25 & 0.25 & 0.25 \\
\cline{2-9} 
&$n = 1.5p$&  $\text{90\% }$ & 0.75 & 1.00 & 0.50 & 0.25 & 0.50 & 0.50\\
&&  $\text{95\%}$ & 0.50 & 0.50 & 0.50 & 0.50 & 0.50 & 0.50 \\
\hline
Very Sparse &$n = 0.5p$&  $\text{90\% }$ & 0.75 & 0.50 & 0.50 & 0.50 & 0.50 & 0.50 \\
Case & & $\text{95\%}$ &  0.50  &  0.50  &  0.50  &  0.50  &  0.25  &  0.50  \\
\cline{2-9} 
&$n = 1.5p$&  $\text{90\% }$ & 0.75 & 0.75 & 0.75 & 0.50 & 1.00 & 1.00\\
&&  $\text{95\%}$ & 0.50 & 1.00 & 0.75 & 1.00 & 0.00 & 0.50 \\
\hline
\end{tabular}}

\end{table}

\subsection{Comparison With An Independence Screening Procedure}
\label{ss: screen}

In this section we compare our method (psFGGM) with another approach meant only for estimating a sparse graph identifying the conditionally independence pairs in a multivariate Gaussian process. This method, denoted as psSCREEN, is based on the sure independence screening procedure of \citet{jianqing:2008}. It assumes partial separability just like psFGGM, but the graph is estimated by thresholding the off-diagonal entries of the matrix $\big[\sum_{l=1}^L \hat{r}_{ljk}^2\big]$ for $j,k = 1, \dots, p$.  Figures \ref{fig:ROC1SCREEN} and \ref{fig:ROC2SCREEN} follow the  sparse case settings of section \ref{sec: sim} with $\pi=0.05$ and $\tau = 0.$ And figures \ref{fig:ROC3SCREEN} and \ref{fig:ROC4SCREEN} follow the very sparse case settings of section \ref{secSupp:simverysparse} with $\pi=0.025$ and $\tau = 0.$

\begin{figure}[H]
\centering
\subcaptionbox{$n = p / 2$\label{fig:ROC1SCREEN}}{\includegraphics[height=0.26\textheight]{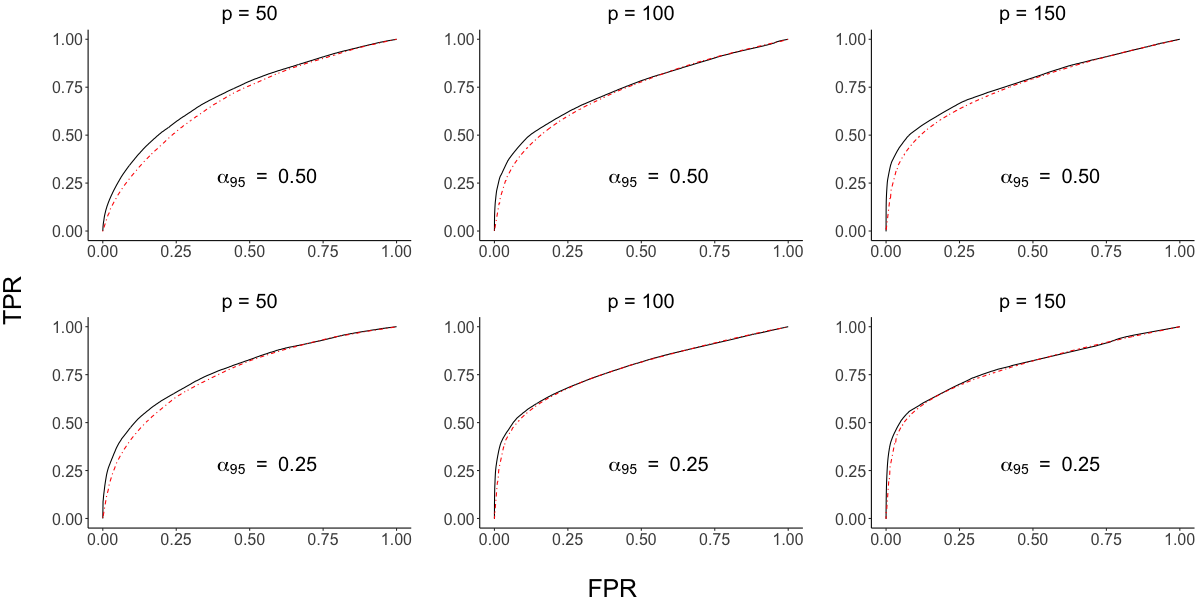}} %
\subcaptionbox{$n = 1.5p$\label{fig:ROC2SCREEN}}{\includegraphics[height=0.26\textheight]{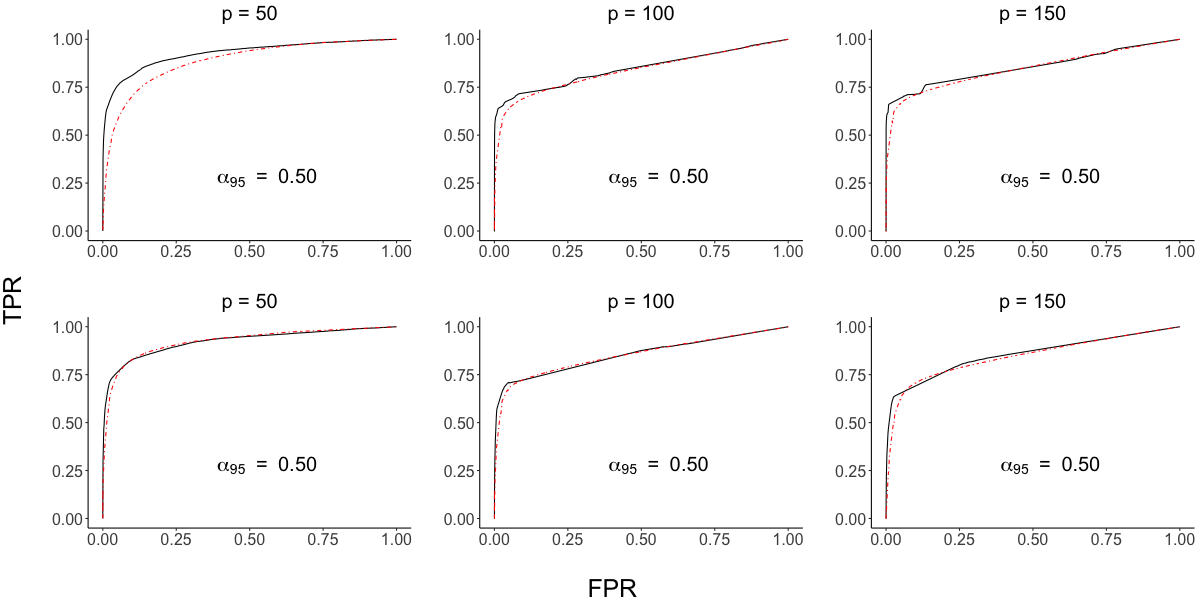}} %
\caption{ Mean receiver operating characteristic curves for the proposed method (psFGGM) and the independence screening procedure (psSCREEN) under $\Sigma_\text{ps}$ (top) and $\Sigma_\text{non-ps}$ (bottom) for $p=50,100,150$. We see psFGGM (\sampleline{}) and psSCREEN(\textcolor{red}{\sampleline{dash pattern=on .7em off .2em on .05em off .2em}}) both at $95\%$ of variance for the sparse case.
}
\label{fig:fig:ROCSCREEN1_and_2}
\end{figure}

\begin{figure}[H]
\centering
\subcaptionbox{$n = p / 2$\label{fig:ROC3SCREEN}}{\includegraphics[height=0.26\textheight]{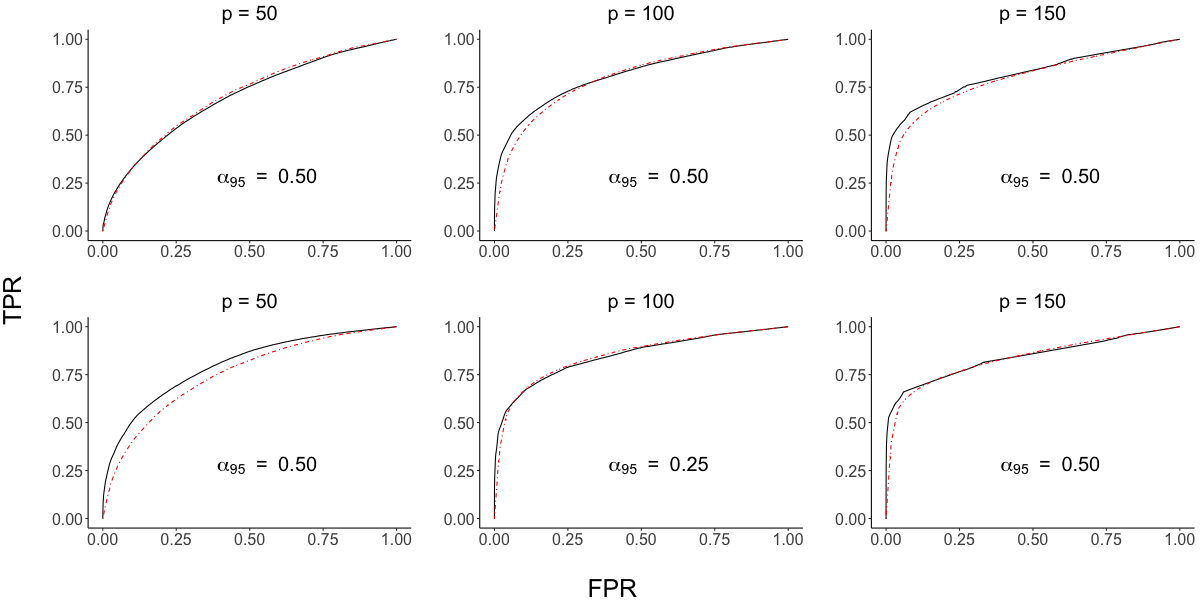}} %
\subcaptionbox{$n = 1.5p$\label{fig:ROC4SCREEN}}{\includegraphics[height=0.26\textheight]{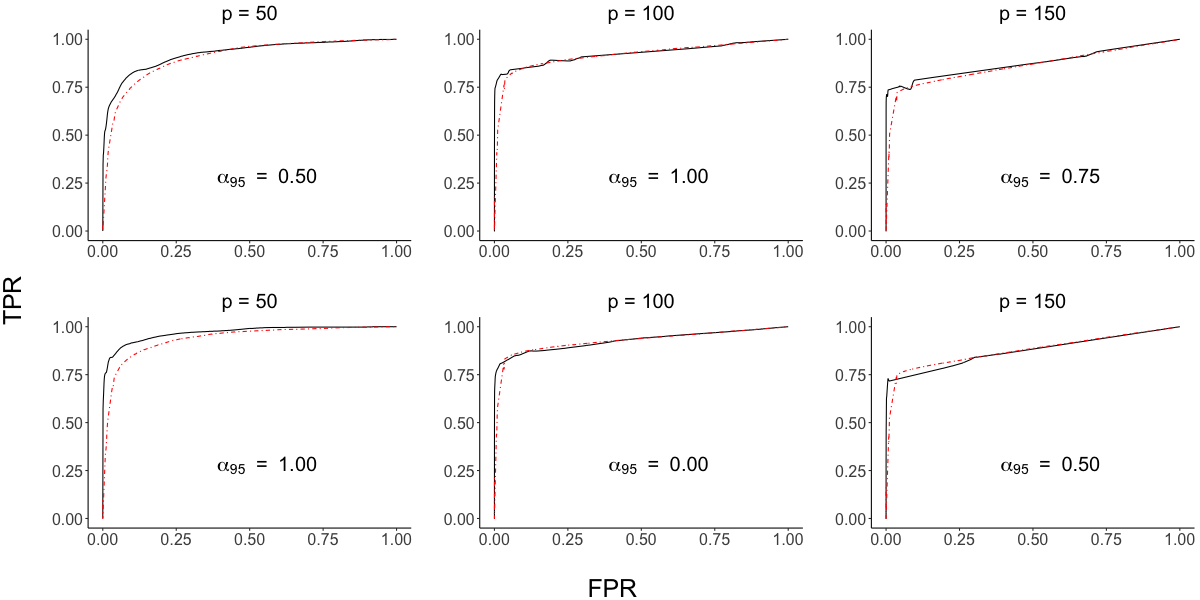}} %
\caption{ Mean receiver operating characteristic curves for the proposed method (psFGGM) and the independence screening procedure (psSCREEN) under $\Sigma_\text{ps}$ (top) and $\Sigma_\text{non-ps}$ (bottom) for $p=50,100,150$. We see psFGGM (\sampleline{}) and psSCREEN(\textcolor{red}{\sampleline{dash pattern=on .7em off .2em on .05em off .2em}}) both at $95\%$ of variance for the very sparse case.
}
\label{fig:fig:ROCSCREEN3_and_4}
\end{figure}

\begin{table}[t]\centering \small
\caption{{Mean area under the curve (and standard error) values for Figures $\ref{fig:ROC1SCREEN}$ and $\ref{fig:ROC2SCREEN}$}}
\label{tab:AUC12SCREEN}
{\small
\begin{tabular}{ccc|ccc|ccc} 
\hline
& & & \multicolumn{3}{|c|}{$\Sigma = \Sigma_\text{ps}$} & \multicolumn{3}{c}{$\Sigma = \Sigma_\text{non-ps}$}\\
$n$ &  &  & $p = 50$ & $p = 100$ & $p = 150$ & $p = 50$ & $p = 100$ & $p = 150$ \\
\hline
\hline
$p/2$& AUC &  $\text{psFGGM}_\text{95\%}$ & 0.72(0.04)  &  0.74(0.02) &  0.77(0.02)  &0.77(0.03) & 0.78(0.02) & 0.79(0.02) \\
&  &  $\text{psSCREEN}_\text{95\%}$ & 0.69(0.04)& 0.73(0.02)& 0.75(0.02)&0.75(0.03)& 	0.78(0.02)& 	0.79(0.02) \\
\cline{2-9} 
& $\text{AUC}15^\dagger$ &$\text{psFGGM}_\text{95\%}$ & 0.29(0.05) & 0.40(0.03) & 0.46(0.03) &  0.41(0.05) & 0.48(0.03) & 0.51(0.03)\\
&  &  $\text{psSCREEN}_\text{95\%}$ & 0.23(0.05)	& 0.34(0.02)	& 0.40(0.02) &	0.34 (0.05) & 0.46(0.03)& 0.49(0.03)\\
\hline
$1.5p$&AUC &$\text{psFGGM}_\text{95\%}$ & 0.92(0.02) & 0.84(0.02) & 0.85(0.02) & 0.92(0.03) & 0.85(0.02) & 0.85(0.02)\\
&  &  $\text{psSCREEN}_\text{95\%}$ & 0.89(0.02)& 0.84(0.02)& 0.85(0.02)& 0.93(0.02)& 0.86(0.02)& 0.85(0.02)  \\
\cline{2-9} 
&$\text{AUC}15^\dagger$ & $\text{psFGGM}_\text{95\%}$ & 0.75(0.04) & 0.68(0.03) & 0.69(0.03) & 0.76(0.06) & 0.68(0.03) & 0.64(0.03) \\
&  &  $\text{psSCREEN}_\text{95\%}$ & 0.61(0.05) & 0.65(0.03) & 0.67(0.03) &	0.76(0.05) & 0.68(0.03)& 0.63(0.03) \\
\hline
\end{tabular}}
{\footnotesize$\dagger$AUC15 is AUC computed for FPR in the interval [0, 0.15], normalized to have maximum area 1.}
\end{table}

\begin{table}[t]\centering \small
\caption{{Mean area under the curve (and standard error) values for Figures $\ref{fig:ROC3SCREEN}$ and $\ref{fig:ROC4SCREEN}$}}
\label{tab:AUC34SCREEN}
{\small
\begin{tabular}{ccc|ccc|ccc} 
\hline
& & & \multicolumn{3}{|c|}{$\Sigma = \Sigma_\text{ps}$} & \multicolumn{3}{c}{$\Sigma = \Sigma_\text{non-ps}$}\\
$n$ &  &  & $p = 50$ & $p = 100$ & $p = 150$ & $p = 50$ & $p = 100$ & $p = 150$ \\
\hline
\hline
$p/2$& AUC &  $\text{psFGGM}_\text{95\%}$ &  0.70(0.05) &  0.81(0.02)  & 0.82(0.02)  &   0.80(0.05) & 0.85(0.03)&  0.84(0.02)\\
&  &  $\text{psSCREEN}_\text{95\%}$ & 0.71(0.05)&	0.80(0.02)	&0.80(0.02)&	0.75(0.04)&	0.85(0.02)&	0.84(0.02)\\
\cline{2-9} 
& $\text{AUC}15^\dagger$ &$\text{psFGGM}_\text{95\%}$ & 0.27(0.07) & 0.50(0.04) & 0.57(0.03) & 0.42(0.08) & 0.59(0.05) &  0.63(0.04)\\
&  &  $\text{psSCREEN}_\text{95\%}$ & 0.26(0.07)	&0.43(0.04)& 0.49(0.03) &	0.32(0.07) & 0.58(0.04) & 0.59(0.03)\\
\hline
$1.5p$&AUC &$\text{psFGGM}_\text{95\%}$ & 0.92(0.03) & 0.92(0.02) & 0.87(0.02) & 0.96(0.02) & 0.93(0.02) &0.87(0.02)\\
&  &  $\text{psSCREEN}_\text{95\%}$ & 0.91(0.03) & 0.92(0.02) & 0.86(0.02) & 	0.94(0.03) & 0.93(0.02) & 0.88(0.02)  \\
\cline{2-9} 
&$\text{AUC}15^\dagger$ & $\text{psFGGM}_\text{95\%}$ &  0.74(0.07) & 0.83(0.03) & 0.75(0.03) & 0.86(0.05) & 0.83(0.03) &0.74(0.03)\\
&  &  $\text{psSCREEN}_\text{95\%}$ & 0.66(0.07)	 & 0.80(0.03) & 0.72(0.03) &	0.77(0.06) & 0.83(0.03) & 0.75(0.03)\\
\hline
\end{tabular}}
{\footnotesize$\dagger$AUC15 is AUC computed for FPR in the interval [0, 0.15], normalized to have maximum area 1.}
\end{table}

\section{Right- and Left-Hand Task Partial Separability Estimates for Section~\ref{sec: app}}
\label{s-sec: app}

Motivated by Theorem \ref{thm: PSequiv}.3 we compare the empirical performance of the partially separable and univariate Karhunen-Lo\`eve type expansions on the right-hand task fMRI curves under different number of eigenfunction estimates. Subjects are randomly assigned into training and validation sets of equal size. The training set is used to estimate the eigenfunctions of each expansion, whereas the validation set is used to compute out-of-sample variance explained percentages. Boxplots are computed on 100 simulations of this procedure. Figure \ref{fig:in-sample-VE-rh} shows that the univariate Karhunen-Lo\`eve exhibits a better in-sample performance, a known optimality property of functional principal component analysis. On the other hand, Figures \ref{fig:out-of-sample-VE-rh} and \ref{fig:ratio-of-sample-VE-rh} show that the partially separable decomposition exhibits a better out-of-sample performance in both absolute terms and in relative comparison to its in-sample performance. Similar conclusion can be obtained for the left-hand task in Figure \ref{fig:motorTask_boxplots_lh}.

\begin{figure}[H]
\subcaptionbox{ In-Sample \label{fig:in-sample-VE-rh}}{\includegraphics[width=0.32\textwidth]{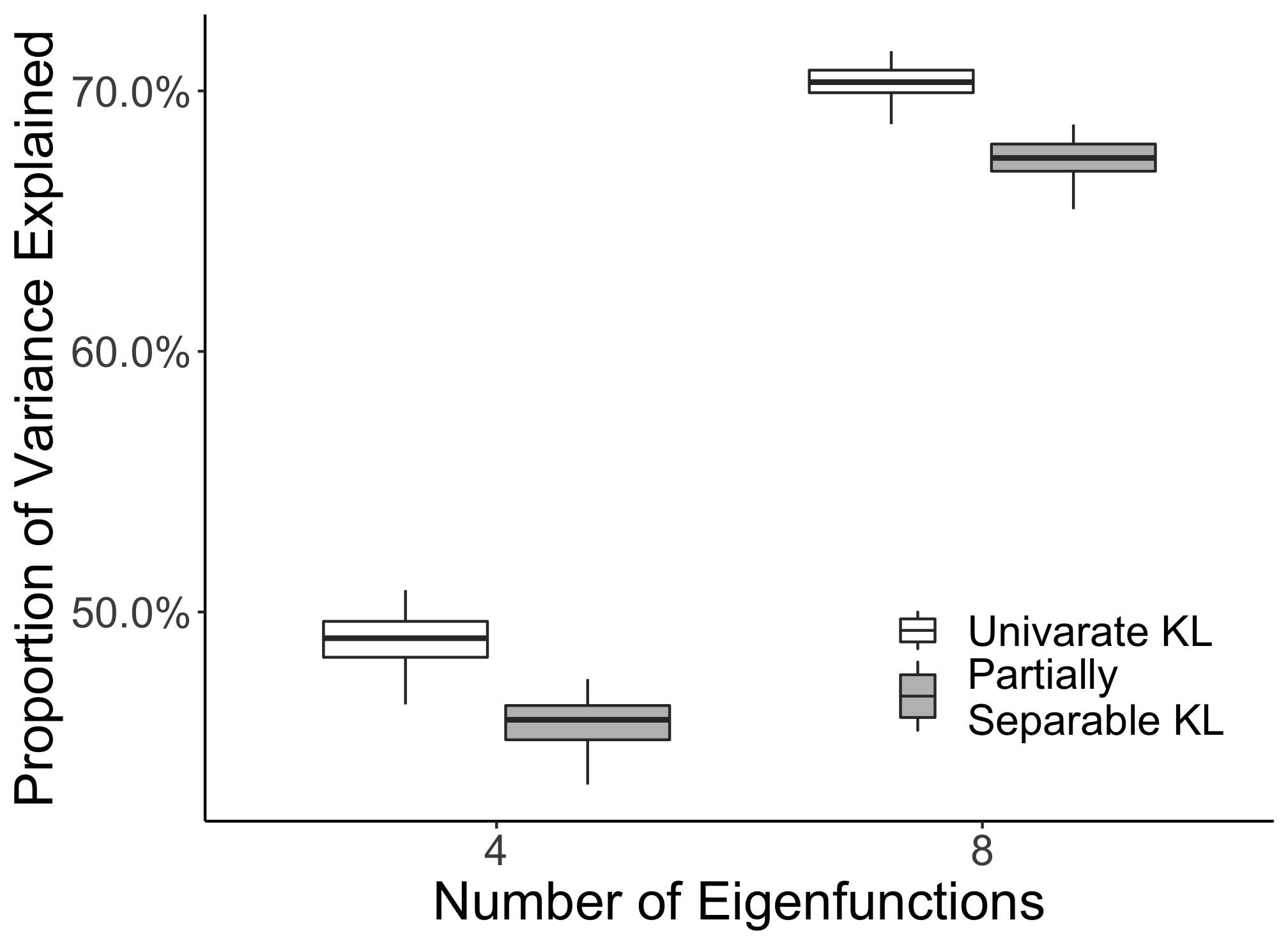}}%
\subcaptionbox{Out-of-Sample \label{fig:out-of-sample-VE-rh}}{\includegraphics[width=0.32\textwidth]{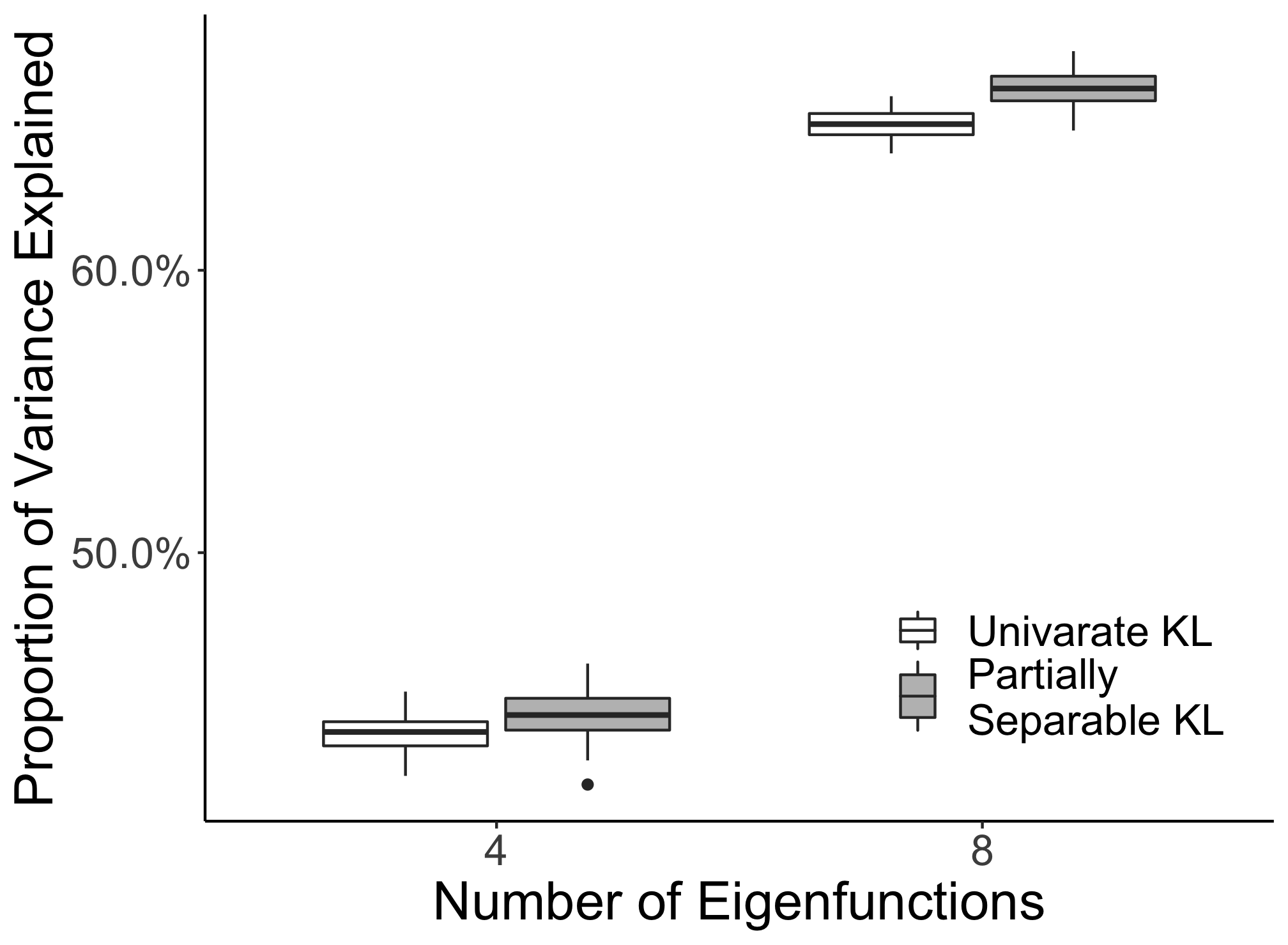}}%
\subcaptionbox{Out-of- to In-Sample Ratio \label{fig:ratio-of-sample-VE-rh}}{\includegraphics[width=0.32\textwidth]{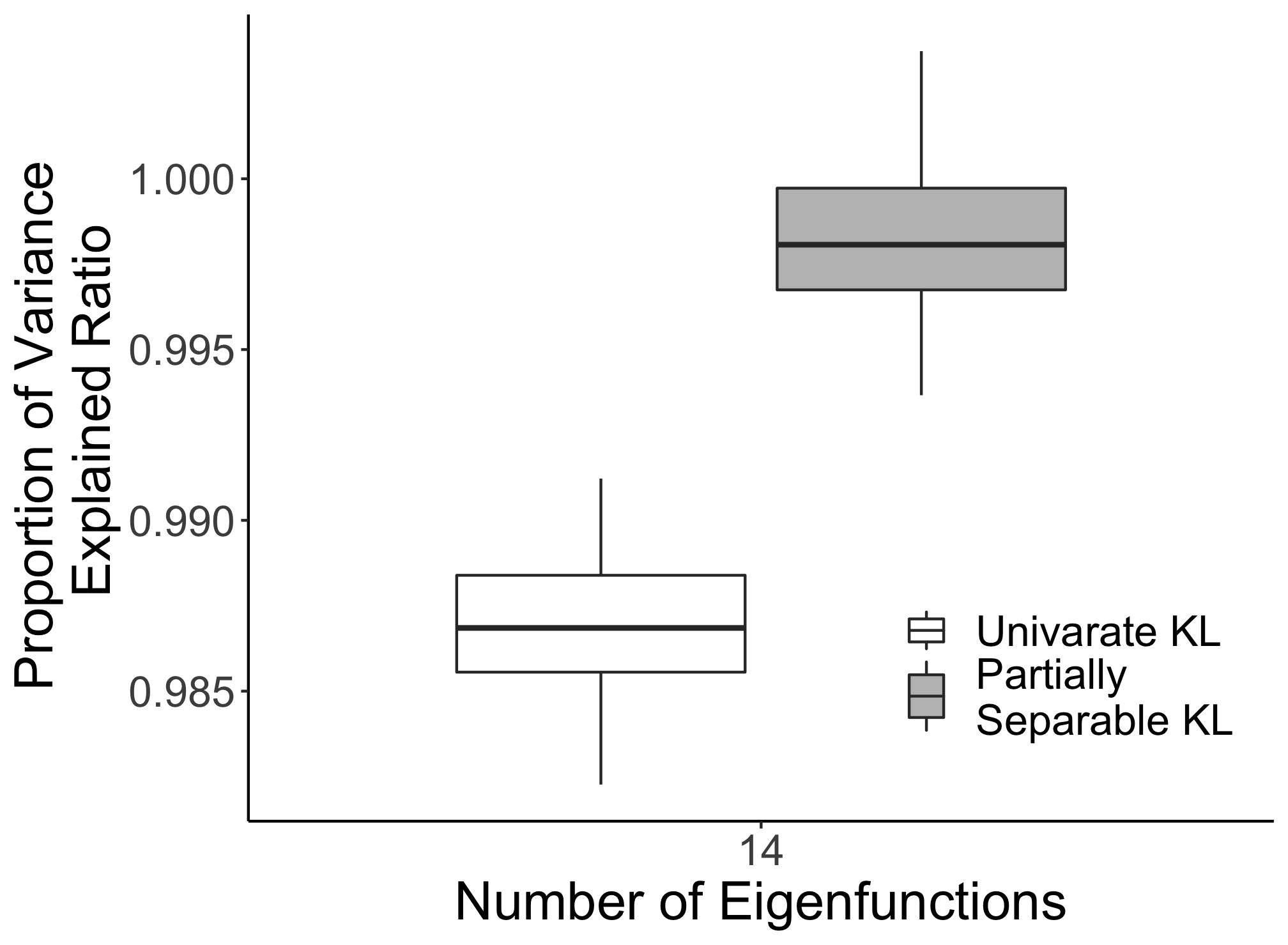}}%
\caption{ Estimated variance explained for partially separable and univariate Karhunen-Lo\`eve type expansions for right-hand task fMRI curves: (a) In-Sample, (b) Out-of-Sample and (c) Out-of- to In-Sample Ratio. 
}
\label{fig:motorTask_boxplots_rh}
\end{figure}

\begin{figure}[H]
\subcaptionbox{ In-Sample \label{fig:in-sample-VE-lh}}{\includegraphics[width=0.32\textwidth]{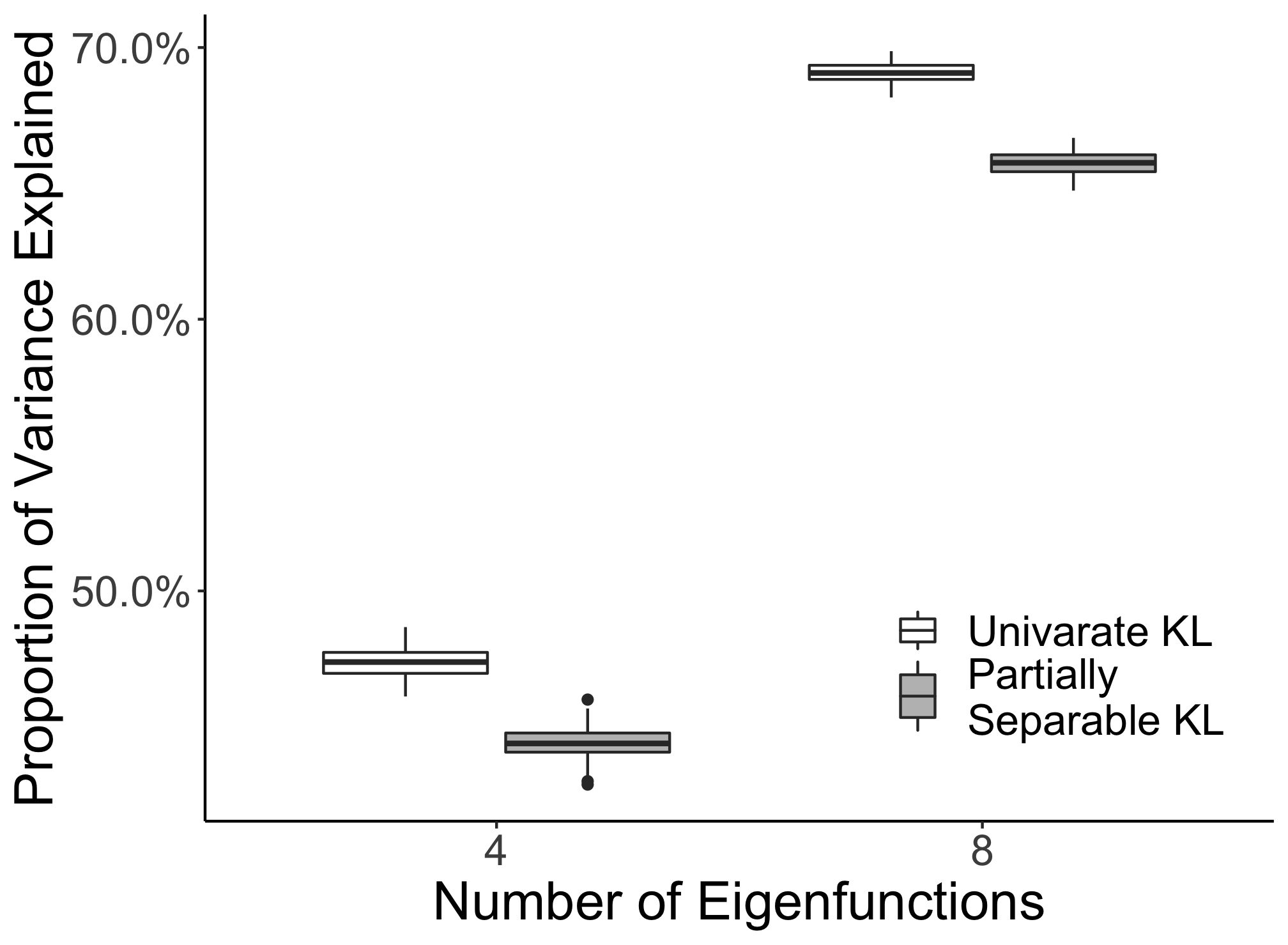}}%
\subcaptionbox{Out-of-Sample \label{fig:out-of-sample-VE-lh}}{\includegraphics[width=0.32\textwidth]{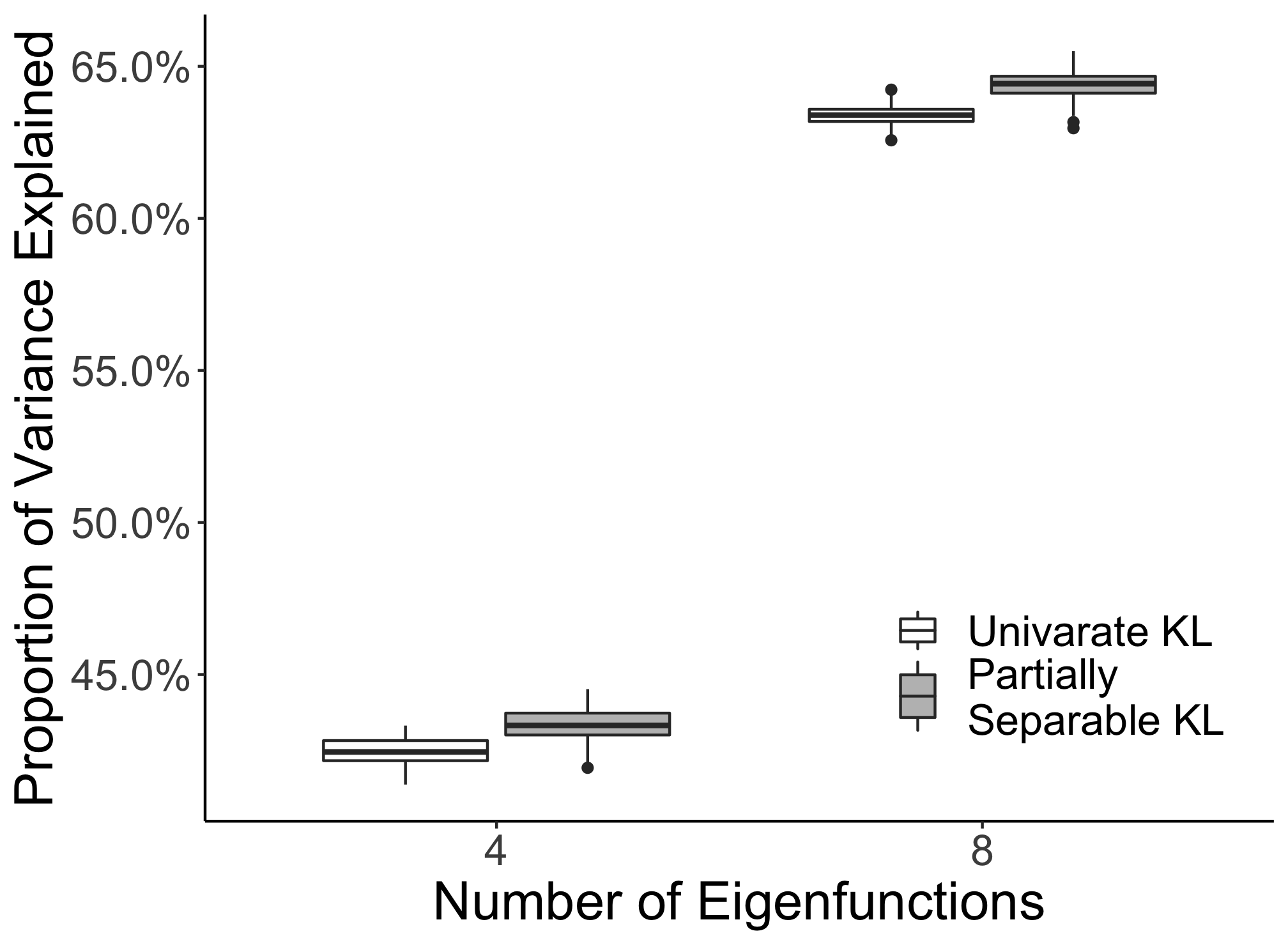}}%
\subcaptionbox{Out-of- to In-Sample Ratio \label{fig:ratio-of-sample-VE-lh}}{\includegraphics[width=0.32\textwidth]{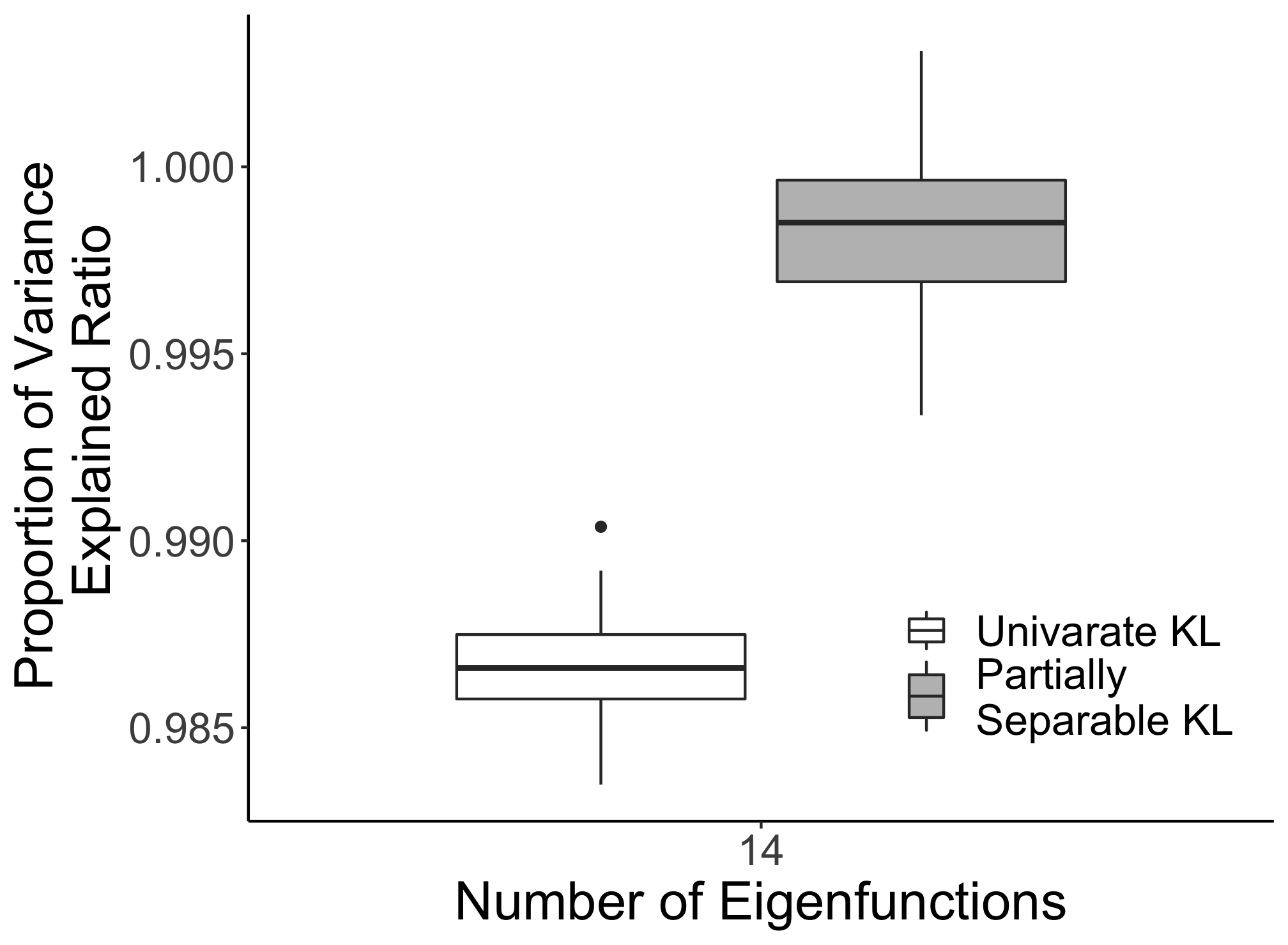}}%
\caption{ Estimated variance explained for partially separable and univariate Karhunen-Lo\`eve type expansions for left-hand task fMRI curves: (a) In-Sample, (b) Out-of-Sample and (c) Out-of- to In-Sample Ratio. 
}
\label{fig:motorTask_boxplots_lh}
\end{figure}

A second way to probe the partial separability assumption is to examine the correlations of the $\theta_{lj}.$ Part 3 of Theorem~\ref{thm: PSequiv} suggest that these should be zero for different $l$.  Figure \ref{fig:motorTask_blockcorrelations_rh} compares the correlation structures of the random coefficients in using both the functional principal component and partial separability bases, on a basis-first order (as illustrated in Figure \ref{fig:Cov_Structure}(c)) using the entire right-hand task dataset. The partially separable expansion coefficients have a similar correlation structure to their univariate Karhunen-Lo\`eve type counterpart, and show that the strongest correlations exist between scores obtained from the same basis function.  Similar conclusions can be drawn for the left-hand task in Figure \ref{fig:motorTask_blockcorrelations_lh}. All in all, no contraindication of the partial separability structure was found in the dataset.

\begin{figure}[H]
\subcaptionbox{\cite{qiao:19} \label{fig:fgm_blockCorr_rh}}{\includegraphics[width=0.5\textwidth]{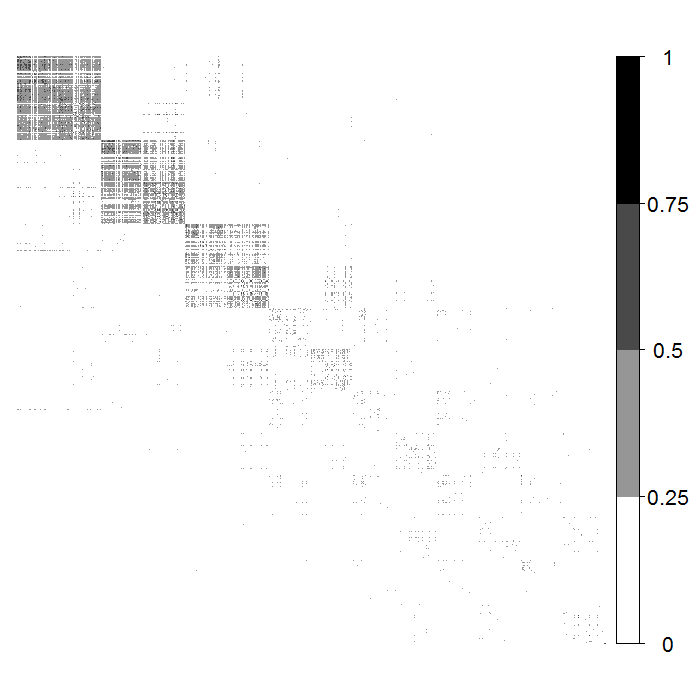}}%
\subcaptionbox{Partially separable \label{fig:psfgm_blockCorr_rh}}{\includegraphics[width=0.5\textwidth]{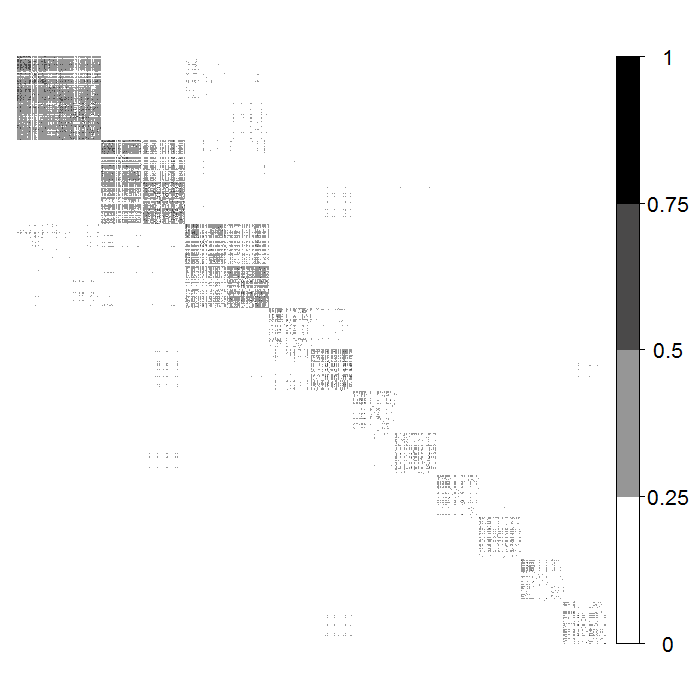}}%
\caption{ Estimated correlation structures of $\mathbb{R}^{Lp}$-valued random coefficients from different $L$-truncated  Karhunen-Lo\`eve type expansions for the right-hand task. The figure shows the upper left 7 x 7 basis blocks of the absolute correlation matrix in basis-first order for: (a) functional principal component coefficients $(\xi_1^T, \dots, \xi_p^T)^T$ in \eqref{eq: basisExp} as in \cite{qiao:19}, and (b) random coefficients $(\theta_1\T,\ldots,\theta_L\T)\T$ under partial separability in \eqref{eq: PS_KL}.}
\label{fig:motorTask_blockcorrelations_rh}
\end{figure}

\begin{figure}[H]
\subcaptionbox{\cite{qiao:19} \label{fig:fgm_blockCorr_lh_qiao}}{\includegraphics[width=0.5\textwidth]{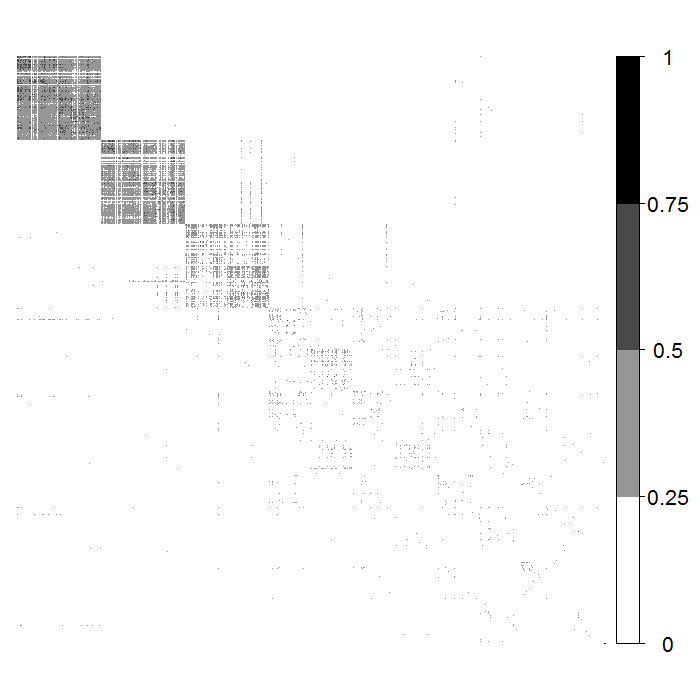}}%
\subcaptionbox{Partially separable \label{fig:fgm_blockCorr_lh_ps}}{\includegraphics[width=0.5\textwidth]{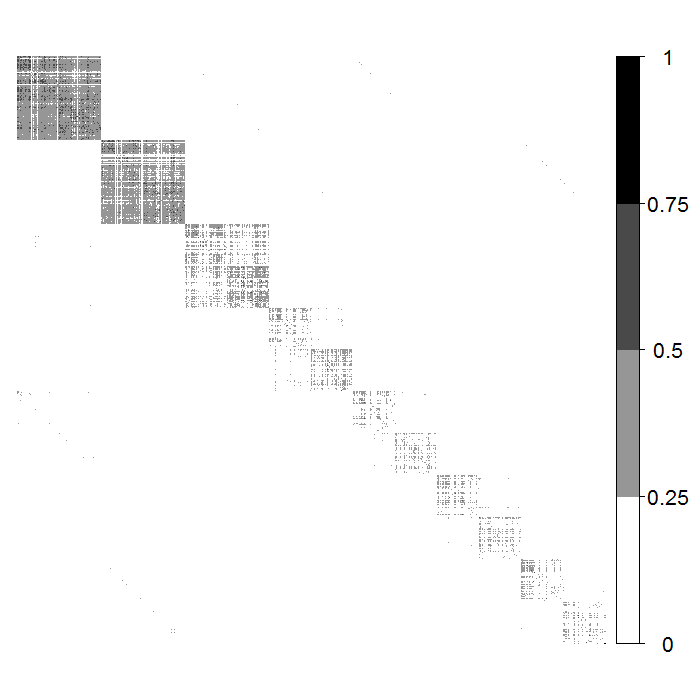}}%
\caption{Estimated correlation structures of $\mathbb{R}^{Lp}$-valued random coefficients from different $L$-truncated  Karhunen-Lo\`eve type expansions for the left-hand task. The figure shows the upper left 7 x 7 basis blocks of the absolute correlation matrix in basis-first order for: (a) functional principal component coefficients $(\xi_1^T, \dots, \xi_p^T)^T$ in \eqref{eq: basisExp} as in \cite{qiao:19}, and (b) random coefficients $(\theta_1\T,\ldots,\theta_L\T)\T$ under partial separability in \eqref{eq: PS_KL}.}
\label{fig:motorTask_blockcorrelations_lh}
\end{figure}

\clearpage

\bibliographystyle{biometrika}
\bibliography{FunGM_ParSep}

\end{document}